\documentclass[ba,preprint]{imsart}

\usepackage{graphicx}%
\usepackage{multirow}%
\usepackage{amsmath,amssymb,amsfonts,amsthm,hyperref}%
\hypersetup{hypertexnames=false}
\usepackage[title]{appendix}%
\usepackage{xcolor}%
\usepackage{textcomp}%
\usepackage{manyfoot}%
\usepackage{booktabs}%
\usepackage{float}
\usepackage{array}
\usepackage{algorithm}%
\usepackage{algorithmicx}%
\usepackage{algpseudocode}%
\usepackage{bm}%
\usepackage{microtype,enumitem}%
\usepackage[authoryear]{natbib}

\arxiv{2608.16573}
\firstpage{1}

\startlocaldefs
\newtheorem{theorem}{Theorem} 
\newtheorem{corollary}{Corollary} 
\newtheorem{lemma}{Lemma} 
\newtheorem{proposition}{Proposition} 

\newtheorem{assumption}{Assumption}

\newcommand{\E}{\mathbb{E}}
\newcommand{\Var}{\operatorname{Var}}
\newcommand{\cum}{\operatorname{cum}}
\newcommand{\TV}{\mathrm{TV}}
\newcommand{\esssup}{\operatorname*{ess\,sup}}
\newcommand{\cO}{{O}}

\newcommand{\dd}{\,\mathrm{d}}
\newcommand{\abs}[1]{\left\lvert #1\right\rvert}
\newcommand{\norm}[1]{\left\lVert #1\right\rVert}
\newcommand{\X}{X} %
\newcommand{\V}{X} %
\newcommand{\n}{n} %

\endlocaldefs

\begin{document}
\setcounter{secnumdepth}{3}

\begin{frontmatter}

\title{Bessel-Debiased Pseudo-Marginal MCMC for Generalised Bayesian Inference}
\runtitle{Bessel-debiased pseudo-marginal MCMC}

\begin{aug}
\author{\fnms{Yingkai} \snm{Lu}\thanksref{addr1}\ead[label=e1]{yingkailu@link.cuhk.edu.cn}},
\author{\fnms{Jeong Eun} \snm{Lee}\thanksref{addr2}\ead[label=e2]{kate.lee@auckland.ac.nz}},
\author{\fnms{Geoff K.} \snm{Nicholls}\thanksref{addr3}\ead[label=e3]{nicholls@stats.ox.ac.uk}},

\address[addr1]{
The Chinese University of Hong Kong, 
Shenzhen, China.
}
\address[addr3]{
Department of Statistics,
  University of Oxford,
  Oxford, UK.
}
\address[addr2]{
Department of Statistics,
  University of Auckland,
  Auckland, NZ.
}
\end{aug}

\begin{abstract}

Generalised Bayesian inference exponentiates a weighted loss, sometimes an integral estimated by simulation. Loss estimators generally give biased posterior weights which Pseudo-Marginal MCMC does not correct, and greater variability raises computational cost. We introduce Sign-Corrected Bessel Debiasing (SCBD), a signed pseudo-marginal method using independent block estimates from ordinary Monte Carlo or randomised quasi-Monte Carlo (RQMC). The correction is exact for i.i.d. Gaussian block estimates with unknown variance. For non-Gaussian blocks, we bound the posterior error in total variation. Under regularity conditions, the correction improves the fixed-data Monte Carlo bound from $O(M_n^{-1})$ to $O(M_n^{-2})$, where $M_n$ counts loss-integrand evaluations per loss estimate. Suppose the RQMC variance for a block of size \(B\) is uniformly \(O(B^{-\alpha}(\log B)^{d-1})\) in fixed dimension \(d\). With a fixed number of blocks and a learning rate proportional to the sample size \(n\), a per-proposal budget of order \(n^{2/\alpha}\), up to logarithmic factors, suffices to keep the leading log-weight variance bounded. Our examples show that favourable RQMC representations can inherit this budget scaling and that RQMC variance reduction complements Bessel correction. Once the block estimates have been formed, the Bessel correction comes at negligible additional cost. The correction is generic, straightforward to implement, and supported by theory.  

\end{abstract}


\end{frontmatter}

\section{Introduction}
\label{sec:intro}

Pseudo-marginal MCMC \citep{LinLiuSloan00,beaumont03,andrieu2009pseudo} is an exact MCMC method that uses nonnegative, unbiased estimates of an unnormalised target density in its acceptance probability. This approach works well for latent variable models, where importance sampling and particle filters \citep{andrieu2010particle} give unbiased likelihood estimators. However, in some settings we only have unbiased estimates of the log-likelihood and there is a need to de-bias the likelihood estimate. Russian-roulette methods \citep{lyne2015russian} can remove this bias at the cost of sign-reweighting (discussed below). However, they use random truncation which inflates weight-variance and adds computational cost. Our debiasing correction for unbiased log-likelihood estimators is exact in the Gaussian case and avoids this truncation. %

\citet{frazier2025exact} point out that noisy target evaluation is sometimes used in Generalised Bayesian Inference (GBI, \cite{bissiri2016general}, \cite{knoblauch22}), which updates belief from prior $\pi_0(\theta)$ to posterior $\pi_n(\theta)$ by exponentiating a loss
rather than a log-likelihood.  The target distribution is
\begin{equation}
    \pi_n(\theta)\propto \exp\{-\beta_n\ell_{n}(\theta)\}\pi_0(\theta),
    \label{eq:intro_gibbs_posterior}
\end{equation}
where \(\ell_n(\theta)\) is a parameter-dependent loss based on $n$ observations and \(\beta_n\) is an inverse-temperature or
learning-rate.   
In the
finite-data regimes considered below, the loss is normalized by \(1/n\), so \(\beta_n=O(n)\).   
Loss-based posteriors are used in robust or
misspecification-aware work \citep{grunwald2017inconsistency,jewson2018principles,matsubara2022robust}.

The computational problem studied here begins after the target in
Eq.~\eqref{eq:intro_gibbs_posterior} has been specified.  In simulator-based
GBI, \(\ell_{n}(\theta)\) may be an expected discrepancy, such as a maximum mean discrepancy (MMD) \citep{cherief2020mmd} or other
scoring-rule loss \citep{gao2023generalized,pacchiardi24}, where each evaluation may involve simulation and estimation.  In robust
generalized Bayes with tractable working densities, for example beta-type \citep{basu1998densitypower}, the loss may contain a model integral that is not
available in closed form and in tall-data settings in Bayesian inference, an average log-likelihood
contribution or average loss may be estimated by subsampling \citep{bardenet2017markov}.  

These examples
have different statistical motivations but the same obstructions. First of all, an
unbiased estimate of a loss is not an unbiased estimate of its exponential. Indeed, for an unbiased stochastic loss \(\bar\ell_{M}(\theta,\xi)\) based on estimation-variables $\xi$ with computational budget $M$ and having variance of order \(M^{-1}\), we have in general,
\begin{equation}
    \E_\xi\{\exp[-\beta_n\bar\ell_{M}(\theta,\xi)]\}
    >
    \exp[-\beta_n\ell_{n}(\theta)],
    \label{eq:intro_jensen_gap}
\end{equation}
provided the noise is non-degenerate and the expectation on the left is finite.  A pseudo-marginal MCMC
method needs an unbiased estimator of the posterior weight, not of its logarithm.  A parameter-dependent bias in the weight
changes the target; the construction below
addresses this bias by correcting the Gibbs weight formed from
the estimated loss.

The second obstruction is variance scaling.  Pseudo-marginal performance is sensitive to the
relative variability of the weight estimator \citep{andrieu2009pseudo,doucet2015efficient}.  If \(\beta_n\) is proportional to \(n\), the
leading log-weight variance is
\[
    \beta_n^2\operatorname{Var}\{\bar\ell_M(\theta,\xi)\}\asymp n^2/M.
\]
If this is not bounded over $n$ then the MCMC acceptance rate may go to zero \citep{doucet2015efficient}. Although neither necessary nor sufficient for geometric ergodicity \citep{andrieu2015convergence}, cost-efficiency analyses \citep{pitt2012some,sherlock2015efficiency} show that targeting a log-weight variance of order one supports MCMC mixing. Keeping this variance bounded requires a per-proposal budget $M$ of order $n^2$. %

In a recent paper, \citet{frazier2025exact} discuss pseudo-marginal MCMC (P-MCMC) for GBI with estimated losses. Under their assumptions, \(M\asymp n\) simulated pseudo-observations per loss estimate are needed in their MMD and beta-divergence examples to retain the usual posterior concentration behaviour. Their analysis then gives \(O(n^2)\) simulator calls per P-MCMC iteration as a requirement for concentration; the scaling requirement for good mixing is a separate condition which is verified empirically. They also give an exact Piecewise-Deterministic Markov Process (PDMP) construction \citep{fearnhead2018pdmp,bierkens2019zigzag}, using unbiased stochastic gradient estimates and tractable event-rate bounds. This works well but is not available for general targets. A user with a single data-set, and fixed $n$, may be interested in using a simple, generic but approximate P-MCMC algorithm, keeping the per-update budget as small as possible but still controlling bias, and that is what we provide.

We make three contributions. First, we introduce Sign-Corrected Bessel Debiasing (SCBD), which forms \(K\) independent block estimates of the loss and uses their across-block sample variance to construct a Bessel correction to the estimated Gibbs weight. The correction exactly removes the exponential bias under the Gaussian block model, with posterior expectations recovered by sign reweighting \citep{lyne2015russian} when the corrected weight is negative. Second, we considered non-Gaussian blocks, quantify the remaining finite-block target error and show, for ordinary Monte Carlo blocks, that we get a fixed-\(n\) approximation gain from \(O(M^{-1})\) to \(O(M^{-2})\). Third, we combine the correction with independently randomized quasi-Monte Carlo blocks \citep{l2000variance,owen1997b,loh2003asymptotic} to keep the leading log-weight variance bounded with a budget of order \(n^{2/\alpha}\), up to logarithmic factors, compared with order \(n^2\) for ordinary Monte Carlo. We call our corrected P-MCMC algorithms based on Monte Carlo (MC) and Randomized
Quasi Monte Carlo (RQMC) samples the MC-SCBD and RQ-SCBD algorithms respectively. Finally, we develop sign diagnostics and practical screening procedures to select and verify the finite-block configuration.

The paper is organized as follows. Section~\ref{sec:background} introduces the randomized block estimators of the loss used later. Section~\ref{sec:methodology} develops the SCBD construction, establishes the exact Gaussian Bessel identity, and gives the sign-corrected pseudo-marginal sampler. Section~\ref{sec:theory} studies the finite-block target induced by the corrected weight, derives convergence rates and computational-budget requirements for MC and RQMC blocks, and develops diagnostics and screening procedures for sign stability and block selection. A discussion of related work is given in Section~\ref{sec:related-work-bounds}. Section~\ref{sec:experiments} presents numerical experiments illustrating the target-error and variance reductions and comparing SCBD with relevant alternative methods. Section~\ref{sec:discussion} concludes with a discussion of the whole.

\section{Randomized Block Estimates of the Loss}
\label{sec:background}

Dropping $n$ subscripts (until we consider the $n$-dependence of the estimation budget), the loss \(\ell(\theta)\) in Eq.~\eqref{eq:intro_gibbs_posterior} is
estimated from \(K\) independently randomized blocks.  We assume that each block estimate is
unbiased and in particular,
\begin{equation}
    \ell(\theta)
    =
    \int_{\mathcal X} L(\theta,\xi)\,P(d\xi),
    \label{eq:loss_integrand_definition}
\end{equation}
where \(P\) is the reference law for a single input \(\xi\) used for
simulation, numerical integration, or subsampling.

For a block size \(B\), let
\(\xi_{k}=(\xi_{k,1},\ldots,\xi_{k,B})\)
denote the random input block used by the \(k\)th randomization.
The associated block estimate is
\begin{equation}
    \widehat\ell_B^{\;(k)}(\theta,\xi_{k})
    =
    \frac{1}{B}\sum_{b=1}^B L(\theta,\xi_{k,b}),
    \qquad k=1,\ldots,K.
    \label{eq:setup_block_estimators}
\end{equation}
The block rule is chosen so that, for all \(B\) and all
\(\theta\) for which \(L(\theta,\cdot)\) is integrable,
\begin{equation}
    \E\{\widehat\ell_B^{\;(k)}(\theta,\xi_{k})\}
    =
    \ell(\theta).
    \label{eq:setup_block_unbiasedness}
\end{equation}
For ordinary Monte Carlo, Eq.~\eqref{eq:setup_block_unbiasedness}
holds as the \(\xi_{k,b}\) are independent draws from \(P\).  The randomized stratification used in
Section~\ref{sec:crps_calibration} is also unbiased: it averages the
integrand at one uniformly sampled point from each of \(B\) equal
subintervals of \([0,1]\).

For RQMC, after absorbing any inverse transform or simulator map into \(L\),
we take \(\xi_{k,b}=u_{k,b}\in[0,1)^d\).  Here \(d\) is the dimension of
the randomized integration or simulator input.  Our Sobol blocks use a left linear matrix (LMS)
scramble followed by a digital random shift.  Each randomized point is
marginally uniform on \([0,1)^d\)
\citep[Proposition~3.1]{owen2003variance}; linearity of expectation therefore
gives Eq.~\eqref{eq:setup_block_unbiasedness}, although the points within a
block are dependent.  %

In both MC- and RQ-SCBD, the \(K\) blocks are independent across blocks.  Taking \(K\ge2\),
the per-proposal loss-integrand budget is $M=KB$,
and the block mean and block sample variance are
\begin{align}
    \bar\ell_M(\theta,\xi)
    &=
    \frac{1}{K}\sum_{k=1}^K
    \widehat\ell_B^{\;(k)}(\theta,\xi_{k}),
    \label{eq:setup_block_average}\\
    S_{K,B}^2(\theta,\xi)
    &=
    \frac{1}{K-1}\sum_{k=1}^K
    \left\{
        \widehat\ell_B^{\;(k)}(\theta,\xi_{k})
        -
        \bar\ell_M(\theta,\xi)
    \right\}^2.
    \label{eq:setup_block_sample_variance}
\end{align}
We suppress the argument \(\xi\) when the particular realization of the
auxiliary variables need not be displayed.
If one evaluation of
\(L(\theta,x)\) requires several simulator calls, \(M\) counts loss-integrand
evaluations rather than simulator calls.
A subscript \(M\) on a finite-block quantity stands for the full choice of
\(K\), \(B\), and block rule at fixed \(n\) and \(\beta\). We use this notation for simplicity, but the reader should note that two configurations with the same numerical value of
\(M\) can define different finite-block quantities. 

\section{The Sign-Corrected Bessel Debiasing Construction}
\label{sec:methodology}

The construction uses \(K\) independent, finite-\(B\) unbiased block
estimates of \(\ell(\theta)\).  Their average remains unbiased, but
its exponential is not.  SCBD adds a Bessel factor computed
from the across-block sample variance which removes this bias exactly under a Gaussian block model.  Because it can be negative,
the Markov chain uses the absolute corrected weight and posterior averages are
sign-reweighted.

\subsection{Gaussian \texorpdfstring{\(K\)}{K}-block correction}
\label{sec:gaussian_bessel_correction}

\subsubsection{Exact Gaussian identity}

Fix \(\theta\) and \(B\), and suppress them and $M$ from the notation.  The
Gaussian block model below is an idealization.  For ordinary Monte
Carlo blocks it is motivated by the central limit theorem; for
scrambled-net blocks, an analogous limit requires additional
regularity conditions \citep{loh2003asymptotic}.  In this subsection we assume that%
\begin{equation}
    \widehat\ell^{(1)},\ldots,\widehat\ell^{(K)}
    \stackrel{\rm iid}{\sim}
    N(\ell,\sigma^2),
    \qquad K\ge2,
    \label{eq:gaussian_block_model}
\end{equation}
where \(\ell=\ell(\theta)\) with expectation taken under Eq.~\eqref{eq:gaussian_block_model}. Let
\begin{equation}
    \bar\ell
    =
    \frac1K\sum_{k=1}^K\widehat\ell^{\;(k)},
    \qquad
    S_K^2
    =
    \frac1{K-1}\sum_{k=1}^K
    \{\widehat\ell^{\;(k)}-\bar\ell\}^2.
    \label{eq:gaussian_block_mean_variance}
\end{equation}
so that \(\bar\ell\sim N(\ell,\sigma^2/K)\),
\((K-1)S_K^2/\sigma^2\sim\chi^2_{K-1}\), \(\bar\ell\) is independent of
\(S_K^2\) and
\begin{equation}
    \E\{\exp(-\beta\,\bar\ell)\}
    =
    \exp(-\beta\,\ell)
    \exp\left(\frac{\beta^2\sigma^2}{2K}\right).
    \label{eq:gaussian_exponential_inflation}
\end{equation}
The first factor is the intended Gibbs weight; the second is the inflation due
to Gaussian block noise.  We seek an unbiased estimator of the reciprocal of the second factor. Because \(\bar\ell\) and \(S_K^2\) are independent, it is fine for the estimator to depend on \(S_K^2\).

For \(s^2\ge0\), let
\begin{equation}\label{eq:bessel-arguments-nu-z}
    \nu=\frac{K-3}{2},
    \qquad
    z_\beta(s^2,K)
    =
    \beta\,\sqrt{\frac{(K-1)s^2}{K}},
\end{equation}
and define
\begin{equation}
    F_\beta(s^2,K)
    \equiv
    \begin{cases}
    \displaystyle
    \Gamma(\nu+1)
    \left\{\frac{z_\beta(s^2,K)}{2}\right\}^{-\nu}
    J_\nu\!\left\{z_\beta(s^2,K)\right\},
        & z_\beta(s^2,K)>0,\\[1.2ex]
    1,  & z_\beta(s^2,K)=0.
    \end{cases}
    \label{eq:bessel_correction}
\end{equation}
where \(J_\nu\) is the Bessel function of the first kind, and the second
line is the continuous extension at \(z_\beta(s^2,K)=0\).  
The whole expression for $F_\beta$ in the first line of \eqref{eq:bessel_correction} is actually the generalized hypergeometric function ${}_0F_1\!\left((K-1)/{2},-{z_\beta^2}/{4}\right)$.%

The function $F_{\beta}$ above appears in \citet{ceperley1999penalty} but is used in a different way. They sought to correct a penalty-method algorithm rather than a pseudo-marginal algorithm. The penalty method is a randomised MCMC algorithm \citep{nicholls2012coupled} which refreshes both numerator and denominator factors at an MCMC update like the Monte Carlo within Metropolis (MCWM) algorithm in \cite{andrieu2009pseudo}. The correction is not exact (even in the Gaussian case) when used in that way\footnote
{The Penalty Method and the Pseudo-Marginal Method were published around the same time. Indeed, \citet{ceperley1999penalty} writes of \cite{LinLiuSloan00}, ``A recent preprint proposed to solve the problem of violating the constraints on the acceptance probabilities by introducing negative signs into the estimators. We have not explored this possibility''.}.

Proposition~\ref{prop:exact_debiasing}
gives the identy we need: under the Gaussian
\(K\)-block model, the resulting signed weight has expectation
\(\exp(-\beta\,\ell)\). The identity separates the roles of the two block summaries: \(\bar\ell\)
enters the Gibbs weight, while \(S_K^2\) corrects the Jensen bias.
\begin{proposition}[Gaussian \(K\)-block Bessel identity]
\label{prop:exact_debiasing}
Under Eq.~\eqref{eq:gaussian_block_model}, for every \(\beta\,\ge0\),
\begin{equation}
    \E\{F_\beta(S_K^2,K)\}
    =
    \exp\left(-\frac{\beta^2\sigma^2}{2K}\right),
    \label{eq:gaussian_block_unbiasedness}
\end{equation}
and hence
\begin{equation}
    \E\{\exp(-\beta\,\bar\ell)F_\beta(S_K^2,K)\}
    =
    \exp(-\beta\,\ell).
    \label{eq:signed_estimator}
\end{equation}
\end{proposition}

\begin{proof}
Set \(r=(K-1)/2=\nu+1\) and $Y=\frac{(K-1)S_K^2}{\sigma^2}\sim\chi^2_{2r}$. Using the power series for \(J_{r-1}\) and
\(z^2=\beta^2\sigma^2Y/K\),
\begin{equation}
    F_\beta(S_K^2,K)
    =
    \Gamma(r)
    \sum_{j=0}^{\infty}
    \frac{(-1)^j}{j!\,\Gamma(r+j)}
    \left(\frac{\beta^2\sigma^2Y}{4K}\right)^j.
    \label{eq:bessel_series_expansion}
\end{equation}
Because
\(
\E(Y^j)=2^j\Gamma(r+j)/\Gamma(r)
\), the expectation of the corresponding absolute series is
\(\exp\{\beta^2\sigma^2/(2K)\}\).  Termwise expectation is therefore valid,
and
\[
    \E\{F_\beta(S_K^2,K)\}
    =
    \sum_{j=0}^{\infty}
    \frac1{j!}
    \left(-\frac{\beta^2\sigma^2}{2K}\right)^j
    =
    \exp\left(-\frac{\beta^2\sigma^2}{2K}\right).
\]
Equation~\eqref{eq:signed_estimator} follows from the independence of
\(\bar\ell\) and \(S_K^2\), together with
Eq.~\eqref{eq:gaussian_exponential_inflation}.
\end{proof}

\subsubsection{Sign of the correction}

SCBD uses \(F_\beta(S_K^2,K)\) as a factor in a signed random weight.
For \(z_\beta>0\), every factor in Eq.~\eqref{eq:bessel_correction} other than
\(J_\nu(z_\beta)\) is positive.  If \(j_{\nu,m}\) denotes the \(m\)th positive
zero of \(J_\nu\), then
\[
    F_\beta(S_K^2,K)>0
    \quad\text{for}\quad
    0\le z_\beta<j_{\nu,1},
\]
and its sign changes as \(z_\beta\) crosses successive zeros. When $S_K^2$ is small, $z_\beta(S_K^2,K)$ in \eqref{eq:bessel-arguments-nu-z} is small and
\begin{equation}\label{eq:F-expansion-small-z}
    F_\beta(S_K^2,K)
    =
    1-\frac{z_\beta^2}{2(K-1)}+O(z_\beta^4),
\end{equation}
so the correction is close to one for small \(S_K^2\) as we might expect.  Near a positive
Bessel zero, \(|F_\beta|\) is close to zero; beyond that zero the
correction may be negative. In the next section we give the standard importance-sampling correction, which is exact. However, MCMC efficiency goes down with increasing negative-sign probability. We discuss this in Appendix~\ref{app:negative_sign_probability}. 
If negative or
near-zero factors occur too frequently in the MCMC, \(B\) and/or \(K\) can be increased and the run restarted. We give a screening rule for $K$ and $B$ in Section~\ref{sec:block_configuration_screen} and Appendix~\ref{app:pilot_algorithms}.

\subsection{SCBD pseudo-marginal sampler}
\label{sec:SCBD_sampler}

The Sign-Corrected Bessel Debiasing (SCBD) algorithm is a variant of P-MCMC given in Algorithm~\ref{alg:SCBD}. We now consider its target distribution. Let \(p_\theta(d\xi)\) denote the joint law of
\(\xi=(\xi_1,\ldots,\xi_K)\).  For \(M=KB\), define the signed random weight
\begin{equation}
    \widehat Z_M(\theta,\xi)
    =
    \exp\{-\beta\,\bar\ell_M(\theta,\xi)\}
    F_{\beta}\!\left(S_{K,B}^2(\theta,\xi),K\right).
    \label{eq:SCBD_weight_estimator}
\end{equation}
Under the Gaussian block model, Proposition~\ref{prop:exact_debiasing} gives
\begin{equation}
    \E_\xi\{\widehat Z_M(\theta,\xi)\}
    =
    \exp\{-\beta\,\ell(\theta)\}.
    \label{eq:gaussian_SCBD_weight_unbiasedness}
\end{equation}
Equation~\eqref{eq:SCBD_weight_estimator} leads to MC-SCBD when the blocks are
ordinary Monte Carlo and RQ-SCBD when they are independently randomized Sobol
blocks. 

The Gaussian model is used to establish
Eq.~\eqref{eq:gaussian_SCBD_weight_unbiasedness}; it is not assumed by the
sampler. When the blocks are Gaussian, the marginal distribution of the absolute-weight chain is generally not $\pi_n$, but sign correction recovers $\pi_n$ exactly. When the blocks are not Gaussian, the sign-corrected samples target an approximation to $\pi_n$. Because \(\widehat Z_M\) may be negative, the Metropolis--Hastings transition
uses its absolute value. If the average absolute weight is strictly positive and finite,
\begin{equation}
    0<
    C_M^{\rm abs}
    \equiv
    \int
    \E_\xi\{|\widehat Z_M(\theta,\xi)|\}
    \pi_0(\theta)\,d\theta
    <\infty,
    \label{eq:absolute_normalizing_constant}
\end{equation}
then the transition leaves invariant the extended distribution
\begin{equation}
    \widetilde\Pi_M(d\theta,d\xi)
    =
    \frac{
    |\widehat Z_M(\theta,\xi)|\pi_0(\theta)p_\theta(d\xi)
    }{C_M^{\rm abs}},
    \label{eq:absolute_extended_target}
\end{equation}
with \(\theta\)-marginal
\begin{equation}
    \widetilde\pi_M(\theta)
    =
    \frac{
    \E_\xi\{|\widehat Z_M(\theta,\xi)|\}\pi_0(\theta)
    }{C_M^{\rm abs}}.
    \label{eq:absolute_theta_marginal}
\end{equation}
Thus unweighted \(\theta\)-draws follow \(\widetilde\pi_M\), which need not be the intended generalized posterior.  

The recorded signs correct
posterior averages \citep{lyne2015russian}.
Let
\(
    s(\theta,\xi)=\operatorname{sign}\{\widehat Z_M(\theta,\xi)\}
\).
For any \(h\) satisfying
\[
    \int |h(\theta)|
    \E_\xi\{|\widehat Z_M(\theta,\xi)|\}
    \pi_0(\theta)\,d\theta<\infty,
\]
the identity \(s\cdot |\widehat Z_M|=\widehat Z_M\) gives
\begin{equation}
    \frac{
    \E_{\widetilde\Pi_M}\{s(\theta,\xi)h(\theta)\}
    }{
    \E_{\widetilde\Pi_M}\{s(\theta,\xi)\}
    }
    =
    \frac{
    \int h(\theta)\E_\xi\{\widehat Z_M(\theta,\xi)\}
          \pi_0(\theta)\,d\theta
    }{
    \int \E_\xi\{\widehat Z_M(\theta,\xi)\}
          \pi_0(\theta)\,d\theta
    },
    \label{eq:generic_signed_ratio_target}
\end{equation}
if the denominator
is non-zero.
Under Gaussian blocks as in Eq.~\eqref{eq:gaussian_block_model}, the right-hand side is
\(\E_{\pi_n}\{h(\theta)\}\).  For a general finite-block law it has a bias which depends on $M$.

\begin{algorithm}[htbp]
\caption{SCBD pseudo-marginal sampler}
\label{alg:SCBD}
\begin{algorithmic}[1]
\State \textbf{Input:} \(\beta\), \(K\), \(B\), prior \(\pi_0\), proposal
\(q\), initial state \(\theta_0\), and a block rule generating
\(\xi\sim p_\theta\)
\State Draw \(\xi_0\sim p_{\theta_0}\) and compute
\(\widehat Z_0=\widehat Z_M(\theta_0,\xi_0)\)
\State If \(\widehat Z_0=0\), redraw \(\xi_0\); if no non-zero value is found
after a prespecified number of attempts, stop and revise \((K,B)\)
\For{\(t=1\) to \(T\)}
    \State Propose \(\theta'\sim q(\cdot\mid\theta_{t-1})\) and draw
    \(\xi'\sim p_{\theta'}\), independently of the current auxiliary variable
    \State Compute \(\widehat Z'=\widehat Z_M(\theta',\xi')\)
    \If{\(\widehat Z'=0\)}
        \State Set \(A=0\)
    \Else
        \State Set
        \[
            A=\min\left\{1,
            \frac{|\widehat Z'|\pi_0(\theta')q(\theta_{t-1}\mid\theta')}
            {|\widehat Z_{t-1}|\pi_0(\theta_{t-1})q(\theta'\mid\theta_{t-1})}
            \right\}.
        \]
    \EndIf
    \State With probability \(A\), set
    \((\theta_t,\xi_t,\widehat Z_t)=(\theta',\xi',\widehat Z')\); otherwise
    retain \((\theta_{t-1},\xi_{t-1},\widehat Z_{t-1})\)
    \State Record \(s_t=\operatorname{sign}(\widehat Z_t)\)
\EndFor
\State \textbf{Output:} \(\{(\theta_t,s_t)\}_{t=1}^T\)
\end{algorithmic}
\end{algorithm}

The empirical version of
Eq.~\eqref{eq:generic_signed_ratio_target} is
\begin{equation}
    \widehat I_T(h)
    =
    \frac{\sum_{t=1}^T s_t h(\theta_t)}
    {\sum_{t=1}^T s_t},
    \label{eq:signed_ratio_estimator}
\end{equation}
provided the denominator is non-zero.  If the extended chain is Harris ergodic
with invariant distribution \(\widetilde\Pi_M\),
\(\E_{\widetilde\Pi_M}|h(\theta)|<\infty\), and
\(\E_{\widetilde\Pi_M}s(\theta,\xi)\ne0\), then \(\widehat I_T(h)\) converges almost surely to the
right-hand side of Eq.~\eqref{eq:generic_signed_ratio_target}.  %

\section{Accuracy, Variance Budgets, and Sign Diagnostics}
\label{sec:theory}

At a
finite MC or RQMC block size, three further questions arise.  The expectation
of the corrected weight may alter the sign-reweighted posterior, the variance
of the estimated loss determines the required per-proposal budget, and negative
weights may reduce precision through sign cancellation.  We address these
questions and describe how a fixed \((K,B)\) configuration is
selected and checked before the production run.

\subsection{Target induced by the corrected weight}
\label{sec:finite_block_signed_target}

Throughout this section, \(\pi_n\) denotes the normalized posterior in
Eq.~\eqref{eq:intro_gibbs_posterior}. Define the multiplicative bias for the estimator $\widehat Z_M(\theta,\xi)$ against its target $\exp\{-\beta\,\ell(\theta)\}$,
\begin{align}
    w_M(\theta)
    &\equiv
    \exp\{\beta\,\ell(\theta)\}\,{\E_\xi\{\widehat Z_M(\theta,\xi)\}}
    \nonumber\\
    &= 
    \E_\xi\left\{\exp\left(-\beta[\bar\ell_M(\theta,\xi)-\ell(\theta)]\right)
    F_{\beta}\!\left(S_{K,B}^2(\theta,\xi),K\right)\right\}
    \label{eq:wM_definition}
\end{align}
Under the Gaussian block model,
Proposition~\ref{prop:exact_debiasing} gives \(w_M\equiv1\).  For a general block law, \(w_M\) is the remaining multiplicative bias in the mean
corrected weight, so it can be negative, and that
motivates the following assumption.

\begin{assumption}[Positive mean corrected weight]
\label{ass:finite_block_target}
For the fixed block configuration under consideration,
\(w_M(\theta)>0\) for \(\pi_n\)-almost every \(\theta\).
\end{assumption}

This condition always holds in the Gaussian block model.  For general
finite blocks it is an assumption on the mean corrected weight;
individual corrected weights may still be negative.  Together with the
absolute integrability condition in
Eq.~\eqref{eq:absolute_normalizing_constant}, it implies
\(0<\E_{\pi_n}w_M<\infty\).  Substituting
Eq.~\eqref{eq:wM_definition} into
Eq.~\eqref{eq:generic_signed_ratio_target} shows sign reweighting targets
\begin{equation}
    \pi_{n,M}(\theta)
    \equiv
    \frac{\pi_n(\theta)w_M(\theta)}
    {\E_{\pi_n}w_M}.
    \label{eq:signed_finite_block_target}
\end{equation}
Define the normalized variation
\begin{equation}
    \bar w_M(\theta)
    \equiv
    \frac{w_M(\theta)}{\E_{\pi_n}w_M}
    =
    \frac{d\pi_{n,M}}{d\pi_n}(\theta).
    \label{eq:posterior_normalized_factor}
\end{equation}
Values above or below one identify regions to which \(\pi_{n,M}\) assigns more
or less mass than \(\pi_n\).  In particular, \(\bar w_M=1\) \(\pi_n\)-almost everywhere if and only if
the two posterior targets $\pi_{n,M}$ and $\pi_n$ coincide.

For comparison, omitting the Bessel factor gives the positive uncorrected
weight
\[
    \widehat Z_M^{(0)}(\theta,\xi)
    =
    \exp\{-\beta\,\bar\ell_M(\theta,\xi)\}.
\]
Its multiplicative bias is
\begin{equation} 
    u_M(\theta) \equiv
    e^{\beta\,\ell(\theta)}
    \E_\xi\{\widehat Z_M^{(0)}(\theta,\xi)\} =
    \E_\xi
    \exp\!\left[
      -\beta
      \{\bar\ell_M(\theta,\xi)-\ell(\theta)\}
    \right].
    \label{eq:uncorrected_target_factor}
\end{equation}
Whenever this expectation is finite,
\(u_M(\theta)\ge1\) by Jensen's inequality.  If \(u_M\) is
integrable under \(\pi_n\), the uncorrected finite-block
posterior and its normalized variation are
\begin{equation}
    \pi_{n,M}^{(0)}(\theta)
    =
    \frac{
      \pi_n(\theta)u_M(\theta)
    }{
      \E_{\pi_n}u_M
    },
    \qquad
    \bar u_M(\theta)
    =
    \frac{
      u_M(\theta)
    }{
      \E_{\pi_n}u_M
    }.
    \label{eq:uncorrected_finite_block_target}
\end{equation}

\subsection{Budget controlling convergence}
\label{sec:tv-rate-gain-debias}

Theory and experiments below show gains from debiasing. The efficiency of RQMC gives additional gains which are visible in our experiments. The variance of $\widehat\ell_B^{\;(k)}$ in MC blocks based on $B$ samples $\xi_{k}=(\xi_{k,1},\dots,\xi_{k,B})$ is $O(B^{-1})$. \citep{owen1997b} shows that RQMC with \(d=\dim(\xi_{k,b})\)-dimensional Sobol blocks can achieve
\begin{equation}\label{eq:RQ-var-bound}
    \operatorname{Var}\{\widehat\ell_B^{\;(k)}(\theta,\xi_{k})\}
    \le
    C B^{-\alpha}\{\log(2+B)\}^{d-1}.
\end{equation}
This form includes the nested-uniformly scrambled digital-net variance bound
with \(\alpha=3\) for sufficiently smooth integrands
\citep{owen1997b}. This behaviour is a property of the target loss and its randomization; it is not
implied by the use of Sobol points alone
\citep{l2000variance,owen1997b,owen2003variance,loh2003asymptotic}. %

In this section we give convergence rates for uncorrected (P-MCMC) and corrected (SCBD) posteriors with both ordinary MC and RQ block errors, so four cases in all. 
Fix an integer $K\geq2$.  We restore the $n$ subscript and write $\beta_n,\ell_n$, $u_{n,M}$, and $w_{n,M}$
to show dependence on the number of observations $n$ in the target; the block sizes may scale with $n$ so $B=B_n$ below; any additional $n$-dependence of the block-error law, as in Section~\ref{sec:crps_calibration},
is suppressed. For a block size $B$, the block-errors
\begin{equation}\label{eq:block-error}
  \X_{k}(\theta)\equiv\widehat\ell^{\;(k)}_{n,B}(\theta,\xi_k)
  -\ell_n(\theta),
  \qquad k=1,\ldots,K,
\end{equation}
are iid and centered given $\theta$. 
The $B$ and $n$ dependence of the distribution of the block-errors controls convergence.
To derive the posterior convergence rates, we expand the cumulant generating function (cgf) of the centred block errors $X_k(\theta)$. These expansions quantify the distortion in the expected weights and total variation bounds below.
    
\begin{theorem}[MC and RQMC rates]
\label{cor:method-rates} 
Let \(K\ge2\) be fixed, \(B_n\to\infty\), and \(M_n=KB_n\). Suppose Assumption \ref{ass:unified-analytic} holds. The uncorrected and Bessel-corrected posteriors converge to the target posterior in total variation at rates of the following orders.
\begin{enumerate}[label=(\alph*)]
\item \textbf{MC-blocks:}\\ 
If 
\({\beta_n^2}/{M_n}\longrightarrow0\), then 
\(  \left\|
    \pi^{(0)}_{n,M_n}-\pi_n
  \right\|_{\mathrm{TV}}
  =
  O\!\left(
    \dfrac{\beta_n^2}{M_n}
  \right).\)

If
$
  {\beta_n^3}/{M_n^2}\longrightarrow0,
$
then $
  \left\|
    \pi_{n,M_n}-\pi_n
  \right\|_{\mathrm{TV}}
  =O\!\left(
    \dfrac{\beta_n^3}{M_n^2}
  \right).$ \\
When $\beta_n\asymp n$, $M_n/n^2\!\to\! \infty$ is sufficient for convergence of the uncorrected posterior and $M_n/n^{3/2}\to\infty$ in the debiased case.
\item \textbf{RQMC blocks:} Let
$\displaystyle
  \norm{g}_{\infty,n}
  \equiv\esssup_{\theta\sim\pi_n}\abs{g(\theta)}
$ and $\delta_{n,M_n}^{\mathrm{RQ}}
=\beta_n^2M_n^{-\alpha}{\log(2+M_n)}^{d-1}$ and assume that
\( |Var(\widehat{\ell}_{B_n}^{(k)})|_{\infty,n}
\le C B_n^{-\alpha}{\log(2+B_n)}^{d-1},\) for $C<\infty$, $\alpha>0$,
with fixed $d$. If $\delta^{\mathrm{RQ}}_{n,M_n}\to0$,
\begin{align}
  \left\|
    \pi^{(0)}_{n,M_n}-\pi_n
  \right\|_{\mathrm{TV}}
  &=
  O\!\left(\delta^{\mathrm{RQ}}_{n,M_n}\right), \\
  \left\|
    \pi_{n,M_n}-\pi_n
  \right\|_{\mathrm{TV}}
  &=
  O\!\left\{
    \left(\delta^{\mathrm{RQ}}_{n,M_n}\right)^{3/2}
  \right\} 
\end{align}
When $\beta_n\asymp n$, the sufficient condition
$\delta^{\mathrm{RQ}}_{n,M_n}\to0$ is equivalently
\[
  \frac{M_n}
  {n^{2/\alpha}
   \{\log(2+M_n)\}^{(d-1)/\alpha}}
  \longrightarrow\infty.
\]
\end{enumerate}
\end{theorem}

\begin{proof}
See Appendix~\ref{app:convergence}. Assumption \ref{ass:unified-analytic} is an analyticity assumption for the cgf of the block-error in \eqref{eq:block-error}. Lemma~\ref{lem:cumulant-envelope} uses this to get a Cauchy bound on the remainder in the cummulant expansion of this cgf. The multiplicative biases $u_{n,M}$ and $w_{n,M}$ in \eqref{eq:uncorrected_target_factor} and \eqref{eq:wM_definition} can be written as functions of the cgf (see \eqref{eq:u-as-psi} and \eqref{eq:w-spherical-cgf}) so Lemma \ref{lem:unified-weight-expansion} uses the cgf expansions to expand $u_{n,M}$ and $w_{n,M}$. Corollary \ref{cor:unified-TV} converts these into total-variation bounds on $\pi^{(0)}_{n,M_n}$ and $\pi_{n,M}$ via \eqref{eq:signed_finite_block_target} and \eqref{eq:uncorrected_finite_block_target}. In Appendix~\ref{app:convergence-TV-bounds-main} we substitute MC and RQMC variance and cumulant scales into the Lemma~\ref{lem:unified-weight-expansion}-bounds to get the rates in Theorem~\ref{cor:method-rates}. 
\end{proof}
\begin{table}[ht]
\small
\setlength{\tabcolsep}{4pt}
\renewcommand{\arraystretch}{1.35}
\begin{tabular}{@{}lcccc@{}} 
\toprule
&
\shortstack{P-MCMC\\TV rate}
&
\shortstack{SCBD\\TV rate}
&
\shortstack{P-MCMC:\\sufficiency}
&
\shortstack{SCBD:\\sufficiency}
\\
\midrule
MC-blocks
&
$\beta_n^2/M_n$
&
$\beta_n^3/M_n^2$
&
$M_n/\beta_n^2\to\infty$
&
$M_n/\beta_n^{3/2}\to\infty$
\\[1ex]
RQ-blocks
&
$\delta^{\mathrm{RQ}}_{n,M_n}$
&
$(\delta^{\mathrm{RQ}}_{n,M_n})^{3/2}$
&
$\delta^{\mathrm{RQ}}_{n,M_n}\to 0$
&
$\delta^{\mathrm{RQ}}_{n,M_n}\to0$
\\
\bottomrule
\end{tabular}\caption{Convergence rate comparison between MC and RQMC blocks.}\label{table_rate}
\end{table}
Table \ref{table_rate} shows that the Bessel correction improves the target-convergence rate in both settings.  For MC-blocks it also weakens the sufficient
budget condition for convergence from
$M_n/\beta_n^2\to\infty$ to
$M_n/\beta_n^{3/2}\to\infty$.
Under the RQMC variance rate assumption in \eqref{eq:RQ-var-bound}, with $\delta_n^{\mathrm{RQ}}$ defined in Theorem~\ref{cor:method-rates}(b), the TV discrepancies are
$O\{\delta^{\mathrm{RQ}}_{n,M_n}\}$ and
$O\{(\delta^{\mathrm{RQ}}_{n,M_n})^{3/2}\}$ for the uncorrected and
corrected targets. However, 
$\delta^{\mathrm{RQ}}_{n,M_n}\to 0$ is sufficient in both cases as $\delta^{3/2}\to 0$ if and only if $\delta\to 0$.

These posterior convergence results don't guarantee good mixing.  In MC-SCBD the corrected target converges for $n^{3/2}\ll M_n\ll n^2$, even though the leading loss-induced log-weight variance, of order $n^2/M_n$, diverges.  
From the large-$n$ perspective of \cite{frazier2025exact}, the results above are generally negative. The MC-SCBD sufficient condition
$M_n/n^{3/2}\to\infty$ is stronger than the order-$n$
criterion derived by those authors, because we have a different approximation-criterion.  However, for RQMC,
when $\alpha>2$, the sequence $M_n$ can be sublinear,
up to logarithmic factors. Also, from the perspective of someone wanting to use a simple MCMC algorithm to get a good posterior approximation for a given data set with fixed or bounded budget, it is the $M$ dependence not the $n$ dependence that maters, and the MC bounds give
approximation orders $M^{-1}$ without correction and $M^{-2}$ with correction (at fixed $n$) and if RQMC is available, greater gains are possible.
The finite-budget examples are consistent with this theory.

\subsubsection{Related work on convergence and P-MCMC}\label{sec:related-work-bounds}
Previous authors give results for convergence of the uncorrected target defined in \eqref{eq:uncorrected_finite_block_target}.
Theorem~1 of 
\citet[first arXiv version]{frazier2024impact} gives fixed-$n$ convergence in weighted $L^1$ distances, including a bound of order $1/M$. \cite{lie2018random} study marginal posteriors
formed from expected random likelihoods in Hellinger distance. Their bounds gives more general random approximations, but don't make the order-$M_n^{-1}$ rate visible. 
\citet{sprungk2020local} proves local Lipschitz
dependence of posterior measures on perturbations of the Gibbs-potential
in several %
distance, which could also give a fixed-$n$ bound of order $1/M$.
Under stronger assumptions of unbiased loss estimation and analyticity, our result for $\pi^{(0)}_{n,M_n}$ in \eqref{eq:MC-u-TV-a-app-full} identifies the leading posterior-centred variance term and its coefficient while tracking the dependence on
$\beta_n$ and $M_n$. The other results in Theory~\ref{cor:method-rates} treat new debiased and RQ-block settings.

\citet{quiroz2019speeding} take inspiration from \cite{ceperley1999penalty} and are a closer precedent. They combine a control-variate difference estimator of an
additive log likelihood with a positive Gaussian variance correction, obtaining fixed-$n$ posterior error of order $M^{-2}$. Under Bernstein von Mises concentration and further regularity conditions, including an accurate full-data MLE or mode as an expansion point, their error is $O(1/(nM^2))$. Our Bessel correction instead removes the unknown-variance Gaussian inflation exactly and extends to general MC or
RQMC loss estimators, at the cost of signed weights. Their results are specialised to the data-subsampling regime and extending them to our setting faces practical obstacles in constructing control variates. Although the Gaussian corrections are alternatives, their variance-reduction approaches are complementary: control variates can be used within SCBD blocks, RQMC can reduce residual variance, and correlated pseudo-marginal updates may improve mixing without changing the finite-block target.

Exact signed estimators based on Poisson randomisation provide a
second close comparison. \citet{quiroz21blockpoisson} construct an
unbiased block-Poisson likelihood estimator from unbiased
log-likelihood estimates.  The estimator is a product of possibly negative Poisson
factors, so like SCBD, it uses an absolute-weight extended target and a sign reweighting. Unlike SCBD, which is exact for i.i.d. Gaussian blocks and otherwise targets an approximate finite-block posterior, block-Poisson is unbiased for the intended likelihood for general finite-variance base estimators, making sign reweighting exact in the simulation-consistency sense.

The blocks also have different roles.  Our $K$ independent loss
blocks provide the mean and sample variance entering the Bessel
correction.  The block-Poisson product factors are principally used
to correlate successive likelihood estimates by updating only part
of the auxiliary state; the product blocking does not itself reduce
the estimator variance.  Block-Poisson has random computational cost
and requires tuning of its Poisson factors, centring or soft lower
bound, and subsample size.  SCBD instead has the fixed budget
$M=KB$ and requires no lower-bound parameter, but pays for this
simplicity through finite-block target error outside the Gaussian
case.

\citet{yang25} apply the block-Poisson estimator to doubly
intractable models and give a finite-chain result for the sign
denominator.  Their sign-stability theorem applies once $R_M$ in \eqref{eq:global_sign_def} is
bounded away from zero and the absolute-weight chain mixes.  It
does not establish that condition from a finite pilot, nor controls the discrepancy between $\pi_{n,M}$ and $\pi_n$; these are
the separate roles of our sign screen and finite-block target audit. Independently randomised RQMC
estimates could also serve as the unbiased base estimates
within block-Poisson, although the Bessel and Poisson factors
are alternative exponential corrections. 

\subsection{Budget controlling log-weight variance}
\label{sec:rqmc_sublinear}

In this and the next two subsections we consider how $K$ and $B$ are chosen to promote mixing and approximation accuracy. Although our screening rules are supported by the theory in this section, they are ultimately heuristics. We begin with the budget required to control the variance of the loss-estimate. For a non-zero corrected weight,
\[
    \log\left|\widehat Z_M(\theta,\xi)\right|
    =
    -\beta_n\bar\ell_M(\theta,\xi)
    +
    \log\left|
        F_{\beta_n}\!\left(
            S_{K,B}^2(\theta,\xi),K
        \right)
    \right|.
\]
For budget selection we control the variance contributed by the first term,
\begin{equation}
    \tau_M^2(\theta)
    \equiv
    \beta_n^2
    \operatorname{Var}_{\xi}
    \left\{\bar\ell_M(\theta,\xi)\right\}.
    \label{eq:stability_budget_condition}
\end{equation}
The budget calculation focuses on this loss-induced term. However, it gives some control over the second term as well.
Appendix~\ref{app:small_noise_log_variance} gives the Gaussian expansion for fixed \(K\) as \(\tau_M^2(\theta)\to0\). 
The Bessel-term contribution is therefore second order in the small-noise
Gaussian regime.  It need not be small at a finite operating point, especially
near a Bessel zero.  The pilot screen below records
\(\tau_M^2\) and evaluates zero frequencies, negative
weights, and weighted signs using the full corrected weight.

Let \(\mathcal C_n\) denote a region of parameter space containing the
posterior concentration region, over which \(\tau_M^2(\theta)\) is required
to remain uniformly bounded.  The following result converts an MC or RQMC
variance rate into a sufficient per-proposal budget.  

\begin{theorem}[Budget orders for the loss-induced log-weight variance]
\label{thm:stability_budgets}
Fix \(0<C_\tau<\infty\). Assume that \(c_\beta n\le\beta_n\le C_\beta n\), that \(K\) is fixed or
belongs to a fixed finite set, and that all constants below are uniform in
\(n\) and \(\theta\in\mathcal C_n\).
\begin{enumerate}
    \item If MC-blocks satisfy \[c_\ell M^{-1}
        \le
        \operatorname{Var}\{\bar\ell_M(\theta,\xi)\}
        \le
        C_\ell M^{-1},\] then a uniform bound
    \(\sup_{\theta\in\mathcal C_n}\tau_M^2(\theta)\le C_\tau\)
    requires \(M\ge c_0n^2\) and is achieved by \(M=C_0n^2\) for sufficiently
    large \(C_0\).

    \item Suppose that, for the chosen \(d\)-dimensional RQMC
representation and some constants \(C<\infty\) and \(\alpha>0\), the estimator variance satisfies 
\[
    \operatorname{Var}\{\widehat\ell_B^{\;(k)}(\theta,\xi_{k})\}
    \le
    C B^{-\alpha}\{\log(2+B)\}^{d-1}.
\]
This is the bound in Eq.~\eqref{eq:RQ-var-bound}; see \citet{owen1997b}. Then, for sufficiently large \(C_0\),
    \begin{equation}
        M_n
        =
        C_0 n^{2/\alpha}
        \{\log(2+n)\}^{(d-1)/\alpha}
        \label{eq:rqmc_budget_choice}
    \end{equation}
    satisfies the same uniform bound.  Equivalently, a sufficient RQMC budget
    has order \(\widetilde O(n^{2/\alpha})\), where \(\widetilde O\) suppresses
    logarithmic factors.
\end{enumerate}
\end{theorem}

The proof is given in Appendix~\ref{app:variance_budget_proof}.

The RQ-gain over MC-blocks occurs when \(\alpha>1\): the
budget keeping \(\tau_M^2\) bounded is then subquadratic in \(n\),
and is sublinear, up to logarithmic factors, when \(\alpha>2\). Under the Gaussian block model, Corollary~\ref{cor:operational_threshold} shows that, for \(q_0\in(0,1/2)\), taking \(C_\tau=\tau_{\rm sign}^2(K,q_0)\) ensures \(q_M^-(\theta)\le q_0\) whenever \(\tau_M^2(\theta)\le C_\tau\). Outside that special case, the sign screen in Section~\ref{sec:block_configuration_screen} is used.

\subsection{Sign cancellation}
\label{sec:sign_stability}

In this section we define the quantities that determine the budget requirements imposed by sign-stability. Negative weights matter through both their frequency and their magnitude.  At a
fixed \(\theta\), define
\begin{equation}
    q_M^-(\theta)
    \equiv
    \mathbb P_\xi\{\widehat Z_M(\theta,\xi)<0\}.
    \label{eq:negative_probability_def_intro}
\end{equation}
When zero weights have probability zero, \(1-2q_M^-(\theta)\) is the
expected sign under a fresh block randomization.  If the block law has atoms,
as in finite-population lookup or subsampling, zero weights can have positive
probability. With \(\operatorname{sign}(0)=0\),
\[
    \E_\xi[\operatorname{sign}\{\widehat Z_M(\theta,\xi)\}]
    =
    1-2q_M^-(\theta)
    -\mathbb P_\xi\{\widehat Z_M(\theta,\xi)=0\}.
\]
The zero-weight frequency is recorded separately.
Magnitude is captured by the local weighted expected sign
\begin{equation}
    r_M(\theta)
    \equiv
    \frac{\E_\xi\{\widehat Z_M(\theta,\xi)\}}
    {\E_\xi\{|\widehat Z_M(\theta,\xi)|\}}.
    \label{eq:weighted_sign_def_intro}
\end{equation}
Under Assumption~\ref{ass:finite_block_target},
\(0<r_M(\theta)\le1\) whenever the denominator is finite.  It equals one when
the corrected weight is nonnegative almost surely, and approaches zero when
positive and negative absolute weight nearly cancel.  Averaging over the
absolute-value chain gives
\begin{equation}
    R_M
    \equiv
    \E_{\widetilde\Pi_M}
    [\operatorname{sign}\{\widehat Z_M(\theta,\xi)\}]
    =
    \int r_M(\theta)\widetilde\pi_M(\theta)\,d\theta.
    \label{eq:global_sign_def}
\end{equation}
The empirical mean of the signs in the recorded chain estimates \(R_M\).  This is the
population denominator in Eq.~\eqref{eq:generic_signed_ratio_target}, so values
near zero make sign reweighting unstable.

A finite-chain result is given by \cite{yang25} for
signed block pseudo-marginal MCMC.  Under stationarity,
reversibility, a positive spectral gap, and $R_M\neq0$, their
concentration argument implies that, for every
$0<c<|R_M|$,
$
  \Pr\left(
    \left|
      \sum_{t=1}^T s_t
    \right|>cT
  \right)
  \rightarrow1
$
exponentially fast.  This result controls the denominator once the
global mean sign is separated from zero.%

A small \(q_M^-(\theta)\) does not by itself guarantee a large
\(r_M(\theta)\), since rare negative weights may have large magnitude.  By
Cauchy--Schwarz,
\begin{equation}
    r_M(\theta)
    \ge
    \left[
    \frac{\E_\xi\{\widehat Z_M(\theta,\xi)^2\}}
    {\E_\xi\{\widehat Z_M(\theta,\xi)\}^2}
    \right]^{-1/2},
    \label{eq:weighted_sign_second_moment_bound}
\end{equation}
whenever the second moment is finite.  Consequently, if the displayed
relative second moment is at most \(C<\infty\), then
\(r_M(\theta)\ge C^{-1/2}\).  In Eq.~\eqref{eq:weighted_sign_def_intro}, positive and negative
contributions cancel in the numerator, whereas their magnitudes add
in the denominator. A smaller relative second moment gives a stronger
lower bound on \(r_M(\theta)\).

The same quantity relates the marginal distribution of the raw chain to the
sign-reweighted target.  From
Eqs.~\eqref{eq:absolute_theta_marginal}, \eqref{eq:wM_definition},
\eqref{eq:posterior_normalized_factor}, and
\eqref{eq:weighted_sign_def_intro},
\begin{equation}
    \widetilde\pi_M(\theta)
    \propto
    \pi_n(\theta)\frac{\bar w_M(\theta)}{r_M(\theta)}
    =
    \frac{\pi_{n,M}(\theta)}{r_M(\theta)}.
    \label{eq:absolute_value_marginal_relation}
\end{equation}
Thus \(\bar w_M\) is the tilt from \(\pi_n\) to the sign-corrected
finite-block target \(\pi_{n,M}\), while \(1/r_M\) is the additional
tilt from \(\pi_{n,M}\) to the absolute-weight marginal.  Repeated block
randomizations at a fixed state estimate
\(r_M\) without knowing \(\ell(\theta)\), whereas estimating
\(\bar w_M\) requires a reference loss.  This distinction motivates the
pilot screen below.

\subsection{Pilot screening of block configurations}
\label{sec:block_configuration_screen}

Theorem~\ref{thm:stability_budgets} determines the scale on which block
sizes should be considered, while the quantities in
Section~\ref{sec:sign_stability} determine whether a proposed $(K,B)$-configuration is numerically usable.  Before the production run, we fix
an ordered candidate set
\[
    \mathcal C=\{(K_1,B_1),\ldots,(K_L,B_L)\}
\]
and a finite collection of posterior-representative pilot states
\(\mathcal O_{\rm pilot}\).  These states may be obtained from preliminary runs and are fixed before candidate budgets are screened.

Before examining the candidates, choose
a negative-weight tolerance
\(q_0\in(0,1/2)\).  For example, \(q_0=0.05\) targets at least
\(95\%\) positive corrected weights at each pilot state when no zero weights
occur.

For each candidate \((K,B)\in\mathcal C\), replicate corrected weights
are generated at each state in \(\mathcal O_{\rm pilot}\).  Independence
across pilot states is not required: a prespecified common-random-number
coupling may be used, provided that each state has the production-estimator
marginal law.  These draws
estimate \(q_M^-(\theta)\), the zero-weight
frequency, and \(r_M(\theta)\); the corresponding estimate of \(\tau_M^2(\theta)\) is also recorded. For an optional cap \(C_\tau(K)\in(0,\infty]\), a candidate is retained when
\(\widehat\tau^2_{K,B}(\theta)\le C_\tau(K)\), \(q_M^-(\theta)\le q_0\), no zero
weights are observed, and the estimated weighted sign is positive at every
pilot state. Setting \(C_\tau(K)=\infty\) gives a sign-only screen. This procedure, which we call the pilot screen, is given in Appendix~\ref{app:pilot_algorithms}. The Gaussian first-zero calculation in
Appendix~\ref{app:negative_sign_probability} may guide the candidate range,
but the pilot screen itself estimates \(q_M^-(\theta)\) and \(\tau_M^2(\theta)\) using the production-run estimators and no Gaussian assumption.

The pilot screen is reference-free.  Without a reliable reference loss, one
may repeat the analysis at larger block configurations to assess the
sensitivity of posterior summaries; agreement is a stability check for the
selected finite-block targets, not a total-variation guarantee.  When a (presumably small number of) reliable
loss evaluations are available on a finite design, the optional audit in
Appendix~\ref{sec:finite_block_audit_main} estimates the normalized variation
of \(w_M\) and refines the final selection. We call this the audit screen below.

\section{Numerical Experiments}
\label{sec:experiments}

The examples are distinguished by properties of the loss and the distribution of its estimator as summarized in Table~\ref{tab:experiment_overview} at the start of Appendix~\ref{app:experimental_details}. %
Section~\ref{sec:crps_calibration} uses a continuous ranked probability score (CRPS) available in closed form to assess the variance rate, the selected per-proposal budget, and the finite-block target
check against an exact loss.  Section~\ref{sec:random_phase_simulator}
keeps the integrated loss fixed while increasing the variability of a pathwise
gradient, allowing comparison with the exact stochastic-gradient
Zig--Zag construction of \citet{frazier2025exact}.  Section~\ref{sec:g_and_k_benchmark} considers an empirical MMD
posterior for the \(g\)-and-\(k\) model, comparing RQ-SCBD and 
MC-SCBD with an iid-MC Russian-roulette estimator for the
same intended target. %
We do not include block-Poisson P-MCMC \citep{quiroz21blockpoisson}, because a competitive implementation would require its own control-variate and centring choices, soft lower bound, Poisson-factor tuning, correlated auxiliary update, and random-cost accounting.

Finally, Appendix~\ref{sec:real_ftir} applies the density-power score or $\beta$-loss of \cite{basu1998densitypower} to Fourier Transform Infrared (FTIR) spectra.  Its stochastic loss is a discrete lookup, so the main question is whether RQ-SCBD reduces exponentiation-induced target distortion at a fixed total point budget, rather than improves the RQMC rate.

Appendix~\ref{app:experimental_details} gives the common screening tolerances and budget conventions.

\subsection{CRPS Calibration}
\label{sec:crps_calibration}

The continuous ranked probability score is a proper scoring rule for univariate
predictive distributions and admits an expectation representation in terms of
absolute differences \citep{gneiting2007strictly}.  We use a one-parameter
family so this score is available analytically and as a
non-Gaussian stochastic integral. We estimate all quantities of interest directly from their definitions without running an MCMC chain. We run pre-screening Algorithms~\ref{alg:adaptive_tuning} and \ref{alg:finite_block_audit} and the
configuration passing both checks is
evaluated directly on a one-dimensional grid.

Let \(\theta\in(0,1)\) be the predictive mean and define an observation model
\begin{equation}
    Y_\theta=U^{\rho(\theta)},\qquad
    U\sim{\rm Unif}(0,1),\qquad
    \rho(\theta)=\frac{1-\theta}{\theta}.
    \label{eq:crps_power_model}
\end{equation}
Then \(\E(Y_\theta)=\theta\).  The true mean is $\theta^*=1/5$ so $\rho^*=4$ and we suppose the data,
\begin{equation}
    y_i=\left(\frac{i-1/2}{n}\right)^4,\qquad i=1,\ldots,n,
    \label{eq:crps_calibration_panel}
\end{equation}
are the midpoint quantiles of \(Y=U^4\), \(U\sim{\rm Unif}(0,1)\). This removes a layer of
data-generation noise and leaves a controlled target for studying the $(K,B)$-block selection rule and the dependence of the block-estimator variance and $\pi_{n,M}$ on $n$ and $M=KB$.

We suppose the generalized posterior $\pi_n$ in \eqref{eq:intro_gibbs_posterior} uses a uniform prior on \(\theta\in (0,1)\) and the average CRPS loss
\begin{equation}
    \ell_n(\theta)
    =
    \frac1n\sum_{i=1}^n
    \left[
      \E|Y_\theta-y_i|
      -\frac12\E|Y_\theta-Y_\theta'|
    \right],
    \qquad \beta_n=n,
    \label{eq:crps_loss}
\end{equation}
where \(Y_\theta\) and \(Y_\theta'\) are independent realisations of \eqref{eq:crps_power_model}.  Writing
\(\rho=\rho(\theta)\),
\begin{align}
    \E|Y_\theta-y|
    &=
    \frac1{\rho+1}-y
    +\frac{2\rho}{\rho+1}y^{1+1/\rho}
    \label{eq:crps_abs_exact}\\
    \intertext{and}
    \frac12\E|Y_\theta-Y_\theta'|
    &=
    \frac{\rho}{(\rho+1)(\rho+2)}.
    \label{eq:crps_pair_exact}
\end{align}
Thus \(\ell_n\) and the intended posterior are available in closed form.
The stochastic representation, with $u$ playing the role of $\xi_{k,b}$ in \eqref{eq:setup_block_estimators}, is
\begin{equation}
    L(\theta,u)
    =
    \frac1n\sum_{i=1}^n |u^{\rho(\theta)}-y_i|
    -\frac{\rho(\theta)}
    {\{\rho(\theta)+1\}\{\rho(\theta)+2\}},
    \qquad
    \E\{L(\theta,U)\}=\ell_n(\theta).
    \label{eq:crps_integrand}
\end{equation}
An MC block averages iid uniforms.  For $k=1,\dots,K$, the $k$'th RQ-block uses independently
randomized one-dimensional strata,
\[
    u_{k,b}=\frac{b-1+V_b}{B},\qquad
    V_b\stackrel{\rm iid}{\sim}{\rm Unif}(0,1),\quad b=1,\ldots,B,
\]
with independent randomizations across blocks.  The integrand is continuous and
piecewise smooth in one dimension, so randomized stratification
is particularly effective.

In order to estimate $\alpha$ in \eqref{eq:RQ-var-bound} we fit $cB^{-\alpha}$ to $\Var\{\widehat\ell_B^{\;(k)}(\theta)\}$ over
\(B=16,\ldots,1024\), evaluated at the posterior mode. Figure~\ref{fig:crps_rate_cost}(a) shows this block-estimator variance against $B$. The fit gives
\[
    \widehat\alpha_{\rm MC}=1.014
    \quad (95\%~{\rm CI}: 0.999,1.028),
    \qquad
    \widehat\alpha_{\rm RQMC}=2.997
    \quad (95\%~{\rm CI}: 2.976,3.018).
\]
We repeat the $(K,B)$-selection procedure at
\(n=100,200,\ldots,6400\) using Algorithms~\ref{alg:adaptive_tuning} and \ref{alg:finite_block_audit}. The selected budgets are shown in Figure~\ref{fig:crps_rate_cost} in panel (b). Under the sign and TV-audit thresholds \(q_0=0.05\) and $\epsilon_w=0.01$, the fitted slope of
\(\log M\) against \(\log n\), over \(n=200,\ldots,6400\), is
\(0.612\) with 95\% interval \([0.543,0.681]\) for RQMC.  This interval contains the predicted sublinear $n$-dependence \(2/\widehat\alpha_{\rm RQMC}=0.667\).  The cheap Gaussian MC sign-screen calculation in Corollary~\ref{cor:operational_threshold} has slope \(1.903\), close to the value \(2/\widehat\alpha_{\rm MC}=1.972\) predicted by the MC variance rate, consistent with the expected near-quadratic cost growth; this planning benchmark does not receive the baseline audit in Algorithm~\ref{alg:finite_block_audit}.

Panel (c) gives a direct finite-\(B\) check of the claimed $M=KB$ dependence set out in
Corollary~\ref{cor:method-rates}, fixing \(n=800\) and \(K=5\).  Over the
filled-symbol ranges, the fitted log--log slopes are \(-0.99\) and \(-2.01\)
for P-MCMC-MC and MC-SCBD, and \(-2.99\) and \(-6.27\) for P-MCMC-RQ and
RQ-SCBD.  The MC slopes are consistent with the uncorrected \(B^{-1}\) and corrected \(B^{-2}\) orders in part 1 of the corollary; the
uncorrected RQ slope is consistent with \(B^{-\widehat\alpha_{\rm RQMC}}\simeq B^{-3}\). For RQ-SCBD, the generic and negligible-third-cumulant orders are \(B^{-3\widehat\alpha_{\rm RQMC}/2}\simeq B^{-4.5}\) and \(B^{-2\widehat\alpha_{\rm RQMC}}\simeq B^{-6}\), respectively. The fitted slope is compatible with the latter over this range, indicating that the third-cumulant contribution is small there. Open symbols in panel (c) mark pre-asymptotic budgets and are excluded
from the fits.

\begin{figure}[ht!]
    \centering
    \includegraphics[width=\textwidth, clip, trim={0 4 0 4}]{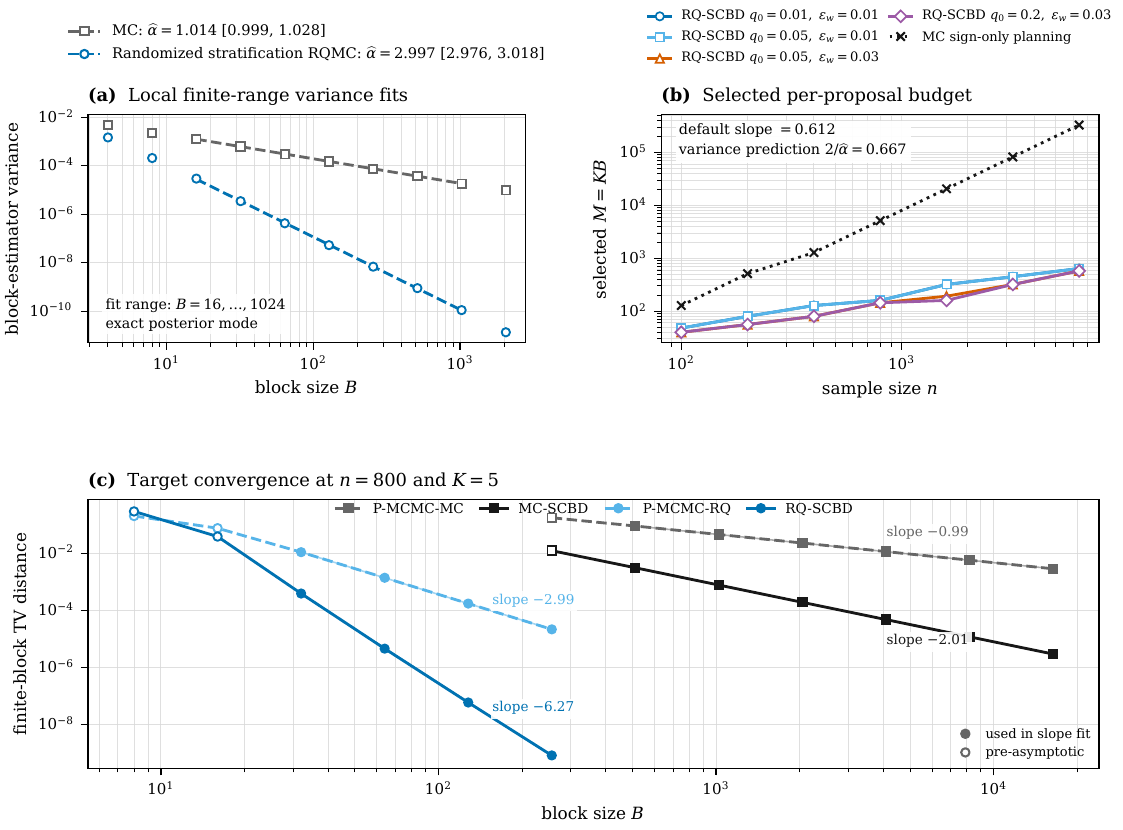}
    \caption{CRPS block-estimator variance, selected per-proposal budget, and target convergence.  Panel (a) displays and reports local variance-fits at the posterior mode.  Panel (b)
    plots and reports the value \(M=KB\) selected at each \(n\); this will become the cost of one SCBD MCMC proposal.  The
    solid RQ-SCBD curves use the full sign and target checks in Algorithms~\ref{alg:adaptive_tuning} and \ref{alg:finite_block_audit}.  The MC curve is a
    Gaussian sign-screen planning calculation.  Panel (c) reports finite-block TV
    distance for the corrected and uncorrected MC and RQ constructions; the
    P-MCMC curves omit the Bessel correction.  Filled symbols define the fitted
    ranges and open symbols mark pre-asymptotic budgets.}
    \label{fig:crps_rate_cost}
\end{figure}

At \(n=800\), the sign screen admits \((K,B,M)=(9,16,144)\), but its
Monte Carlo audit estimate is \(0.0255\), so Algorithm~\ref{alg:finite_block_audit}
continues down the ordered feasible list and selects \((5,32,160)\).  Numerical integration gives TV values \(0.00765\) and \(0.000386\), respectively,
so the Monte Carlo audit is conservative in this case.

The threshold-sensitivity results in Table \ref{tab:crps_tolerance} and Figure \ref{fig:crps_threshold_effect} in Appendix \ref{app:CRPS} show a complementary trade-off: relaxing the audit tolerance reduces the selected budget but increases posterior distortion.
As predicted by Theorem \ref{thm:stability_budgets}, fixed log weight or sign threshold give budget scaling 
\(B\asymp n^{2/\alpha}\); changing the threshold affects the multiplicative constant, not the exponent. The TV-audit criterion is separate and can be checked here because the exact CRPS target is available.

\subsection{Random-Phase Simulation and Pathwise Zig--Zag Cost}
\label{sec:random_phase_simulator}

Random-phase spectral representations \citep{shinozuka1991simulation} allow a smooth integrated loss to have highly variable pathwise derivatives. We compare RQ-SCBD, which uses only forward loss evaluations, with the exact stochastic-gradient Zig–Zag construction of \citet{frazier2025exact}.

\subsubsection{Model, target, and pathwise gradient}
\label{sec:random_phase_model}

Let \(y_1,\ldots,y_n\in\mathbb R^{d_y}\) be observed vectors and write \(\overline y_n\) for their sample mean, with
\(d_y=10\).  For \(\theta\in\mathbb R^2\) and $i=1,\dots,d_y$, fixed vectors
\(a_i,c_i\in\mathbb R^2\), fixed  amplitude $\epsilon>0$ and frequency \(\kappa\) and independent phases
\(\phi_i\sim{\rm Unif}(0,1)\), define the simulator
output for one component of $Y=(Y_1,\dots,Y_{d_y})$ as
\begin{equation}
    Y_i(\theta,\phi_i)
    =
    a_i^\top\theta
    +
    \epsilon\,w_\eta(\kappa c_i^\top\theta+2\pi\phi_i),
    \qquad i=1,\ldots,d_y,
    \label{eq:random_phase_output}
\end{equation}

Here \(w_\eta\) is a smooth, \(2\pi\)-periodic waveform, normalized to have mean zero and variance one under a uniform phase. Its definition (\ref{eq:random_phase_waveform}) and normalization are given in the Appendix \ref{random_phase2_model}.

For a fixed weighting matrix \(W\), specified in (\ref{eq:W}), the loss is 
\begin{equation}
    L(\theta,\phi)=
    \frac{1}{2d_y}
    \left\|W\{Y(\theta,\phi)-\overline y_n\}\right\|^2,
    \qquad W\in\mathbb R^{d_y\times d_y},
    \label{eq:random_phase_loss}
\end{equation}
Let $\mathsf A$ and $\mathsf C$ have rows $a_i^\top$ and $c_i^\top$,
respectively. The loss, $\ell(\theta)=\E(L(\theta,\phi))$, is calculated in Appendix \ref{random_phase2_model},\\[-2ex]
\begin{equation}\label{eq:exact-rand-phase-loss}
  \ell(\theta)
  =
  \frac1{2d_y}\|W(\mathsf A\theta-\overline y_n)\|^2
  +
  \frac{\epsilon^2}{2d_y}{\rm tr}(W^\top W).
\end{equation}
With prior $\pi(\theta)\propto 1$, the Gibbs posterior density for $\theta$ is Gaussian with covariance
\[
    \Sigma_{\rm ref}
    =
    \{\beta_n\mathsf A^\top W^\top W\mathsf A/d_y\}^{-1}.
\]
The closed form and gradient are not available to either sampler. We let $\kappa$ range over values favouring Zig--Zag ($\kappa=12$) and RQ-SCBD ($\kappa=60$). The parameter settings are
\[
    n=\beta_n=20000,\qquad \epsilon=0.04,\qquad \eta=0.8,
    \qquad \kappa\in\{12,20,60\}.
\]

The exact Zig--Zag sampler in \citet{frazier2025exact} requires an unbiased gradient estimator and a tractable bound
valid for every auxiliary draw. Holding \(\phi\) fixed and differentiating the simulated path gives
\begin{equation}
    \nabla_\theta L(\theta,\phi)
    =
    \frac1{d_y}J_Y(\theta,\phi)^\top W^\top W
    \{Y(\theta,\phi)-\overline y_n\},
    \qquad
    J_Y(\theta,\phi)
    =
    \mathsf A+\epsilon\kappa D_\eta(\theta,\phi)\mathsf C,
    \label{eq:random_phase_pathwise_grad}
\end{equation}
where \(D_\eta\) is diagonal with entries
\(w_\eta'\{\kappa c_i^\top\theta+2\pi\phi_i\}\). 
The exact target with loss $\ell(\theta)$ in \eqref{eq:exact-rand-phase-loss} is independent of $\kappa$, whereas the
pathwise Jacobian contains the term
$\epsilon\kappa D_\eta(\theta,\phi)\mathsf C$; increasing $\kappa$
increases gradient oscillation without changing the
phase-averaged loss. We use this to stress-test the Zig--Zag sampler.
Differentiation and expectation may be interchanged in this bounded smooth construction:
\[
    \E_\phi\{\nabla_\theta L(\theta,\phi)\}
    =
    \nabla_\theta\,\ell(\theta),
\]
so gradient estimates are unbiased.
The Zig--Zag implementation in Appendix~\ref{app:random_phase_validation} draws a fresh \(\phi\) at every candidate event and
uses the rigorous auxiliary-independent affine upper bound in
Eq.~\eqref{eq:random_phase_affine_bound}.  Combined with the unbiased pathwise
gradient, this thinning construction leaves \(\pi_n\) invariant in the
\(\theta\)-marginal.

Appendix~\ref{app:random_phase_validation} defines \(\mathcal R_{\rm ZZ}\), the ratio of the stationary mean switching rates with stochastic and exact gradients.

\begin{figure}[htbp]
    \centering
    \includegraphics[width=\textwidth, clip, trim={0 2 0 4}]{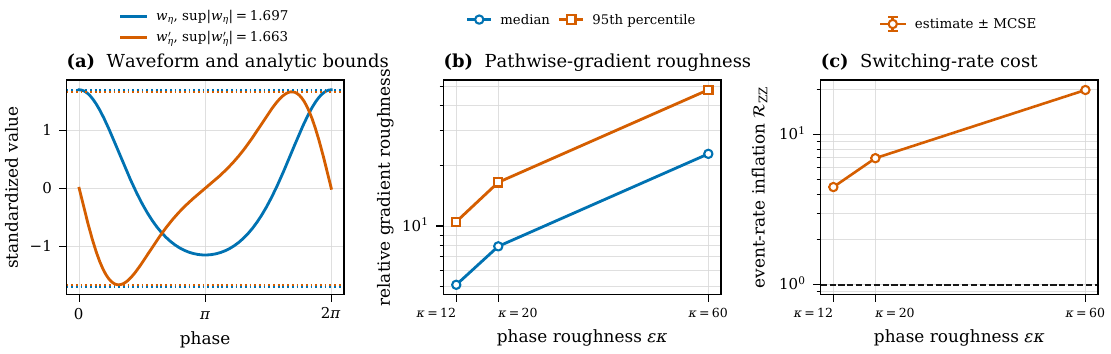}
    \caption{Random-phase mechanism.  Panel (a) verifies the analytic bounds on
    the standardized waveform $w_\eta$ in \eqref{eq:random_phase_waveform} and its derivative.  Panel (b) shows the
    increasing roughness of the pathwise gradient while the phase-averaged loss
    remains unchanged.  Panel (c) shows the switching-rate
    inflation \(\mathcal R_{\rm ZZ}\), from \(4.47\) at \(\kappa=12\) to
    \(19.75\) at \(\kappa=60\).}
    \label{fig:random_phase_mechanism}
\end{figure}

\subsubsection{Block selection and sampler comparison}
\label{sec:random_phase_budget}

Focusing on the block configuration at $\kappa=60$, we choose \(K=9, B=16, M=144\) in Appendix~\ref{sec:random_phase2}, using the sign-screen and baseline audit in Algorithms~\ref{alg:adaptive_tuning} and \ref{alg:finite_block_audit}. The estimated variance exponents are
\[
    \widehat\alpha_{\rm MC}=0.98,
    \qquad
    \widehat\alpha_{\rm RQMC}=3.01.
\]
These numbers and the scaling plot in Figure~\ref{fig:random_phase_sign_scaling} show behaviour similar to the scaling we saw in Section~\ref{sec:crps_calibration}: faster block-variance decay for RQMC than MC, and the need for the target audit in Algorithm~\ref{alg:finite_block_audit} beyond sign screening.

We use the same RQ-SCBD block configuration at all three \(\kappa\)-values and count one forward loss evaluation for RQ-SCBD or one fresh pathwise-gradient evaluation for Zig–Zag as a stochastic evaluation. This favours Zig–Zag when derivatives cost more than scalar losses. Each method receives \(400000\) stochastic evaluations in each regime.
RQ-SCBD consequently completes \(2777\) whole proposals
(\(399888\) loss evaluations), while Zig--Zag processes \(400000\) candidate
events.  For chain \(c\), mean
error is
\(
\left\{\frac12\sum_{j=1}^2
\left(\frac{\widehat\mu_{c,j}-\mu_{{\rm ref},j}}
{\sigma_{{\rm ref},j}}\right)^2\right\}^{1/2},
\)
using the signed mean for RQ-SCBD and the ordinary mean for Zig--Zag.   
Within each chain, ESS is
the minimum coordinate ordinary-chain ESS after burn-in, normalized by the 
realized production-evaluation count including burn-in.
Figure~\ref{fig:random_phase_efficiency} plots the mean error and ESS-statistics averaged over five independent chains. Cumulative ESS is plotted against stochastic evaluations in
Figure~\ref{fig:random_phase_posterior} (bottom).  
Zig--Zag has a modest
evaluation-normalized advantage at \(\kappa=12\).  The ordering reverses at
\(\kappa=20\) and becomes pronounced at \(\kappa=60\), where
\(\mathcal R_{\rm ZZ}\) is largest.

\begin{figure}[ht!]
    \centering
    \includegraphics[width=\textwidth, clip, trim={0 5 0 5}]{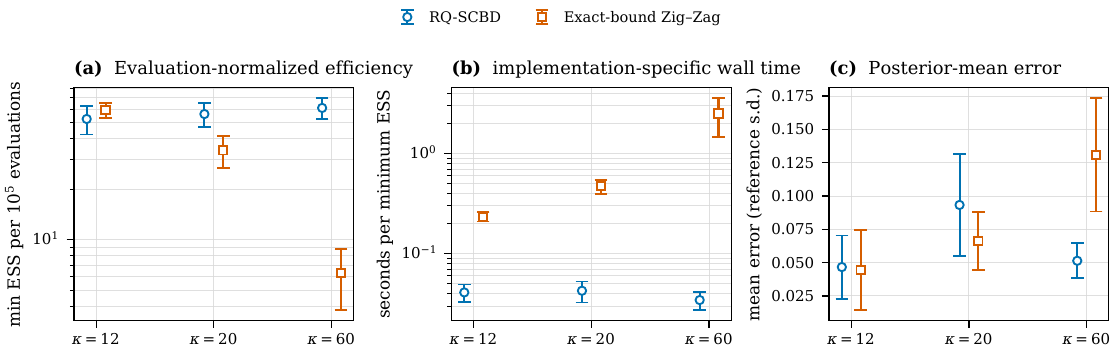}
    \caption{Five-seed random-phase comparison for Section~\ref{sec:random_phase_budget} under \(400000\) stochastic
    evaluations plotted against frequency parameter \(\kappa\).
    Effective sample size (ESS) and mean error are defined at the end of Section~\ref{sec:random_phase_budget}; they are shown with one sample standard deviation over
    five chains.  Panel (a) gives minimum
    ESS per \(10^5\) stochastic evaluations.  Panel (b) reports wall time per
    effective draw and is implementation-specific.  Panel (c) reports
    posterior-mean error in reference standard-deviation units.  The loss-only
    method becomes preferable as pathwise-gradient switching-rate inflation
    increases.}
    \label{fig:random_phase_efficiency}
\end{figure}

The posteriors for both samplers agree with the Gaussian reference in Figure~\ref{fig:random_phase_posterior} (top), although Zig–Zag targets \(\pi_n\) exactly and RQ-SCBD targets the audited finite-block posterior \(\pi_{n,M}\). Across fifteen Zig–Zag chains, no rate-bound violations or clipped probabilities occurred (maximum observed rate-to-bound ratio: 0.418).

\subsection{Likelihood-Free Inference for the
\texorpdfstring{\(g\)-and-\(k\)}{g-and-k} Distribution}
\label{sec:g_and_k_benchmark}

The \(g\)-and-\(k\) model is easy to simulate but has an inconvenient likelihood \citep{drovandi2011likelihoodfree}. 
For a common empirical maximum mean discrepancy (MMD) posterior, we compare RQ-SCBD, MC-SCBD, differing only in their block generators, with iid-MC
Russian-roulette (RR) estimator, which provides 
unbiased signed estimates of the posterior weight from iid MMD-pair
evaluations \citep{lyne2015russian}.  MMD-ABC is given separately
in Appendix~\ref{app:gandk_abc_sweep}, because its tolerance changes the target posterior. 

\subsubsection{Target and stochastic representation}

The observation model is defined through the quantile function
\begin{equation}
    Q_\theta(z)
    =
    A_{\rm GK}+B_{\rm GK}
    \left[
      1+c_{\rm GK}
      \frac{1-\exp(-gz)}{1+\exp(-gz)}
    \right]
    (1+z^2)^k z,
    \qquad z\sim N(0,1),
    \label{eq:gk_quantile}
\end{equation}
where \(A_{\rm GK}=3\), \(B_{\rm GK}=1\), and \(c_{\rm GK}=0.8\).
The associated likelihood is $p_\theta(\cdot)$ with $\theta=(g,k)\in[0,5]\times[0,2]$, and a uniform prior.  The observed data $y$ has size
\(\n=500\) and was generated at
\(\theta_0=(0.10,0.05)\).  All calculations below condition on this
single realized dataset.
In order to define the loss, let
\[
    k_h(x,x')
    =
    \exp\!\left\{-\frac{(x-x')^2}{2h^2}\right\},
    \qquad
    m_y(x)
    =
    \frac1{\n}
    \sum_{i=1}^{\n} k_h(x,y_i),
\]
where \(h=3.949\), four times the median non-zero pairwise distance among
the first \(300\) observations. After replacing the data distribution by the empirical measure and dropping the data-only term, the empirical MMD loss is
\begin{equation}
    \ell_{\rm MMD}(\theta)
    =
    \E_{p_\theta} k_h(X,X')
    -
    2\,\E_{p_\theta} m_y(X),
    \qquad
    X,X'\overset{\rm iid}{\sim}p_\theta.
    \label{eq:gk_mmd_loss}
\end{equation}
This is an empirical loss as it is based on finite data $y$. The exact generalized posterior is
\begin{equation}
    \pi_{\rm MMD}(\theta\mid y)
    \propto
    \exp\{-\beta_n\ell_{\rm MMD}(\theta)\}\pi_0(\theta),
    \qquad
    \beta_n=\lambda \n,
    \qquad
    \lambda=5.
    \label{eq:gk_mmd_posterior}
\end{equation}
The bandwidth and learning-rate multiplier are fixed parts of this
target.

Figure~\ref{fig:gandk_target_geometry} shows the finite-data loss and
the corresponding numerical posterior.  The numerical loss-minimizer is $\widehat\theta_{\rm MMD}=(0.044,0.012)$, whereas the numerical posterior mean is $\widehat\mu_{\rm ref}=\E_{\pi_n}(\theta)=(0.164,0.085)$. These need not coincide with each other or with the generating value:
the minimizer depends on the realized finite-data loss, and the posterior
mean also reflects the prior and the finite value of \(\lambda\).

\begin{figure}[htbp]
    \centering
    \includegraphics[width=\textwidth, clip, trim={0 6 0 12}]
    {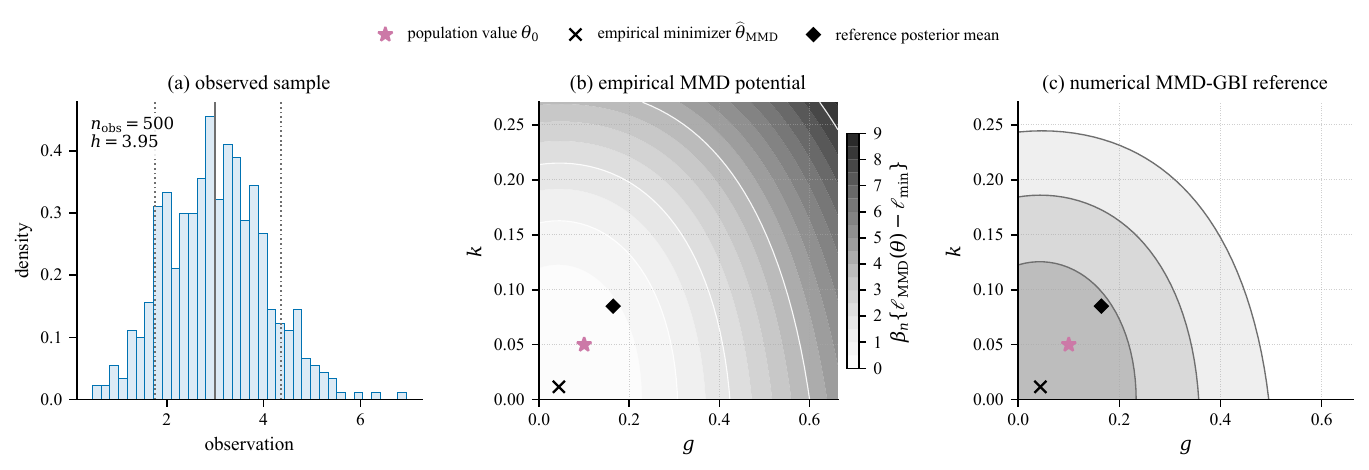}
    \caption{Finite-data target for the \(g\)-and-\(k\) example.
    Panel (a) shows the observed sample, panel (b) the empirical MMD potential \(\beta_n(\ell_{\rm MMD}-\ell_{\min})\), and panel (c) the numerical MMD-GBI posterior. The generating value, empirical minimizer, and posterior mean are marked separately.
    }
    \label{fig:gandk_target_geometry}
\end{figure}

Both SCBD implementations use the unbiased pair contribution
\begin{equation}
\begin{aligned}
    L(\theta ; u,u')
    &=
    k_h\{x_\theta(u),x_\theta(u')\}
    -m_y\{x_\theta(u)\}
    -m_y\{x_\theta(u')\},\\
    x_\theta(u)
    &=
    Q_\theta\{\Phi^{-1}(u)\},
\end{aligned}
    \label{eq:gk_pair_integrand}
\end{equation}
for which \(
    \E\{L(\theta; U,U')\}
    =
    \ell_{\rm MMD}(\theta)\) and 
    \(
    U,U'\overset{\rm iid}{\sim}{\rm Unif}(0,1).
\)
One pair contribution \(L(\theta; U,U')\) requires two independent simulator calls.  These are used to form $\bar\ell_M$ using \eqref{eq:setup_block_estimators} and \eqref{eq:setup_block_average}. MC-SCBD
forms blocks from independent pairs, whereas RQ-SCBD uses independently
scrambled two-dimensional Sobol blocks.  

\subsubsection{Variance scaling and sampler efficiency}

The inverse-normal transform
and the interaction in the self-kernel term make this representation
less regular than the one-dimensional CRPS integral.
At the frozen pilot-grid mode, the fitted finite-range variance
exponents are
\[
    \widehat\alpha_{\rm MC}=1.03,
    \qquad
    \widehat\alpha_{\rm RQMC}=2.26.
\]
The block-selection procedure gives
\[
\begin{aligned}
    \text{RQ-SCBD:}\quad&
    K=7,\quad B=128,\quad M=896,\\
    \text{MC-SCBD:}\quad&
    K=8,\quad B=2048,\quad M=16384.
\end{aligned}
\]
The corresponding simulator counts are \(1792\) and \(32768\) per
weight evaluation.

The Russian-roulette estimator is constructed from iid evaluations of
the same pair contribution \(L(\theta; u,u')\).  It has no block count or
block size.  Its operating point was fixed before the production runs,
using the pilot calculation described in
Appendix~\ref{app:gandk_rr}.  The expected number of MMD pairs per
weight evaluation is
\begin{equation}
    M'
    =
    20{,}000
    +
    12{,}500\,\E(N)
    =
    111{,}400.9,
    \qquad
    \E(N)=7.312072.
    \label{eq:gk_rr_cost}
\end{equation}
This corresponds to \(222{,}801.8\) expected simulator calls. Figure~\ref{fig:gandk_rqmc_diagnostics} places the three operating
points on a common computational scale.  

\begin{figure}[htbp]
    \centering
    \includegraphics[width=\textwidth, clip, trim={0 18 0 8}]
    {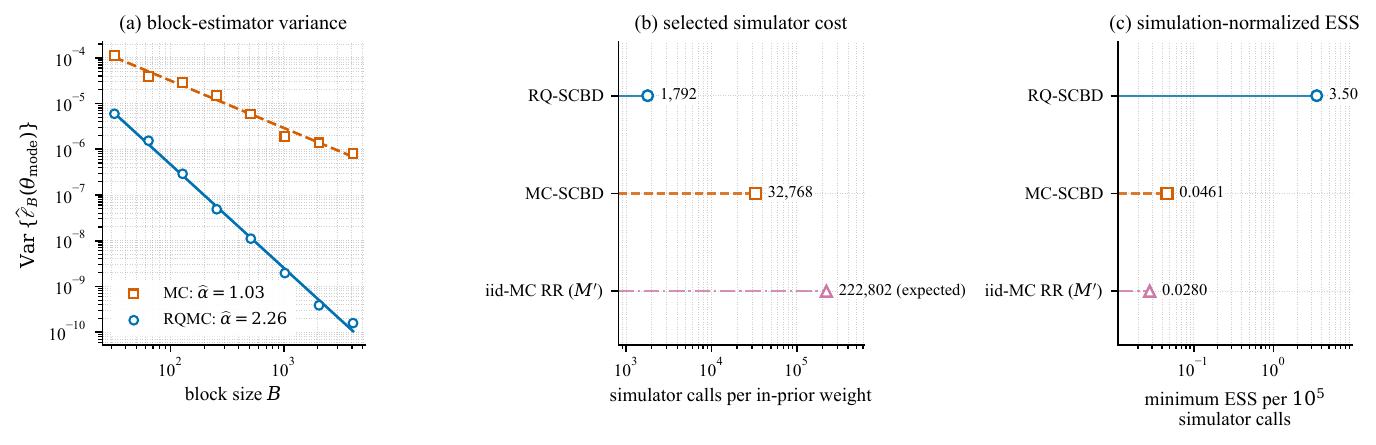}
    \caption{Variance, computational cost, and sampling efficiency for
    the \(g\)-and-\(k\) example; (a) the finite-range
    variance fits; (b) the number of simulator calls per weight evaluation at
    the selected operating point (the Russian-roulette value is \(2M'\) in
    expectation); (c) the minimum ordinary coordinate effective sample
    size per \(10^5\) nominal simulator calls for retained iterations.}
    \label{fig:gandk_rqmc_diagnostics}
\end{figure}

\subsubsection{Posterior comparison}

Figure~\ref{fig:gandk_posterior_comparison} provides a visual check of posterior agreement with the numerical MMD-GBI reference. The SCBD summaries are sign-corrected estimates for their audited finite-block targets, whereas the Russian-roulette signed-ratio estimates target the intended posterior under the required integrability conditions.

\begin{table}[ht!]
    \centering
    \small
    \setlength{\tabcolsep}{2.2pt}
    \caption{Comparison at the selected operating points.  Pair cost is
    \(M=KB\) for SCBD and the expected value \(M'\) for Russian
    roulette.  One MMD pair requires two independent simulator calls.  Mean error is
    the root mean square of the two posterior-mean differences after
    standardization by the marginal numerical-reference posterior
    standard deviations.  ESS is the minimum ordinary coordinate ESS, normalized
    by nominal simulator calls for retained iterations; pilot, initialization,
    and burn-in costs are excluded for all methods.\\}
    \label{tab:gandk_diagnostics}
    \begin{tabular}{@{}lrrrrr@{}}
        \toprule
        Method &
        MMD pairs/weight &
        mean sign &
        acceptance &
        mean error &
        ESS/\(10^5\) sims \\
        \midrule
        RQ-SCBD \((7\times128)\) &
        \(896\) &
        \(1.0000\) &
        \(0.246\) &
        \(0.0212\) &
        \(3.4990\) \\
        MC-SCBD \((8\times2048)\) &
        \(16{,}384\) &
        \(0.9985\) &
        \(0.272\) &
        \(0.0100\) &
        \(0.0461\) \\
        iid-MC RR &
        \(111{,}400.9\) &
        \(0.99543\) &
        \(0.3198\) &
        \(0.0619\) &
        \(0.0280\) \\
        \bottomrule
    \end{tabular}
\end{table}

Table~\ref{tab:gandk_diagnostics} shows that RR uses 6.80 and 124.3 times as many MMD pairs as MC-SCBD and RQ-SCBD, respectively. Its simulation-normalized ordinary coordinate ESS is 0.61 times that of MC-SCBD and about 1/125 that of RQ-SCBD; the efficiency ordering of MC-SCBD and RR is less clear for the signed posterior means (Appendix~\ref{app:gandk_rr}). These comparisons are specific to the dataset, loss representation, and selected operating points; the reported posterior-mean errors also include Monte Carlo variation.

The nine-state discrete reference-relative audit TV values are \(0.00765\) for
RQ-SCBD and \(0.00212\) for MC-SCBD.  These are the selection estimates on the
posterior-weighted audit design.
A separate RQ-SCBD check, not used for selection, uses the same nine
states with \(12800\) new repetitions per state and Gauss--Hermite
reference losses. Its reference-relative TV estimate is
\(0.00907\), with nominal 95\% interval \([0.00764,0.01050]\), which does
not resolve the \(0.01\) threshold.
For RR, the corresponding target discrepancy is zero, provided its signed
weight is integrable, because it is unbiased for the unnormalised target density.

\section{Discussion}
\label{sec:discussion}

Exponentiating an estimated loss creates two distinct problems: variability in the estimated log weight and distortion of the intended posterior. SCBD addresses exponentiation bias in generalised Bayesian inference using independent block estimates of the loss and their sample variance. %
The Bessel correction is exact under an i.i.d. Gaussian block model with unknown variance. For non-Gaussian blocks, cumulant analysis yields posterior total-variation bounds and, under the assumptions of Theorem \ref{cor:method-rates}, improves the fixed-$n$ MC error bound from $O(M^{-1})$ to $O(M^{-2})$, with little overhead beyond block construction. Independently randomised RQMC blocks provide complementary variance reduction. Theorem \ref{thm:stability_budgets} shows how faster block-variance decay reduces the budget growth needed to control leading log-weight variability as the sample size increases.

The experiments illustrate these complementary roles. In the CRPS calibration, RQ-blocks give
a variance exponent close to three, and the selected budget follows the corresponding \(n^{2/3}\) scale. The fitted RQ-SCBD TV slope is consistent with much faster convergence; we believe that is because the third-cumulant contribution is negligible. %
The TV-rate conditions we give are sufficient rather than necessary, are proved with fixed $K$, and do not determine a universally optimal division of budget between $K$ and $B$. However, they do not exploit posterior concentration (see \cite{quiroz2019speeding} and the remarks in Section~\ref{sec:related-work-bounds}), suggesting potential for future work. 

In the random-phase example, the integrated loss and its
RQMC representation remain stable while the pathwise gradient becomes more
oscillatory. Zig--Zag is competitive at \(\kappa=12\),
but event-rate inflation makes RQ-SCBD more efficient at
\(\kappa=20,60\). For the common \(g\)-and-\(k\) MMD-GBI target, Russian-roulette estimator demands greater expected cost and larger finite-run posterior-mean error than MC-SCBD, while RQ-SCBD achieves the highest simulation-normalised efficiency.

The real FTIR example in Appendix~\ref{sec:real_ftir} uses a discontinuous lookup representation.  When corrected and uncorrected weights are evaluated on the same eight blocks,
the Bessel correction reduces the observed grid discrepancy from the exhaustive
finite-data reference.  At the same \(M=256\) budget, direct exponentiation of a one-net RQMC estimate without debiasing gives a visibly different posterior,  %
showing that variance reduction and
exponentiation correction are complementary.

The numerical analysis suggests reporting the pilot screen and target audit separately. For the former, report the pilot design, candidate
ordering, selected \((K,B)\), log-weight variability, and fixed-state
and weighted sign diagnostics. For the latter, describe
the reference loss, construction of audit states and 
posterior-mass weights, and observed normalized target discrepancy.
Table~\ref{tab:screen_audit_comparison}
 shows that the reference audit may leave the pilot choice unchanged, modify it slightly, or reject a substantially smaller configuration.
Without a reliable finite-design reference, budget selection may involve doubling the budget and checking stability. However, this is like an MCMC convergence check and is not guaranteed to identify discrepancies relative to the intended posterior. 

SCBD is most useful when forward loss evaluations are available but reliable
gradients or analytic exponential corrections are not.  It complements exact
gradient-based PDMP samplers \cite{frazier2025exact}, which require problem-specific event-rate bounds, and exact correlated block-Poisson pseudo-marginal methods \citep{quiroz21blockpoisson}, which require control variates and centring choices. MC-SCBD avoids these requirements. The RQMC-extension was straightforward in all our examples. SCBD is easy to use and imposes no additional restriction on the MCMC update but is approximate.

\bibliographystyle{ba}
\bibliography{references}

\begin{thebibliography}{41}
\newcommand{\enquote}[1]{``#1''}
\expandafter\ifx\csname natexlab\endcsname\relax\def\natexlab#1{#1}\fi
\expandafter\ifx\csname url\endcsname\relax
  \def\url#1{{\tt #1}}\fi
\expandafter\ifx\csname urlprefix\endcsname\relax\def\urlprefix{URL }\fi
\ifx\endbibitem\undefined \let\endbibitem\relax\fi

\bibitem[{Al-Jowder et~al.(1997)Al-Jowder, Kemsley, and
  Wilson}]{aljowder1997midinfrared}
Al-Jowder, O., Kemsley, E.~K., and Wilson, R.~H. (1997).
\newblock \enquote{Mid-Infrared Spectroscopy and Authenticity Problems in
  Selected Meats: A Feasibility Study.}
\newblock {\em Food Chemistry\/}, 59(2): 195--201.
\endbibitem

\bibitem[{Andrieu et~al.(2010)Andrieu, Doucet, and
  Holenstein}]{andrieu2010particle}
Andrieu, C., Doucet, A., and Holenstein, R. (2010).
\newblock \enquote{Particle Markov chain Monte Carlo methods.}
\newblock {\em Journal of the Royal Statistical Society: Series B (Statistical
  Methodology)\/}, 72(3): 269--342.
\endbibitem

\bibitem[{Andrieu and Roberts(2009)}]{andrieu2009pseudo}
Andrieu, C. and Roberts, G.~O. (2009).
\newblock \enquote{The pseudo-marginal approach for efficient {Monte Carlo}
  computations.}
\newblock {\em The Annals of Statistics\/}, 37(2): 697--725.
\endbibitem

\bibitem[{Andrieu and Vihola(2015)}]{andrieu2015convergence}
Andrieu, C. and Vihola, M. (2015).
\newblock \enquote{Convergence properties of pseudo-marginal {M}arkov chain
  {M}onte {C}arlo algorithms.}
\newblock {\em The Annals of Applied Probability\/}, 25(2): 1030--1077.
\endbibitem

\bibitem[{Bardenet et~al.(2017)Bardenet, Doucet, and
  Holmes}]{bardenet2017markov}
Bardenet, R., Doucet, A., and Holmes, C. (2017).
\newblock \enquote{On Markov chain Monte Carlo methods for tall data.}
\newblock {\em The Journal of Machine Learning Research\/}, 18(1): 1515--1557.
\endbibitem

\bibitem[{Basu et~al.(1998)Basu, Harris, Hjort, and
  Jones}]{basu1998densitypower}
Basu, A., Harris, I.~R., Hjort, N.~L., and Jones, M.~C. (1998).
\newblock \enquote{Robust and Efficient Estimation by Minimising a Density
  Power Divergence.}
\newblock {\em Biometrika\/}, 85(3): 549--559.
\endbibitem

\bibitem[{Beaumont(2003)}]{beaumont03}
Beaumont, M.~A. (2003).
\newblock \enquote{Estimation of Population Growth or Decline in Genetically
  Monitored Populations.}
\newblock {\em Genetics\/}, 164(3): 1139--1160.
\newline\urlprefix\url{https://doi.org/10.1093/genetics/164.3.1139}
\endbibitem

\bibitem[{Bierkens et~al.(2019)Bierkens, Fearnhead, and
  Roberts}]{bierkens2019zigzag}
Bierkens, J., Fearnhead, P., and Roberts, G.~O. (2019).
\newblock \enquote{The {Zig-Zag} Process and Super-Efficient Sampling for
  {Bayesian} Analysis of Big Data.}
\newblock {\em The Annals of Statistics\/}, 47(3): 1288--1320.
\endbibitem

\bibitem[{Bissiri et~al.(2016)Bissiri, Holmes, and Walker}]{bissiri2016general}
Bissiri, P.~G., Holmes, C.~C., and Walker, S.~G. (2016).
\newblock \enquote{A general framework for updating belief distributions.}
\newblock {\em Journal of the Royal Statistical Society Series B: Statistical
  Methodology\/}, 78(5): 1103--1130.
\endbibitem

\bibitem[{Ceperley and Dewing(1999)}]{ceperley1999penalty}
Ceperley, D. and Dewing, M. (1999).
\newblock \enquote{The penalty method for random walks with uncertain
  energies.}
\newblock {\em The Journal of Chemical Physics\/}, 110(20): 9812--9820.
\endbibitem

\bibitem[{Ch{\'e}rief-Abdellatif and Alquier(2020)}]{cherief2020mmd}
Ch{\'e}rief-Abdellatif, B.-E. and Alquier, P. (2020).
\newblock \enquote{{MMD}-Bayes: Robust Bayesian estimation via maximum mean
  discrepancy.}
\newblock In {\em Symposium on Advances in Approximate Bayesian Inference\/},
  1--21. PMLR.
\endbibitem

\bibitem[{Doucet et~al.(2015)Doucet, Pitt, Deligiannidis, and
  Kohn}]{doucet2015efficient}
Doucet, A., Pitt, M.~K., Deligiannidis, G., and Kohn, R. (2015).
\newblock \enquote{Efficient implementation of {Markov chain Monte Carlo} when
  using an unbiased likelihood estimator.}
\newblock {\em Biometrika\/}, 102(2): 295--313.
\endbibitem

\bibitem[{Drovandi and Pettitt(2011)}]{drovandi2011likelihoodfree}
Drovandi, C.~C. and Pettitt, A.~N. (2011).
\newblock \enquote{Likelihood-Free {Bayesian} Estimation of Multivariate
  Quantile Distributions.}
\newblock {\em Computational Statistics \& Data Analysis\/}, 55(9): 2541--2556.
\endbibitem

\bibitem[{Fearnhead et~al.(2018)Fearnhead, Bierkens, Pollock, and
  Roberts}]{fearnhead2018pdmp}
Fearnhead, P., Bierkens, J., Pollock, M., and Roberts, G.~O. (2018).
\newblock \enquote{Piecewise Deterministic Markov Processes for Continuous-Time
  Monte Carlo.}
\newblock {\em Statistical Science\/}, 33(3): 386--412.
\endbibitem

\bibitem[{Frazier et~al.(2024)Frazier, Knoblauch, and
  Drovandi}]{frazier2024impact}
Frazier, D.~T., Knoblauch, J., and Drovandi, C. (2024).
\newblock \enquote{The Impact of Loss Estimation on Gibbs Measures.}
\newblock {\em arXiv preprint arXiv:2404.15649v1\/}.
\endbibitem

\bibitem[{Frazier et~al.(2025)Frazier, Knoblauch, Jewson, and
  Drovandi}]{frazier2025exact}
Frazier, D.~T., Knoblauch, J., Jewson, J., and Drovandi, C. (2025).
\newblock \enquote{Exact Sampling of Gibbs Measures with Estimated Losses.}
\newline\urlprefix\url{https://arxiv.org/abs/2404.15649}
\endbibitem

\bibitem[{Gao et~al.(2023)Gao, Deistler, and Macke}]{gao2023generalized}
Gao, R., Deistler, M., and Macke, J.~H. (2023).
\newblock \enquote{Generalized Bayesian Inference for Scientific Simulators via
  Amortized Cost Estimation.}
\newblock In {\em Advances in Neural Information Processing Systems\/},
  volume~36.
\endbibitem

\bibitem[{Gneiting and Raftery(2007)}]{gneiting2007strictly}
Gneiting, T. and Raftery, A.~E. (2007).
\newblock \enquote{Strictly Proper Scoring Rules, Prediction, and Estimation.}
\newblock {\em Journal of the American Statistical Association\/}, 102(477):
  359--378.
\endbibitem

\bibitem[{Gr{\"u}nwald and Van~Ommen(2017)}]{grunwald2017inconsistency}
Gr{\"u}nwald, P.~D. and Van~Ommen, T. (2017).
\newblock \enquote{Inconsistency of {Bayesian} inference for misspecified
  linear models, and a proposal for repairing it.}
\newblock {\em Bayesian Analysis\/}, 12(4): 1069--1103.
\endbibitem

\bibitem[{Jewson et~al.(2018)Jewson, Smith, and Holmes}]{jewson2018principles}
Jewson, J., Smith, J.~Q., and Holmes, C. (2018).
\newblock \enquote{Principles of {Bayesian} inference using general divergence
  criteria.}
\newblock {\em Entropy\/}, 20(6): 442.
\endbibitem

\bibitem[{Knoblauch et~al.(2022)Knoblauch, Jewson, and Damoulas}]{knoblauch22}
Knoblauch, J., Jewson, J., and Damoulas, T. (2022).
\newblock \enquote{An Optimization-centric View on Bayes' Rule: Reviewing and
  Generalizing Variational Inference.}
\newblock {\em Journal of Machine Learning Research\/}, 23(132): 1--109.
\newline\urlprefix\url{http://jmlr.org/papers/v23/19-1047.html}
\endbibitem

\bibitem[{L'Ecuyer and Lemieux(2000)}]{l2000variance}
L'Ecuyer, P. and Lemieux, C. (2000).
\newblock \enquote{Variance reduction via lattice rules.}
\newblock {\em Management Science\/}, 46(9): 1214--1235.
\endbibitem

\bibitem[{Lie et~al.(2018)Lie, Sullivan, and Teckentrup}]{lie2018random}
Lie, H.~C., Sullivan, T.~J., and Teckentrup, A.~L. (2018).
\newblock \enquote{Random Forward Models and Log-Likelihoods in Bayesian
  Inverse Problems.}
\newblock {\em SIAM/ASA Journal on Uncertainty Quantification\/}, 6(4):
  1600--1629.
\endbibitem

\bibitem[{Lin et~al.(2000)Lin, Liu, and Sloan}]{LinLiuSloan00}
Lin, L., Liu, K.~F., and Sloan, J. (2000).
\newblock \enquote{A noisy Monte Carlo algorithm.}
\newblock {\em Phys. Rev. D\/}, 61: 074505.
\newline\urlprefix\url{https://link.aps.org/doi/10.1103/PhysRevD.61.074505}
\endbibitem

\bibitem[{Loh(2003)}]{loh2003asymptotic}
Loh, W.-L. (2003).
\newblock \enquote{On the asymptotic distribution of scrambled net quadrature.}
\newblock {\em The Annals of Statistics\/}, 31(4): 1282--1324.
\endbibitem

\bibitem[{Lyne et~al.(2015)Lyne, Girolami, Atchad{\'e}, Strathmann, and
  Simpson}]{lyne2015russian}
Lyne, A.-M., Girolami, M., Atchad{\'e}, Y., Strathmann, H., and Simpson, D.
  (2015).
\newblock \enquote{On {Russian roulette} estimates for {Bayesian} inference
  with doubly-intractable likelihoods.}
\newblock {\em Statistical science\/}, 30(4): 443--467.
\endbibitem

\bibitem[{Matsubara et~al.(2022)Matsubara, Knoblauch, Briol, and
  Oates}]{matsubara2022robust}
Matsubara, T., Knoblauch, J., Briol, F.-X., and Oates, C.~J. (2022).
\newblock \enquote{Robust generalised Bayesian inference for intractable
  likelihoods.}
\newblock {\em Journal of the Royal Statistical Society Series B: Statistical
  Methodology\/}, 84(3): 997--1022.
\endbibitem

\bibitem[{Nicholls et~al.(2012)Nicholls, Fox, and Watt}]{nicholls2012coupled}
Nicholls, G.~K., Fox, C., and Watt, A.~M. (2012).
\newblock \enquote{Coupled MCMC with a randomized acceptance probability.}
\newline\urlprefix\url{https://arxiv.org/abs/1205.6857}
\endbibitem

\bibitem[{Owen(1997)}]{owen1997b}
Owen, A.~B. (1997).
\newblock \enquote{Scrambled net variance for integrals of smooth functions.}
\newblock {\em The Annals of Statistics\/}, 25(4): 1541--1562.
\endbibitem

\bibitem[{Owen(2003)}]{owen2003variance}
--- (2003).
\newblock \enquote{Variance with Alternative Scramblings of Digital Nets.}
\newblock {\em ACM Transactions on Modeling and Computer Simulation\/}, 13(4):
  363--378.
\endbibitem

\bibitem[{Pacchiardi et~al.(2024)Pacchiardi, Khoo, and Dutta}]{pacchiardi24}
Pacchiardi, L., Khoo, S., and Dutta, R. (2024).
\newblock \enquote{{Generalized Bayesian likelihood-free inference}.}
\newblock {\em Electronic Journal of Statistics\/}, 18(2): 3628 -- 3686.
\newline\urlprefix\url{https://doi.org/10.1214/24-EJS2283}
\endbibitem

\bibitem[{Pitt et~al.(2012)Pitt, Silva, Giordani, and Kohn}]{pitt2012some}
Pitt, M.~K., Silva, R.~S., Giordani, P., and Kohn, R. (2012).
\newblock \enquote{On some properties of Markov chain Monte Carlo simulation
  methods based on the particle filter.}
\newblock {\em Journal of Econometrics\/}, 171(2): 134--151.
\endbibitem

\bibitem[{{Quadram Institute Bioscience}(n.d.)}]{quadramFTIRmeat}
{Quadram Institute Bioscience} (n.d.).
\newblock \enquote{Example Datasets for Download: Mid-Infrared Spectra of Fresh
  Minced Meats.}
\newblock Core Science Resources.
\newblock Dataset page; accessed 26 June 2026.
\newline\urlprefix\url{https://csr.quadram.ac.uk/example-datasets-for-download/}
\endbibitem

\bibitem[{Quiroz et~al.(2019)Quiroz, Kohn, Villani, and
  Tran}]{quiroz2019speeding}
Quiroz, M., Kohn, R., Villani, M., and Tran, M.-N. (2019).
\newblock \enquote{Speeding Up {MCMC} by Efficient Data Subsampling.}
\newblock {\em Journal of the American Statistical Association\/}, 114(526):
  831--843.
\endbibitem

\bibitem[{Quiroz et~al.(2021)Quiroz, Tran, Villani, Kohn, and
  Dang}]{quiroz21blockpoisson}
Quiroz, M., Tran, M.-N., Villani, M., Kohn, R., and Dang, K.-D. (2021).
\newblock \enquote{The Block-Poisson Estimator for Optimally Tuned Exact
  Subsampling MCMC.}
\newblock {\em Journal of Computational and Graphical Statistics\/}, 30(4):
  877--888.
\endbibitem

\bibitem[{Schoenberg(1938)}]{schoenberg38}
Schoenberg, I.~J. (1938).
\newblock \enquote{Metric Spaces and Completely Monotone Functions.}
\newblock {\em Annals of Mathematics\/}, 39(4): 811--841.
\newline\urlprefix\url{http://www.jstor.org/stable/1968466}
\endbibitem

\bibitem[{Sherlock et~al.(2015)Sherlock, Thiery, Roberts, and
  Rosenthal}]{sherlock2015efficiency}
Sherlock, C., Thiery, A.~H., Roberts, G.~O., and Rosenthal, J.~S. (2015).
\newblock \enquote{On the efficiency of pseudo-marginal random walk Metropolis
  algorithms.}
\newblock {\em The Annals of Statistics\/}, 43(1): 238--275.
\endbibitem

\bibitem[{Shinozuka and Deodatis(1991)}]{shinozuka1991simulation}
Shinozuka, M. and Deodatis, G. (1991).
\newblock \enquote{Simulation of Stochastic Processes by Spectral
  Representation.}
\newblock {\em Applied Mechanics Reviews\/}, 44(4): 191--204.
\endbibitem

\bibitem[{Sprungk(2020)}]{sprungk2020local}
Sprungk, B. (2020).
\newblock \enquote{On the Local Lipschitz Stability of Bayesian Inverse
  Problems.}
\newblock {\em Inverse Problems\/}, 36(5): 055015.
\endbibitem

\bibitem[{Yang et~al.(2025)Yang, Quiroz, Kohn, and Sisson}]{yang25}
Yang, Y., Quiroz, M., Kohn, R., and Sisson, S.~A. (2025).
\newblock \enquote{{A Correlated Pseudo-Marginal Approach to Doubly Intractable
  Problems}.}
\newblock {\em Bayesian Analysis\/}, 1 -- 30.
\newline\urlprefix\url{https://doi.org/10.1214/25-BA1573}
\endbibitem

\bibitem[{Yin and Dong(2024)}]{yin24besselsphere}
Yin, C. and Dong, H. (2024).
\newblock \enquote{The Bessel function expression of characteristic function.}
\newblock {\em Communications in Statistics - Theory and Methods\/}, 53(22):
  8009--8025.
\newline\urlprefix\url{https://doi.org/10.1080/03610926.2023.2278426}
\endbibitem

\end{thebibliography}

\clearpage

\appendix

\counterwithin{theorem}{section}
\counterwithin{assumption}{section}
\counterwithin{corollary}{section}
\renewcommand{\thetable}{\Alph{section}\arabic{table}}
\setcounter{table}{0}
\renewcommand{\thefigure}{\Alph{section}.\arabic{figure}}
\setcounter{figure}{0}
\renewcommand{\theequation}{\Alph{section}.\arabic{equation}}
\setcounter{equation}{0}

\section{Sign probabilities and variance budgets}
\label{app:negative_sign_probability}

The Gaussian calculation below is used to plan the candidate grid, and the
subsequent convergence result relates it to finite blocks.  The implemented
selection rule nevertheless uses the direct empirical sign screen in
Algorithm~\ref{alg:adaptive_tuning}; the Gaussian calculation is not treated as
a finite-block guarantee.  The sign results below concern the Bessel factor
under fresh auxiliary randomization and do not control either the target factor
\(w_M\) or the weighted sign \(R_M\).
The strategy here is similar to those of \citet{quiroz21blockpoisson} and \citet{yang25}, who derive tractable expressions for tuning by examining the Gaussian case. 

\begin{lemma}[Negative-sign probability under Gaussian blocks]
\label{lem:sign_stability}
Under the Gaussian block model of
Proposition~\ref{prop:exact_debiasing}, let
\[
    \tau^2=\beta_n^2\operatorname{Var}(\bar\ell)
    =\frac{\beta_n^2\sigma^2}{K},
    \qquad
    \nu=\frac{K-3}{2}.
\]
Then the negative-sign probability is
\begin{equation}
    q_{\rm G}^-(K,\tau^2)
    =
    \sum_{i=0}^{\infty}
    \left[
    \overline F_{\chi^2_{K-1}}
    \left(\frac{j_{\nu,2i+1}^2}{\tau^2}\right)
    -
    \overline F_{\chi^2_{K-1}}
    \left(\frac{j_{\nu,2i+2}^2}{\tau^2}\right)
    \right],
    \label{eq:negative_prob_series}
\end{equation}
with \(q_{\rm G}^-(K,0)=0\).
\end{lemma}

\begin{corollary}[Gaussian sign screen]
\label{cor:operational_threshold}
Fix \(K\ge2\) and
\(\delta_{\rm sign}\in(0,1/2)\).  If
\begin{equation}
    \tau^2
    \le
    \tau_{\rm sign}^2(K,\delta_{\rm sign})
    \equiv
    \frac{j_{\nu,1}^2}
    {\chi^2_{K-1,1-\delta_{\rm sign}}},
    \qquad
    \nu=\frac{K-3}{2},
    \label{eq:app_gaussian_sign_screen}
\end{equation}
then
\(q_{\rm G}^-(K,\tau^2)\le\delta_{\rm sign}\).
\end{corollary}

For the finite-block statements, write
\[
    Y_{B,k}(\theta)
    =
    \frac{\widehat\ell_B^{\;(k)}(\theta,\xi_{k})-\ell(\theta)}
    {\sigma_B(\theta)},
    \qquad
    \sigma_B^2(\theta)
    =
    \operatorname{Var}\{\widehat\ell_B^{\;(k)}(\theta,\xi_{k})\},
\]
and define
\[
    T_B(\theta)
    =
    \sum_{k=1}^K
    \{Y_{B,k}(\theta)-\bar Y_B(\theta)\}^2,
    \qquad
    \tau_B^2(\theta)
    =
    \frac{\beta_n^2\sigma_B^2(\theta)}{K}.
\]
Then the squared Bessel argument is
\(z_B^2=\tau_B^2T_B\).  Let
\[
    q_B^-(\theta)=\mathbb P\{J_\nu(z_B)<0\},
    \qquad \nu=(K-3)/2.
\]

\begin{theorem}[Sign convergence for MC-blocks]
\label{thm:sign_convergence}
Fix \(K\) and \(\theta\).  Suppose each block is an average of \(B\) iid
contributions from a fixed distribution with finite, non-zero variance, and the
\(K\) blocks are independent.  More generally, for a triangular array assume
the Lindeberg condition.  Assume also that, along the sequence considered,
\(\tau_B^2(\theta)\to\tau^2(\theta)<\infty\).  Then, at each fixed
\(\theta\),
\begin{equation}
    q_B^-(\theta)
    \longrightarrow
    q_{\rm G}^-\{K,\tau^2(\theta)\}.
    \label{eq:mc_negative_prob_convergence}
\end{equation}
\end{theorem}

\begin{theorem}[Finite-block sign convergence for independently randomized RQMC]
\label{thm:rqmc_sign_convergence}
Fix \(K\) and \(\theta\).  Suppose the \(K\) RQMC blocks use independent randomizations and, for each \(k\),
\[
    \frac{\widehat\ell_B^{\;(k)}(\theta,\xi_{k})-\ell(\theta)}
    {\sigma_B(\theta)}
    \ \Rightarrow\ N(0,1),
\]
and \(\tau_B^2(\theta)\to\tau^2(\theta)<\infty\).  Then
Eq.~\eqref{eq:mc_negative_prob_convergence} holds.
The standardized RQMC central limit theorem is an assumption on the chosen
integrand and randomization; it is not implied by the variance rate alone.
\end{theorem}

\begin{corollary}[Uniform sign control]
\label{cor:uniform_sign}
Let \(\mathcal O\) be a set of states and suppose
\[
    \sup_{\theta\in\mathcal O}
    \left|
    q_B^-(\theta)-q_{\rm G}^-\{K,\tau^2(\theta)\}
    \right|
    \longrightarrow0.
\]
If
\[
    \sup_{\theta\in\mathcal O}\tau^2(\theta)
    \le
    \tau_{\rm sign}^2(K,\delta_{\rm sign})
\]
for some \(\delta_{\rm sign}<1/2\), then
\[
    \sup_{\theta\in\mathcal O}q_B^-(\theta)
    \le
    \delta_{\rm sign}+o(1).
\]
Hence the fixed-state expected sign is uniformly positive for sufficiently
large \(B\) when the zero-weight probabilities are uniformly negligible.
\end{corollary}

\begin{proof}[Proof of Lemma~\ref{lem:sign_stability}]
Under Gaussian blocks,
\[
    T=\frac{(K-1)S_K^2}{\sigma^2}\sim\chi^2_{K-1},
    \qquad
    z^2=\tau^2T.
\]
The function \(J_\nu\) is negative on
\((j_{\nu,2i+1},j_{\nu,2i+2})\), \(i\ge0\).  Summing the corresponding
chi-square probabilities gives Eq.~\eqref{eq:negative_prob_series}.
\end{proof}

\begin{proof}[Proof of Corollary~\ref{cor:operational_threshold}]
For \(\tau=0\), the result is immediate.  For \(\tau>0\), a negative
Bessel factor requires \(z>j_{\nu,1}\).  Therefore
\[
    q_{\rm G}^-(K,\tau^2)
    \le
    \mathbb P\!\left(
        \chi^2_{K-1}>\frac{j_{\nu,1}^2}{\tau^2}
    \right),
\]
and Eq.~\eqref{eq:app_gaussian_sign_screen} makes the right-hand side at most
\(\delta_{\rm sign}\).
\end{proof}

\begin{proof}[Proofs of Theorems~\ref{thm:sign_convergence}
and~\ref{thm:rqmc_sign_convergence}]
For MC-blocks, the univariate central limit theorem and independence of the
blocks give
\((Y_{B,1},\ldots,Y_{B,K})\Rightarrow N_K(0,I_K)\).
For RQMC, the stated one-dimensional convergence and independent
randomizations give the same joint limit.  Hence
\(T_B\Rightarrow\chi^2_{K-1}\), and
\(z_B^2=\tau_B^2T_B\Rightarrow\tau^2\chi^2_{K-1}\).
The boundary of the set on which \(J_\nu(z)<0\) is the countable collection of
Bessel zeros.  The limiting distribution is continuous and assigns this
boundary probability zero.  The portmanteau theorem therefore gives
Eq.~\eqref{eq:mc_negative_prob_convergence}.
\end{proof}

The uniform corollary follows by combining the uniform approximation with
Corollary~\ref{cor:operational_threshold}.

\subsection{Small-noise log-weight variance}
\label{app:small_noise_log_variance}

For
\(z=z_{\beta_n}(S_{K,B}^2,K)\) near zero,
\eqref{eq:F-expansion-small-z} gives
\[
    \log|F_{\beta_n}(S_{K,B}^2,K)|
    =
    -\frac{z^2}{2(K-1)}+O(z^4).
\]
Under the Gaussian block model, \(z^2=\tau_M^2(\theta)Y\) with
\(Y\sim\chi^2_{K-1}\).  For fixed \(K\), the chi-square tail outside
the local expansion region is exponentially small as
\(\tau_M^2(\theta)\to0\), and the resulting small-noise expansion is
\[
    \operatorname{Var}_\xi
    \{\log|F_{\beta_n}(S_{K,B}^2,K)|\}
    =
    \frac{\{\tau_M^2(\theta)\}^2}{2(K-1)}
    +o\!\left(\{\tau_M^2(\theta)\}^2\right).
\]
The block mean and sample variance are independent in this model, so there is
no covariance contribution and
\[
    \operatorname{Var}_\xi\{\log|\widehat Z_M(\theta,\xi)|\}
    =
    \tau_M^2(\theta)
    +\frac{\{\tau_M^2(\theta)\}^2}{2(K-1)}
    +o\!\left(\{\tau_M^2(\theta)\}^2\right).
\]

\subsection{Variance budgets: proof of Theorem~\ref{thm:stability_budgets}}
\label{app:variance_budget_proof}

\begin{proof}
For MC, the lower bounds on \(\beta_n\) and the variance give
\(\tau_M^2(\theta)\ge c_\beta^2c_\ell n^2/M\); the corresponding upper bounds
give sufficiency.
For RQMC, independence across blocks gives \(\operatorname{Var}\{\bar\ell_M(\theta,\xi)\}
    =
    K^{-1}\operatorname{Var}\{\widehat\ell_B^{\;(1)}(\theta,\xi_1)\}.\) Since \(K\) ranges over a fixed finite set and \(M=KB\), the assumed bound is
at most
\(C'M^{-\alpha}\{\log(2+M)\}^{d-1}\).  Substituting
Eq.~\eqref{eq:rqmc_budget_choice} and using
\(\log(2+M_n)=O\{\log(2+n)\}\) gives the result after increasing \(C_0\) if
necessary.
\end{proof}

\section{Proof of results in Sec. \ref{sec:tv-rate-gain-debias}}
\label{app:convergence}
In this Appendix we give convergence rates for the uncorrected and corrected variants of both MC-SCBD and RQ-SCBD when the inverse temperature $\beta_n$ and block size $B_n$ vary with $n$. This is summarized in Table \ref{table_rate}. All \(O\) and \(o\) statements are as \(n\to\infty\), with \(B=B_n\), \(M_n=KB_n\), and \(K\) fixed.
We do not separately consider cases where the block-error variance is zero, so
$v_{B}(\theta)=0$. Since the block error is centered,
$\X_{k}(\theta)=0$ almost surely in that case, so all cumulants
vanish and the bounds below hold trivially.

We define the block-error cumulant generating function (cgf) and its regularity condition in Assumption~\ref{ass:unified-analytic} to establish a common framework in Appendix \ref{app:regularity}. Using the bounds for truncation errors in Lemma \ref{lem:cumulant-envelope}, the expected weight factors is shown in Lemma \ref{lem:unified-weight-expansion}. Corollary~\ref{cor:unified-TV} converts the expansions in Lemma~\ref{lem:unified-weight-expansion} into total variation bounds on the converging posteriors (and shows Assumption~\ref{ass:unified-analytic} plus the smallness conditions in Corollary~\ref{cor:unified-TV} imply Assumption~\ref{ass:finite_block_target} for all sufficiently large $n$). Theorem \ref{cor:method-rates} applies Corollary~\ref{cor:unified-TV} to MC and RQMC blocks by substituting their respective variance and cumulant scales.

\subsection{Notation and regularity condition}\label{app:regularity}
An analyticity assumption will allow us to form and truncate expansions of the cumulant generating function (cgf) of the iid block errors $X=(X_{1},\dots,X_{K})$ in \eqref{eq:block-error}. The cgf is
\begin{align*}
  \psi_{B,\theta}(z)
  \equiv\log\E\exp\{zX_{1}(\theta)\}.
\end{align*}
This cgf enters because 
the $\X_{k}$ are iid with $\X_{k}=\widehat\ell^{\;(k)}_{n,B}-\ell_n(\theta)$ so from \eqref{eq:uncorrected_target_factor} we have
\begin{equation}\label{eq:u-as-psi}
  \log u_{n,M}(\theta)=\log\E \exp\left(-\sum_{k=1}^K \frac{\beta_n \X_{k}}{K}\right)
  =K\psi_{B,\theta}\left(-\beta_n/K\right).
\end{equation}
We show in Appendix~\ref{app:cumulant-expansions} that $w_{n,M}$ is also given in terms of $\psi_{B,\theta}$.
A cumulant expansion of $\psi_{B,\theta}$ will therefore give the $B$-dependence of $u_{n,M}$ and $w_{n,M}$.

For MC-blocks, the within-block errors 
\begin{equation*}
\V_{k,b}(\theta)\equiv L(\theta,\xi_{k,b})-\ell_n(\theta),\ b=1,\dots,B
\end{equation*} 
from \eqref{eq:setup_block_estimators} are also iid and centred given $\theta$ and
$
  \X_{k}(\theta)
  =
  B^{-1}\sum_{b=1}^B \V_{k,b}(\theta)$. Let 
\begin{equation*}
  \sigma^2(\theta)
  =
  \operatorname{Var}\{\V_{k,1}(\theta)\},\quad v_{B}(\theta)
  =\Var\{X_{1}(\theta)\},
  \quad
  \chi_{j,B}(\theta)
  =\cum_j\{X_{1}(\theta)\},
\end{equation*}
for $j\geq2$, so $\chi_{2,B}=v_{B}$. 
We don't treat states where $v_B(\theta)=0$ separately as $\X_{k}(\theta)=0$ almost surely, cumulants vanish and all results hold trivially in that case. Denote by 
\begin{equation}
  \phi^{\mathrm{MC}}_{\theta}(t)
  \equiv
  \log\E
  \exp\left\{
    t{\V_{k,1}(\theta)}/{\sigma(\theta)}
  \right\}
  \label{eq:standardized-MC-cgf}
\end{equation}
the cgf for scaled within-block errors. For MC-blocks, independence within blocks gives
\begin{equation}
  \psi_{B,\theta}(z)=\log\E \exp\left(\sum_{b=1}^B \frac{z\V_{k,b}}{B}\right)
  =
  B\phi^{\mathrm{MC}}_{\theta}
  \left\{
    {z\sigma(\theta)}/{B}
  \right\},
  \label{eq:MC-block-cgf}
\end{equation}
and $v_{B}(\theta)={\sigma^2(\theta)}/{B}$.
  For RQMC blocks, if
\begin{equation}
  \phi^{\mathrm{RQ}}_{B,\theta}(t)
  \equiv
  \log\E_X
  \exp\left\{t{X_{1}(\theta)}/
       {\sqrt{v_{B}(\theta)}}\right\}
  \label{eq:standardized-RQ-cgf}
\end{equation}
is the cgf for standardized block errors then $t=z\,\sqrt{v_{B}(\theta)}$ in \eqref{eq:standardized-RQ-cgf} gives
\begin{equation}
  \psi_{B,\theta}(z)
  =
  \phi^{\mathrm{RQ}}_{B,\theta}
  \left\{
    z\,\sqrt{v_{B}(\theta)}
  \right\}.
  \label{eq:RQ-block-cgf}
\end{equation}
For the two block rules, define
\begin{equation}
  \bigl(q_{B}(\theta),\varphi_{B,\theta}\bigr)
  \equiv
  \begin{cases}
  \left(
    \sigma(\theta)/B,
    \phi^{\mathrm{MC}}_{\theta}
  \right),
  &\text{for MC-blocks},\\[2mm]
  \left(
    \sqrt{v_{B}(\theta)},
    \phi^{\mathrm{RQ}}_{B,\theta}
  \right),
  &\text{for RQ-blocks}.
  \end{cases}
  \label{eq:method-specific-q-varphi}
\end{equation}
Equations~\eqref{eq:MC-block-cgf} and
\eqref{eq:RQ-block-cgf} then give, in both cases,
\begin{equation}
  \psi_{B,\theta}(z)
  =
  \frac{v_{B}(\theta)}
       {q_{B}(\theta)^2}
  \varphi_{B,\theta}
  \{q_{B}(\theta)z\}.
  \label{eq:common-cgf-factorization}
\end{equation}
We seek series expansions for $u_{n,M}$ and $w_{n,M}$ valid at large $B$. For $u_{n,M}$ in \eqref{eq:u-as-psi}, we expand $\psi_{B,\theta}$ at $z=-\beta_n/K$ by expanding $\phi_\theta^{\rm MC}(t)$ at $t=-\beta_n\sigma/KB$ using \eqref{eq:MC-block-cgf}. Analyticity of $\phi_\theta^{\rm MC}$ will bound truncation-errors in terms of \(v_{B}=\sigma^2/B\) and \(q_B=\sigma/B\).
We use the $(q_B,\varphi_{B,\theta})$-notation of \eqref{eq:method-specific-q-varphi} to unify some of the work for $u_{n,M}$ and $w_{n,M}$. 

\begin{assumption}[Uniform analytic cgfs]
\label{ass:unified-analytic}
There are constants $\eta>0$ and $C_\varphi<\infty$ such
that the standardized-error cgf $\varphi_{B,\theta}$ in \eqref{eq:method-specific-q-varphi} extends
analytically to a neighborhood of the complex disc $\{t\in\mathbb C:|t|\leq\eta\}$ and is uniformly bounded 
there, so $\sup_{n,B}\sup_{\theta\in\mathcal O_n}
  \sup_{|t|\leq\eta}
  \left|
    \varphi_{B,\theta}(t)
  \right|
  \leq C_\varphi$ on sets
$\mathcal O_n$ with $\pi_n(\mathcal O_n)=1$.  
For MC-blocks, also suppose that
$\sup_{n,\theta\in\mathcal O_n}\sigma(\theta)<\infty$.
\end{assumption}

In both cases, $\varphi_{B,\theta}$ is the cgf of a centred,
unit-variance variable so  
  $\varphi_{B,\theta}(0)=0$,
  $\varphi'_{B,\theta}(0)=0$, (centred)
  and
  $\varphi''_{B,\theta}(0)=1$ (unit variance).
Differentiating $\psi_{B,\theta}(z)$ in \eqref{eq:common-cgf-factorization} at zero to get the cumulants gives
\begin{equation}\label{eq:cumulant-rescaling}
  \chi_{j,B}(\theta)
  =
  v_B(\theta)q_B(\theta)^{j-2}
  \varphi_{B,\theta}^{(j)}(0),
\end{equation}
so $v_{B}$ controls the Gaussian term,
$v_{B}q_{B}$ controls the first possible non-Gaussian term and \(\varphi_{B,\theta}^{(j)}(0),\ j\ge 3\) is bounded by Assumption~\ref{ass:unified-analytic}. 
For the third cumulant, \eqref{eq:cumulant-rescaling} gives
\begin{align*}
    \chi_{3,B}(\theta)&=\sigma^3\operatorname{cum}_3\{\V_{k,1}/\sigma\}/B^2={\kappa_3(\theta)}/{B^2}\qquad\mbox{for MC-blocks},\\
\intertext{where
$\kappa_3(\theta)=\operatorname{cum}_3\{\V_{k,1}(\theta)\}
$ does not depend on $B$ (with any $n$-dependence suppressed), and}
  \chi_{3,B}(\theta)&=v_B(\theta)^{3/2}
    \left\{
      \phi^{\mathrm{RQ}}_{B,\theta}
    \right\}^{(3)}(0)\qquad\mbox{for RQ-blocks}.
\end{align*}

Analyticity in Assumption~\ref{ass:unified-analytic} supports a Statulevičius-type cumulant bound in Lemma~\ref{lem:cumulant-envelope}. We are assuming that the higher cumulants ${\rm cum}_j(X_k/\sqrt{v_B}),\ j\ge 3$ of the scaled RQ-block error are $O(1)$, not $o(1)$. Thus no central limit theorem is assumed for $X_k$ in RQ-blocks. Our assumption requires the standardized block errors to satisfy a uniform exponential-tail bound. Assumption~\ref{ass:unified-analytic} holds uniformly in $n$ and over $\cO_n$ so we can treat convergence with  $M$ at fixed $n$ and also how $M$ should depend on $n$ for convergence with $n$ increasing.
 
\subsection{Bounds for truncation errors}

\begin{lemma}[Uniform cumulant and Taylor bounds]
\label{lem:cumulant-envelope}
Under Assumption~\ref{ass:unified-analytic}, there is a constant $C<\infty$ such that, uniformly in $n$, $B$, $\theta\in\mathcal O_n$, and $j\geq2$,
\begin{equation}
  \frac{\abs{\chi_{j,B}(\theta)}}{j!}
  \leq C v_{B}(\theta)q_{B}(\theta)^{j-2}\eta^{-j}.
  \label{eq:cumulant-envelope-a}
\end{equation}
For each fixed integer $m\geq3$, whenever
$\abs zq_{B}(\theta)\leq\eta/2$,
\begin{equation}
  \abs{
  \psi_{B,\theta}(z)
  -\sum_{j=2}^{m-1}
  \frac{\chi_{j,B}(\theta)}{j!}z^j}
  \leq C_m v_{B}(\theta)q_{B}(\theta)^{m-2}\abs z^m.
  \label{eq:cgf-tail-a}
\end{equation}
\end{lemma}
\begin{proof}
Cauchy's estimate $|\varphi^{(j)}_{B,\theta}(0)|/j!\leq C_\varphi\eta^{-j}$ applied to \eqref{eq:cumulant-rescaling} on the disc of radius $\eta$ proves \eqref{eq:cumulant-envelope-a}.  For \eqref{eq:cgf-tail-a}, put
$x=\abs zq_{B}(\theta)/\eta\leq1/2$ and sum the geometric tail:
\begin{align*}
  \sum_{j=m}^{\infty}
  \frac{\abs{\chi_{j,B}(\theta)}}{j!}\abs z^j
  &\leq
  C v_{B}(\theta)\abs z^2\eta^{-2}
  \sum_{j=m}^{\infty}x^{j-2}
  \\
  &\leq
  2C\eta^{-m}v_{B}(\theta)q_{B}(\theta)^{m-2}\abs z^m.\\[-0.4in]
\end{align*}
\end{proof}
The bound in \eqref{eq:cgf-tail-a} is used below at
$
  z=-{\beta}/{K}
$
in \eqref{eq:u-cum-bound} and
$
  z=\beta c_k(u)
$
in \eqref{eq:w-spherical-cgf}.
Since $K$ is fixed and the coefficients $c_k(u)$ are uniformly
bounded, the condition $\beta q=o(1)$ gives
$|z|q\leq\eta/2$ in both cases, uniformly for all
sufficiently large $n$, so the lemma applies.

\subsection{Cumulant expansion}
\label{app:cumulant-expansions}

\begin{lemma}[Cumulant expansions]
\label{lem:unified-weight-expansion}
Fix $K\geq2$ and suppose Assumption~\ref{ass:unified-analytic} holds.  Write
\begin{equation*}
  \beta=\beta_n,
  \qquad v=v_{B}(\theta),
  \qquad q=q_{B}(\theta),
  \qquad \chi_3=\chi_{3,B}(\theta).
\end{equation*}
The following expansions hold uniformly in $\theta\in\mathcal O_n$.
\begin{enumerate}[label=(\alph*)]
\item If $\beta q=o(1)$ and $\beta^2v=o(1)$ as $n\to\infty$, then
\begin{equation}
  u_{n,M}(\theta)
  =
  1+\frac{\beta^2v}{2K}
  +R^{(u)}_{n,B}(\theta),
  \quad
  \abs{R^{(u)}_{n,B}(\theta)}
  \leq
  C\left\{
    \beta^3vq+(\beta^2v)^2
  \right\}.
  \label{eq:u-variance-expansion-a}
\end{equation}

\item If $\beta q=o(1)$ and $\beta^3vq=o(1)$ as $n\to\infty$, then
\begin{equation}\label{eq:w-full-expansion}
  w_{n,M}(\theta)
  =1+\frac{\beta^3\chi_3}{3K^2}
  +R^{(w)}_{n,B}(\theta),
  \quad
  \abs{R^{(w)}_{n,B}(\theta)}
  \leq C\left\{\beta^4vq^2+(\beta^3vq)^2\right\}.
\end{equation}
\end{enumerate}
The constant $C$ depends only on $K$, $\eta$, and $C_\varphi$.
\end{lemma}

\begin{proof}
\noindent\textit{(1) Uncorrected factor.}
For Lemma~\ref{lem:cumulant-envelope} to apply at
$z=-\beta/K$, we require
$
  {\beta q}/{K}\leq {\eta}/{2}.
$
This holds for sufficiently large $n$ as
$\beta q=o(1)$. Applying Lemma~\ref{lem:cumulant-envelope}
at $m=3$ gives
\begin{equation}\label{eq:u-cum-bound}
  \log u_{n,M}(\theta)
  =
  K\psi\left(-\frac{\beta}{K}\right)
  =
  \frac{\beta^2v}{2K}
  +O(\beta^3vq).
\end{equation}
Since
$
  \beta^3vq
  =
  (\beta^2v)(\beta q),
$
the right-hand side is uniformly of order
$\beta^2v$ and tends to zero under the stated conditions, so the expansion
$
  \exp(x)=1+x+O(x^2)
$
gives \eqref{eq:u-variance-expansion-a}.\\

\noindent\textit{(2) Corrected factor, spherical representation.} The debiasing function $F_\beta$ in \eqref{eq:bessel_correction} has the form of the characteristic function of a random variable uniform on a unit sphere \citep[Lemma~2.1]{yin24besselsphere}. This is actually the spherical form of a Laplace-transform identity underlying Proposition~\ref{prop:exact_debiasing}, taken on a complex domain. Let
\begin{equation*}
  H=\left\{u\in\mathbb R^K:\sum_{k=1}^Ku_k=0\right\}.
  \qquad
  \mathbb S_H=\{u\in H:\norm u=1\}.
\end{equation*}
and let $\sigma_H$ be normalized surface measure on $\mathbb S_H$.
Let \(X=X_B=(\X_{1},\ldots,\X_{K})\in\mathbb R^K\) be the vector of scalar block errors in \eqref{eq:block-error}; its dimension is \(K\), irrespective of the dimension \(d\) of the randomized input. Let
$A=A_{B}=K^{-1}\sum_{k=1}^K \X_{k}$ so that $S^2=S^2_{K,B}=(K-1)^{-1}\sum_{k=1}^K\{\X_{k}-A_{B}\}^2$. Then
\begin{equation}
  \int_{\mathbb S_H}
  \exp\left\{
    \frac{i\beta}{\sqrt K}\sum_{k=1}^Ku_k\X_{k}
  \right\}\sigma_H(\dd u)=F_\beta(S^2,K)
  \label{eq:bessel-spherical}
\end{equation}
for $F_\beta(S^2,K)$ in \eqref{eq:bessel_correction}. We now show this. If \(U\) is uniform on \(\mathbb S^{m-1}\subset\mathbb R^m\) then its characteristic function \citep{schoenberg38,yin24besselsphere} is
\begin{equation}\label{eq:bessel-char-identity}
\E\!\left[e^{i t^\top U}\right]
=
\Gamma\left(\frac m2\right)
\left(\frac{\|t\|}{2}\right)^{-m/2+1}
J_{m/2-1}(\|t\|), \qquad \mbox{$t\in \mathbb R^m$}.
\end{equation}
We want to choose $m$ and $t$ so that the right hand side is $F_\beta(s^2,K)$ in \eqref{eq:bessel_correction}. If we take $m=K-1$ we get $m/2-1=\nu$ in \eqref{eq:bessel-arguments-nu-z}. We therefore need $U$ uniform on $\mathbb S^{m-1}=\mathbb S^{K-2}$ and $\|t\|=z_\beta(S^2,K)$.

We get both of these by taking $t=\beta (X-A\bm1)/\sqrt{K}$ and $u\in \mathbb S_H$. The subspace $H$ has dimension $K-1$, so its unit sphere is
isometric to $\mathbb S^{K-2}$. If $Q$ is a $K\times K-1$ matrix with orthonormal columns spanning $H$, then $y\mapsto Qy$ is an isometry from $\mathbb R^{K-1}$ onto $H$ and carries uniform surface measure on \(\mathbb S^{K-2}\) (embedded in $\mathbb R^{K-1}$) to $\sigma_H$ on \(\mathbb S_H\) (also $K-2$ dimensional but embedded in $\mathbb R^{K}$). Now consider $u^TX$: because $\sum_k u_k=0$, we have $u^TX=u^T(X-A\bm1)$, and $\|(X-A\bm1)\|^2=(K-1)S^2$. It follows on the one hand that replacing $(X-A\bm1)$ with $X$ to give \eqref{eq:bessel-spherical} changes nothing and on the other hand
$\|t\|=\beta\,\sqrt{(K-1)S^2/K}$ so $\|t\|=z_\beta(S^2,K)$ and \eqref{eq:bessel-char-identity} becomes \eqref{eq:bessel-spherical}.\\

\noindent\textit{(2 continued) Corrected factor, expansion.} From the definition of $w_{n,M}$ in \eqref{eq:wM_definition} and $A_B$ above,
we have $w_{n,M}=\E_X\{\exp(-\beta A_B)F_\beta\,\}$. From the $F_\beta$ representation in \eqref{eq:bessel-spherical},
\begin{align}
w_{n,M}(\theta)
&=
\E_X
\left[
\int_{\mathbb S_H}
\exp\left\{-\beta \frac{1}{K}\sum_{k=1}^K \X_{k}+
\frac{i\beta}{\sqrt K}\sum_{k=1}^K u_k\X_{k}
\right\}
\sigma_H(du)
\right],\nonumber\\
\intertext{exchanging integrals and using $\E_X \exp(z\X_{k})=\exp\{\psi(z)\}$ with independent blocks,}
&=\int_{\mathbb S_H}
  \exp\left\{\sum_{k=1}^K\psi\{\beta c_k(u)\}\right\}
  \sigma_H(\dd u),
  \qquad
  c_k(u)=-\frac1K+\frac{i}{\sqrt K}u_k,
  \label{eq:w-spherical-cgf}
\end{align}
Fubini applies as the complex exponential has modulus one, and $\sigma_H$ is normalised so $\E_X e^{-\beta A}=u_{n,M}<\infty$ is sufficient. However, $\beta q=o(1)$ puts \(-\beta/K\) inside the analytic region in Assumption~\ref{ass:unified-analytic}, so $u_{n,M}=\exp\left\{
K\psi\left(-{\beta}/{K}\right)\right\}<\infty$.

Put $G(u)=\sum_{k=1}^K\psi\{\beta c_k(u)\}$ in \eqref{eq:w-spherical-cgf} and let $Q_j(u)=\sum_{k=1}^Kc_k(u)^j$. Expand the cgf,
\begin{equation*}
  G(u)=\frac{\beta^2\chi_2}{2}Q_2(u)+\frac{\beta^3\chi_3}{6}Q_3(u)+T_4(u).
\end{equation*}
The identities $\sum_ku_k=0$ and $\sum_ku_k^2=1$ give
\begin{equation*}
  Q_2(u)
  =\frac1K-\frac{2i}{K\sqrt K}\sum_{k=1}^Ku_k
  -\frac1K\sum_{k=1}^Ku_k^2
  =0
\end{equation*}
and
\begin{equation*}
  Q_3(u)
  =\frac{2}{K^2}
  -\frac{i}{K^{3/2}}\sum_{k=1}^Ku_k^3.
\end{equation*}
For $C_K=K^{-1}+K^{-1/2}$, $\abs{c_k(u)}\leq C_K$ and Lemma~\ref{lem:cumulant-envelope} gives, when $\beta qC_K/\eta\leq1/2$,
\begin{align*}
  \sup_{u\in\mathbb S_H}\abs{T_4(u)}
  &\leq
  K\sum_{j=4}^{\infty}
  \frac{\abs{\chi_j}\beta^j}{j!}C_K^j
  \\
  &\leq
  Cv\beta^2\sum_{j=4}^{\infty}
  (\beta qC_K/\eta)^{j-2}
  \leq C\beta^4vq^2.
\end{align*}
The cubic term is uniformly $\cO(\beta^3vq)$ so $G(u)$ is $o(1)$ under the stated smallness conditions and we can expand
$\exp\{G(u)\}=1+G(u)+O\{|G(u)|^2\}$ and take the expectation in $U\sim \sigma_H$ to get
\begin{align*}
  w_{n,M}(\theta)
  &=1+\frac{\beta^3\chi_3}{6}
  \int_{\mathbb S_H}Q_3(u)\,\sigma_H(\dd u)
  +\cO\!\left\{\beta^4vq^2+(\beta^3vq)^2\right\}
  \\
  &=1+\frac{\beta^3\chi_3}{3K^2}
  +\cO\!\left\{\beta^4vq^2+(\beta^3vq)^2\right\},
\end{align*}
where $\E_U Q_3(U)=2/K^2$ by symmetry of $\mathbb S_H$,
which proves the result.
\end{proof}

\subsection{Variation bounds}
\label{app:convergence-TV-bounds-preamble}
Let
$\displaystyle
  \norm{g}_{\infty,n}
  \equiv\esssup_{\theta\sim\pi_n}\abs{g(\theta)}
$
and for $B=B_n$, write
\begin{equation*}
  v_n(\theta)=v_{B_n}(\theta),
  \qquad
  q_n(\theta)=q_{B_n}(\theta),
  \qquad
  \chi_{3,n}(\theta)=\chi_{3,B_n}(\theta),
  \qquad
  M_n=KB_n,
\end{equation*}

\begin{corollary}[Unified total-variation expansions]
\label{cor:unified-TV}
Under Assumption~\ref{ass:unified-analytic}:
\begin{enumerate}[label=(\alph*)]
\item If $\norm{\beta_nq_n}_{\infty,n}\to0$ and
$\norm{\beta_n^2v_n}_{\infty,n}\to0$, then
\begin{align*}
  \norm{\pi^{(0)}_{n,M_n}-\pi_n}_{\TV}
  &=\frac{\beta_n^2}{4K}
  \E_{\pi_n}\abs{v_n-\E_{\pi_n}v_n}
  +
  \cO\!\left
  (\norm{\beta_n^3v_nq_n}_{\infty,n}
  +\norm{\beta_n^2v_n}_{\infty,n}^2\right)
  \\
  &=\cO\!\left(\norm{\beta_n^2v_n}_{\infty,n}\right)
  \longrightarrow0.
\end{align*}

\item If $\norm{\beta_nq_n}_{\infty,n}\to0$ and
$\norm{\beta_n^3v_nq_n}_{\infty,n}\to0$, then
$w_{n,M_n}>0$ $\pi_n$-almost everywhere for all sufficiently large $n$, and
\begin{align*}
  \norm{\pi_{n,M_n}-\pi_n}_{\TV}
  &=\frac{\beta_n^3}{6K^2}
  \E_{\pi_n}\abs{\chi_{3,n}-\E_{\pi_n}\chi_{3,n}}
  +
  \cO\!\left
  (\norm{\beta_n^4v_nq_n^2}_{\infty,n}
  +\norm{\beta_n^3v_nq_n}_{\infty,n}^2\right)
  \\
  &=\cO\!\left(\norm{\beta_n^3v_nq_n}_{\infty,n}\right)
  \longrightarrow0.
\end{align*}
\end{enumerate}
\end{corollary}

\begin{proof}
For either target, the exact identity is
\begin{equation}
  \norm{\pi^f_n-\pi_n}_{\TV}
  =\frac12\E_{\pi_n}
  \abs{\frac{f_n}{\E_{\pi_n}f_n}-1},
  \qquad
  \pi^f_n(\dd\theta)
  =\frac{f_n(\theta)}{\E_{\pi_n}f_n}\pi_n(\dd\theta).
\end{equation}
If $f_n=1+a_n+r_n$ with
$\norm{a_n}_{\infty,n}+\norm{r_n}_{\infty,n}=o(1)$, then
\begin{equation}
  \frac{f_n}{\E_{\pi_n}f_n}-1
  =a_n-\E_{\pi_n}a_n+O
  \left\{\norm{r_n}_{\infty,n}+\norm{a_n}_{\infty,n}^2\right\}.
  \label{eq:normalization-step}
\end{equation}
Apply \eqref{eq:normalization-step} to \eqref{eq:u-variance-expansion-a} with
$a_n=\beta_n^2v_n/(2K)$, and to \eqref{eq:w-full-expansion} with
$a_n=\beta_n^3\chi_{3,n}/(3K^2)$.  The bound
$\abs{\chi_{3,n}}\leq Cv_nq_n$ gives the displayed remainders.  Finally,
$\norm{\beta_n^3v_nq_n}_{\infty,n}
\leq\norm{\beta_n^2v_n}_{\infty,n}
\norm{\beta_nq_n}_{\infty,n}$, and
$\norm{\beta_n^4v_nq_n^2}_{\infty,n}
\leq\norm{\beta_n^3v_nq_n}_{\infty,n}
\norm{\beta_nq_n}_{\infty,n}$.
\end{proof}

\subsection{Posterior convergence in total variation: Proof of Theorem~\ref{cor:method-rates}}
\label{app:convergence-TV-bounds-main}
 
{\bf Theorem~\ref{cor:method-rates}} (MC and RQMC rates). 
{\it
Let \(K\ge2\) be fixed, \(B_n\to\infty\), and \(M_n=KB_n\). Suppose Assumption \ref{ass:unified-analytic} holds. The uncorrected and Bessel-corrected posteriors converge to the target posterior in total variation.
\begin{enumerate}[label=(\alph*)]
\item For \textbf{MC-blocks} satisfying Assumption~\ref{ass:unified-analytic}. If
$
  {\beta_n^2}/{M_n}\longrightarrow0,
$
then
\begin{equation}\label{eq:u-post-expansion}
  \left\|
    \pi^{(0)}_{n,M_n}-\pi_n
  \right\|_{\mathrm{TV}}
  =
  O\!\left(
    \frac{\beta_n^2}{M_n}
  \right).
\end{equation}
If
$
  {\beta_n^3}/{M_n^2}\longrightarrow0,
$
then
\begin{equation}\label{eq:w-post-expansion-app}
  \left\|
    \pi_{n,M_n}-\pi_n
  \right\|_{\mathrm{TV}}
  =O\!\left(
    \frac{\beta_n^3}{M_n^2}
  \right).
\end{equation}
When $\beta_n\asymp n$, $M_n/n^2\!\to\! \infty$ is sufficient for convergence of the uncorrected posterior and $M_n/n^{3/2}\to\infty$ in the debiased case. 

\item For \textbf{RQMC blocks} satisfying Assumption~\ref{ass:unified-analytic}, suppose
\begin{equation}\label{eq:RQ-variance-rate-a-app}
  \|v_{B_n}\|_{\infty,n}
  \leq
  C B_n^{-\alpha}
  \{\log(2+B_n)\}^{d-1}
\end{equation}
for constants $C<\infty$, $\alpha>0$, and fixed
randomized-input dimension $d$.  Since $K$ is fixed and
$M_n=KB_n$, define
\[
  \delta^{\mathrm{RQ}}_{n,M_n}
  \equiv
  \beta_n^2M_n^{-\alpha}
  \{\log(2+M_n)\}^{d-1}.
\]
Then, if
$\delta^{\mathrm{RQ}}_{n,M_n}\to0$,
\begin{align}
  \left\|
    \pi^{(0)}_{n,M_n}-\pi_n
  \right\|_{\mathrm{TV}}
  &=
  O\!\left(\delta^{\mathrm{RQ}}_{n,M_n}\right),
  \label{eq:RQ-u-TV-a-app}
  \\
  \left\|
    \pi_{n,M_n}-\pi_n
  \right\|_{\mathrm{TV}}
  &=
  O\!\left\{
    \left(\delta^{\mathrm{RQ}}_{n,M_n}\right)^{3/2}
  \right\}.
  \label{eq:RQ-w-TV-a-app}
\end{align}
When $\beta_n\asymp n$, the sufficient condition
$\delta^{\mathrm{RQ}}_{n,M_n}\to0$ is equivalently
\[
  \frac{M_n}
  {n^{2/\alpha}
   \{\log(2+M_n)\}^{(d-1)/\alpha}}
  \longrightarrow\infty.
\]
\end{enumerate}
}

\begin{proof}
Let $\kappa_{2}(\theta)=\sigma^2(\theta)$ and $\kappa_3(\theta) =\operatorname{cum}_3\{\V_{k,1}(\theta)\}$ be cumulants of $\V_{k,b}$. 
Only substitution is required.  In the MC case,
\[
  v_{B_n}
  =
  \frac{\kappa_{2}}{B_n},
  \qquad
  q_{B_n}
  =
  \frac{\sigma}{B_n},
  \qquad
  \chi_{3,B_n}
  =
  \frac{\kappa_{3}}{B_n^2}.
\]
Since $M_n=KB_n$ and $K$ is fixed, Corollary~\ref{cor:unified-TV}
gives the displayed expansions. If ${\beta_n^2}/{M_n}\longrightarrow 0$, 
then
\begin{eqnarray}
\left\|
    \pi^{(0)}_{n,M_n}-\pi_n
  \right\|_{\mathrm{TV}}
 & = & \frac{\beta_n^2}{4M_n}
  \E_{\pi_n}
  \left|
    \kappa_{2}
    -
    \E_{\pi_n}\kappa_{2}
  \right|
  +
 O\!\left\{
    \frac{\beta_n^3}{M_n^2}
    +
    \left(\frac{\beta_n^2}{M_n}\right)^2
  \right\} \nonumber 
  \\  & = & \frac{\beta_n^2}{4M_n}
  \E_{\pi_n}
  \left|
    \kappa_{2}
    -
    \E_{\pi_n}\kappa_{2}
  \right| +o\left(
    \frac{\beta_n^2}{M_n}
  \right)\label{eq:MC-u-TV-a-app-full}
  \end{eqnarray}
and Equation \eqref{eq:MC-u-TV-a-app-full} gives Equation \eqref{eq:u-post-expansion}.
If
$
  {\beta_n^3}/{M_n^2}\longrightarrow0,
$
then
\begin{eqnarray}
  \left\|
    \pi_{n,M_n}-\pi_n
  \right\|_{\mathrm{TV}}
  &=&
  \frac{\beta_n^3}{6M_n^2}
  \E_{\pi_n}
  \left|
    \kappa_{3}
    -
    \E_{\pi_n}\kappa_{3}
  \right|
  +
   O\!\left\{
    \frac{\beta_n^4}{M_n^3}
    +
    \left(\frac{\beta_n^3}{M_n^2}\right)^2
  \right\} \nonumber \\ &=&
  \frac{\beta_n^3}{6M_n^2}
  \E_{\pi_n}
  \left|
    \kappa_{3}
    -
    \E_{\pi_n}\kappa_{3}
  \right| + o\left(\frac{\beta_n^3}{M_n^2}\right)
  \label{eq:w-post-expansion-app-full}
\end{eqnarray}
and Equation \eqref{eq:w-post-expansion-app-full} gives Equation \eqref{eq:w-post-expansion-app}
For RQMC, $q_n=v_n^{1/2}$ and \eqref{eq:RQ-variance-rate-a-app} gives
\begin{equation*}
  \begin{aligned}\norm{\beta_n^2v_n}_{\infty,n}&=\cO(\delta^{\mathrm{RQ}}_{n,M_n}),\\
  \qquad
  \norm{\beta_nq_n}_{\infty,n}&=\cO\{(\delta^{\mathrm{RQ}}_{n,M_n})^{1/2}\},\\
  \qquad
  \norm{\beta_n^3v_nq_n}_{\infty,n}&=\cO\{(\delta^{\mathrm{RQ}}_{n,M_n})^{3/2}\}.\end{aligned}
\end{equation*}
Also,
\begin{equation*}
  \beta_n^4v_nq_n^2=(\beta_n^2v_n)^2,
  \qquad
  (\beta_n^3v_nq_n)^2=(\beta_n^2v_n)^3.
\end{equation*}
Substituting these into Lemma~\ref{lem:unified-weight-expansion} and Corollary~\ref{cor:unified-TV}, we derive the posterior convergence rate. If
$\delta^{\mathrm{RQ}}_{n,M_n}\to0$,
\begin{align}
  \left\|
    \pi^{(0)}_{n,M_n}-\pi_n
  \right\|_{\mathrm{TV}}
  &=
  \frac{\beta_n^2}{4K}
  \E_{\pi_n}
  \left|
    v_{B_n}-\E_{\pi_n}v_{B_n}
  \right|
  +
  O\!\left\{
    \left(\delta^{\mathrm{RQ}}_{n,M_n}\right)^{3/2}
  \right\}
  \notag\\
  &=
  O\!\left(\delta^{\mathrm{RQ}}_{n,M_n}\right),
  \label{eq:RQ-u-TV-a-app-full}
  \\
  \left\|
    \pi_{n,M_n}-\pi_n
  \right\|_{\mathrm{TV}}
  &=
  \frac{\beta_n^3}{6K^2}
  \E_{\pi_n}
  \left|
    \chi_{3,B_n}-\E_{\pi_n}\chi_{3,B_n}
  \right|
  +
  O\!\left\{
    \left(\delta^{\mathrm{RQ}}_{n,M_n}\right)^2
  \right\}
  \notag\\
  &=
  O\!\left\{
    \left(\delta^{\mathrm{RQ}}_{n,M_n}\right)^{3/2}
  \right\}
  \label{eq:RQ-w-TV-a-app-full}
\end{align}
and Equations \eqref{eq:RQ-u-TV-a-app-full} and \eqref{eq:RQ-w-TV-a-app-full} give Equations \eqref{eq:RQ-u-TV-a-app} and \eqref{eq:RQ-w-TV-a-app} respectively.
\end{proof}

%
%
%

\section{Additional Experimental Details and Checks}
\label{app:experimental_details}

This appendix records the selection rules, implementation settings, and
finite-block checks used in Section~\ref{sec:experiments}. An overview of the experiments and the points they make is given in Table~\ref{tab:experiment_overview}.
\begin{table}[htbp]
    \centering
    \small
    \setlength{\tabcolsep}{4pt}
    \caption{Structure of the numerical study.  %
    }
    \label{tab:experiment_overview}
    \begin{tabular}{@{}p{0.16\textwidth}p{0.27\textwidth}p{0.16\textwidth}p{0.32\textwidth}@{}}
        \toprule
        Example & Stochastic representation & Reference & Main question \\
        \midrule
        CRPS-loss \,\,
        Section~\ref{sec:crps_calibration} &
        One-dimensional randomized stratification &
        Exact &
        Do variance and TV-rates match theory, does the variance rate carry through to the
        per-proposal budget, and what do \(q_0\) and \(\epsilon_w\) control? \\
        Random phase Section~\ref{sec:random_phase_simulator}&
        Smooth phase integral with an oscillatory pathwise derivative &
        Exact Gaussian target &
        How does event-rate inflation in exact Zig--Zag change the comparison with RQ-SCBD? \\
        \(g\)-and-\(k\) MMD Section~\ref{sec:g_and_k_benchmark}&
        Two-dimensional pair integral over simulator draws &
        Fixed numerical grid &
        How do RQMC blocks, MC blocks, and Russian roulette
        compare in accuracy and cost for the same target? \\
        FTIR spectra Appendix \ref{app:real_ftir_details}&
        Three-dimensional lookup over specimen and channel-pair indices &
        Exhaustive finite-data grid &
        Does Bessel correction improve on single-net RQMC at a fixed
        lookup budget? \\
        \bottomrule
    \end{tabular}
\end{table}

The
Algorithm~\ref{alg:adaptive_tuning} pilot and the Algorithm~\ref{alg:finite_block_audit} audit serve distinct roles.  The
former records log-weight variability and screens signs using generated
corrected weights; the latter combines reference losses with
posterior-mass weights to assess the finite-block target.  Reference
losses are required only at the frozen audit states and may therefore be
computed by an expensive offline procedure.  Pilot, audit, and holdout
designs are fixed before the corresponding candidate evaluations.

For the baseline comparisons, the loss, prior, learning scale, bandwidth, and other target-defining
constants are fixed before selecting a block configuration, \(M=KB\). Russian
roulette has no block split; its cost \(M'\) denotes the expected
number of MMD-pair evaluations per weight evaluation. Algorithm~\ref{alg:adaptive_tuning} uses \(q_0=0.05\) as the negative-weight tolerance in \eqref{eq:negative_probability_def_intro}.  When an exact or sufficiently accurate
reference loss is available, Algorithm~\ref{alg:finite_block_audit} also uses
\(\epsilon_w=0.01\) as a posterior-weighted finite-state TV tolerance. This check is not implied by sign stability alone. The main examples use this audit to choose \((K,B)\); the FTIR appendix compares fixed configurations and reports the associated audit diagnostics. The audit reference loss is analytic in the CRPS and random-phase, frozen numerical for \(g\)-and-\(k\) and exhaustive finite-data evaluation for FTIR. These offline calculations are excluded from the reported production budgets. Section~\ref{sec:block_configuration_screen} describes reference-free sensitivity analysis when this audit is unavailable.

\subsection{Pilot and finite-block check algorithms}
\label{app:pilot_algorithms}
\label{sec:adaptive_budget}

Algorithm~\ref{alg:adaptive_tuning} records an estimate of the leading log-weight variance $\tau_M^2(\theta)$ in \eqref{eq:stability_budget_condition} and estimates the fixed-state negative probability $q^-_{K,B}(\theta)$ in \eqref{eq:negative_probability_def_intro}, at each candidate budget \((K,B)\),
directly at a
fixed collection of posterior-representative pilot states. Candidate selection uses the sign conditions below and the optional variance cap \(C_\tau(K)\).
All statistics calculated in line 7 are output as stress diagnostics.  The optional
cap was nonbinding in every reported selection.  The pilot construction and
candidate list are fixed before signed weights are generated.  In the
experiments, \(B\) is ordered from small to large and,
within each \(B\), \(K=5,\ldots,10\).  The first-pass convention favors a
smaller block size when two pairs have the same total point count; another
fixed ordering would change only finite-grid constants.
The Gaussian approximation to the sign-screen from
Appendix~\ref{app:negative_sign_probability} can be used to choose the range of
candidate block sizes, but the sign conditions below are based
on estimated signed weights without a Gaussian assumption.

\begin{algorithm}[htbp]
\caption{Ordered finite-grid sign and optional variance screen}
\label{alg:adaptive_tuning}
\begin{algorithmic}[1]
\Require fixed \(\beta_n\); ordered candidate list
\(\mathcal C=((K_1,B_1),\ldots,(K_L,B_L))\); pilot states
\(\mathcal O_{\rm pilot}\); block generator \(G\); repeat count
\(N_{\rm sign}\); sign tolerance \(q_0\); optional variance cap \(C_\tau(K)\in(0,\infty]\).

\State Set the ordered feasible list \(\mathcal F\leftarrow()\).
\For{\((K,B)\) in \(\mathcal C\), in the given order}
    \For{\(\theta\in\mathcal O_{\rm pilot}\)}
        \For{\(r=1,\ldots,N_{\rm sign}\)}
            \State Generate \(K\) independent blocks of size \(B\); compute
            \(\widehat Z_r(\theta)\) and \(S_{K,B,r}^2(\theta)\).
        \EndFor
        \State Record
        \begin{align*}
        \widehat q^-(\theta)&=N_{\rm sign}^{-1}\sum_r
        \mathbf 1\{\widehat Z_r(\theta)<0\},&
        \widehat p_0(\theta)&=N_{\rm sign}^{-1}\sum_r
        \mathbf 1\{\widehat Z_r(\theta)=0\},
        \\
        \widehat\tau^2_{K,B}(\theta)&=
    \frac{\beta_n^2}{K N_{\rm sign}}
    \sum_{r=1}^{N_{\rm sign}}S_{K,B,r}^2(\theta),&
        \widehat r(\theta)&=
        \frac{\sum_r\widehat Z_r(\theta)}
        {\sum_r|\widehat Z_r(\theta)|},
        \end{align*}
    \EndFor
    \If{\(\max_\theta\widehat\tau^2_{K,B}(\theta)\le C_\tau(K)\),
    \(\max_\theta\widehat q^-(\theta)\le q_0\),
    \(\max_\theta\widehat p_0(\theta)=0\), and
    \(\widehat r(\theta)>0\) at every pilot state
    \newline\hspace*{\algorithmicindent}}
    \State Append \((K,B)\) and its diagnostics to \(\mathcal F\).
    \EndIf
\EndFor
\State \Return the ordered feasible list \(\mathcal F\).  If it is empty,
expand the candidate grid or revise the representation.
\end{algorithmic}
\end{algorithm}

Since $S^2_{K,B}$ is unbiased for the variance of a single block
estimate and the $K$ blocks are independent,
$S^2_{K,B}/K$ estimates
$\Var\{\bar\ell_M(\theta,\xi)\}$. The tolerance \(q_0=0.05\) is a practical screen: in the
absence of zero weights it targets a fixed-state raw expected sign of at least
\(0.90\).  Weighted signs are always reported as well because rare large
negative weights can matter more than their frequency.  Positivity of
\(\widehat r\) screens out an estimated sign reversal but does not impose a
quantitative margin away from cancellation.

\subsection{Optional reference audit and final selection}
\label{sec:finite_block_audit_main}

When reliable loss evaluations are available on a finite design, the optional
audit estimates the normalized variation of \(w_M\) and can refine the pilot
choice.  Let
\[
    \mathcal O_{\rm audit}
    =
    \{\theta_1,\ldots,\theta_J\}
\]
be fixed before the candidate configurations are audited.  This design need
not equal \(\mathcal O_{\rm pilot}\): the pilot set records variance and screens sign
behavior at a small collection of representative states, whereas the audit
design carries posterior-mass weights and is chosen to assess variation in
\(w_M\).  The two designs may overlap or coincide when one set serves both
purposes.

Depending on the application, reference losses
\(\ell_{\rm ref}(\theta)\), \(\theta\in\mathcal O_{\rm audit}\), may be
obtained by high-accuracy quadrature or, for small \(J\), by averaging the loss
estimator at a much larger budget.  The resulting values are frozen before
candidates are compared; these calculations are offline and do not enter the
Markov-chain transition.

For a candidate configuration with total block budget \(M=KB\), let
\(N_{\rm audit}\) be the number of replicate weight evaluations at each audit
state.  For each \(r=1,\ldots,N_{\rm audit}\), generate the vector
\(\{\widehat Z_M^{(r)}(\theta_j):j=1,\ldots,J\}\) using a prespecified
coupling across audit states.  These replicate vectors are independent over
\(r\).  Within a replicate, the auxiliaries may be independent across states or
shared as common random numbers, provided that every state has the block laws
and Bessel correction used by the production estimator.  Define
\begin{equation}
    \widehat w_{M,j}
    =
    e^{\beta_n\ell_{\rm ref}(\theta_j)}
    \frac{1}{N_{\rm audit}}
    \sum_{r=1}^{N_{\rm audit}}
    \widehat Z_M^{(r)}(\theta_j),
    \qquad j=1,\ldots,J.
    \label{eq:wM_audit_estimator}
\end{equation}
If \(\ell_{\rm ref}=\ell\), then \(\widehat w_{M,j}\) estimates the
finite-block target factor \(w_M(\theta_j)\).  If instead
\[
    \ell_{\rm ref}(\theta)
    =
    \ell(\theta)+\Delta(\theta),
\]
then it estimates
\[
    e^{\beta_n\Delta(\theta_j)}w_M(\theta_j).
\]
A common additive error in the reference loss cancels when the audited target
is normalized.  State-dependent variation in \(\Delta\) does not, so the
accuracy of a numerical reference must be judged on the \(\beta_n\)-scale.

A set of audit weights \(\omega_j\), \(j=1,\ldots,J\), measures overall
accuracy by weighting audit states according to posterior mass under the
intended target.  This defines the discrete reference posterior.  If the audit
states are approximate posterior draws or equal-posterior-mass quantiles, take
\(\omega_j=1/J\).  For a deterministic audit design with cell or quadrature
weights \(v_1,\ldots,v_J\), let
\begin{equation}
\begin{aligned}
    \widetilde\omega_j
    &=
    v_j\pi_0(\theta_j)
    \exp\{-\beta_n\ell_{\rm ref}(\theta_j)\},\\
    \omega_j
    &=
    \frac{\widetilde\omega_j}
    {\sum_{m=1}^J\widetilde\omega_m}.
\end{aligned}
    \label{eq:audit_posterior_weights}
\end{equation}
On a uniform grid the \(v_j\) are constant and cancel.  On a nonuniform
deterministic design they are retained because they represent posterior mass
by cell.

When \(\widehat w_{M,j}>0\) for every \(j\), the corresponding finite-block posterior masses are
\begin{equation}
    \widehat\pi^{\rm audit}_{M,j}
    =
    \frac{
      \omega_j\widehat w_{M,j}
    }{
      \sum_{m=1}^J\omega_m\widehat w_{M,m}
    },
    \qquad j=1,\ldots,J.
    \label{eq:audited_finite_block_mass}
\end{equation}
Equivalently, writing
\[
    \widehat\mu_M
    =
    \sum_{m=1}^J\omega_m\widehat w_{M,m},
    \qquad
    \widehat{\bar w}_{M,j}
    =
    \frac{\widehat w_{M,j}}{\widehat\mu_M},
\]
we have
\[
    \widehat\pi^{\rm audit}_{M,j}
    =
    \omega_j\widehat{\bar w}_{M,j}.
\]
The finite-state total-variation discrepancy is therefore
\begin{equation}
\begin{aligned}
    \widehat{\operatorname{TV}}_{\rm disc}
    &=
    \frac12
    \sum_{j=1}^J
    \left|
      \widehat\pi^{\rm audit}_{M,j}-\omega_j
    \right|\\
    &=
    \frac12
    \sum_{j=1}^J
    \omega_j
    \left|
      \widehat{\bar w}_{M,j}-1
    \right|.
\end{aligned}
    \label{eq:discrete_tv_audit}
\end{equation}
Thus the audit measures the change in posterior mass induced by the finite-block
factor on the chosen design.

Algorithm~\ref{alg:finite_block_audit} applies this calculation to the
configurations retained by Algorithm~\ref{alg:adaptive_tuning}, in their
prespecified order.  The first candidate for which the estimated factors are
positive and numerically stable, and for which
\(\widehat{\operatorname{TV}}_{\rm disc}\) is below the stated tolerance with
sufficiently small Monte Carlo uncertainty, is used in the production run.
Once selected, \((K,B)\) is fixed for the production run.

In the CRPS calibration, the \(181\) audit states are
equal-posterior-mass quantiles and hence have weight \(1/181\).  In the FTIR
example, the audit uses all \(41\times41\) compatibility-grid states with
normalized reference-posterior weights.  The nine posterior medoids used by
Algorithm~\ref{alg:adaptive_tuning} in that example belong only to the pilot
screen.  A holdout audit, when reported, is performed after selection and does
not alter the chosen configuration.

\begin{algorithm}[htbp]
\caption{Finite-block target check and final candidate choice}
\label{alg:finite_block_audit}
\begin{algorithmic}[1]
\Require ordered sign-feasible list \(\mathcal F\); audit states
\(\theta^{(1)},\ldots,\theta^{(J)}\); reference losses
\(\ell_{\rm ref}(\theta^{(j)})\); weights \(\omega_j\ge0\),
\(\sum_j\omega_j=1\); repeat count \(N_{\rm audit}\); tolerance
\(\epsilon_w\); prespecified cross-state coupling and uncertainty rule
\For{\((K,B)\) in \(\mathcal F\), in order}
    \For{\(j=1,\ldots,J\)}
        \State Generate replicate copies using the prespecified cross-state coupling,
        \(\widehat Z_{K,B}^{(r)}(\theta^{(j)})\),
        \(r=1,\ldots,N_{\rm audit}\), and compute
        \[
        \widehat w_j
        =
        e^{\beta_n\ell_{\rm ref}(\theta^{(j)})}
        \frac1{N_{\rm audit}}\sum_r
        \widehat Z_{K,B}^{(r)}(\theta^{(j)}),
        \qquad
        \widehat r_j
        =
        \frac{\sum_r\widehat Z_{K,B}^{(r)}(\theta^{(j)})}
        {\sum_r|\widehat Z_{K,B}^{(r)}(\theta^{(j)})|}.
        \]
        \State Estimate the Monte Carlo standard error of \(\widehat w_j\).
    \EndFor
    \If{all \(\widehat w_j>0\) and the estimates are numerically stable}
        \State Set \(\bar w=\sum_j\omega_j\widehat w_j\) and compute
        \[
        \widehat{\rm TV}_{\rm disc}
        =
        \frac12\sum_j\omega_j
        \left|\frac{\widehat w_j}{\bar w}-1\right|,
        \qquad
        \widehat{\rm CV}_{\rm disc}
        =
        \frac{\{\sum_j\omega_j(\widehat w_j-\bar w)^2\}^{1/2}}
        {\bar w}.
        \]
        \State Estimate the Monte Carlo uncertainty of
        \(\widehat{\rm TV}_{\rm disc}\) using the prespecified rule.
        \If{\(\widehat{\rm TV}_{\rm disc}<\epsilon_w\), its nominal 95\%
        upper uncertainty limit is at most \(\epsilon_w\), and the
        prespecified precision condition is satisfied}
            \State \Return this \((K,B,M=KB)\), the pointwise factors,
            weighted signs, TV/CV summaries, and changes in the posterior
            summaries of interest.
        \EndIf
    \EndIf
\EndFor
\State \Return ``no audited configuration on the current grid.''
\end{algorithmic}
\end{algorithm}

The uncertainty rule is fixed separately for each experiment.  The original CRPS audit uses 11
independent batches and a normal 95\% upper limit.  The random-phase production
and learning-scale audits use six and four independent batches, respectively,
Student-\(t\) upper limits, and require the standard error to be no more than
one quarter of the absolute decision margin.  The \(g\)-and-\(k\) audit uses
256 batches with a delete-one-batch jackknife standard error and a
Student-\(t\) upper limit.  The original FTIR audit uses adaptive batches of 1000
replicates, a Student-\(t\) upper limit, and the same quarter-margin condition;
the configuration used for production has three batches.  Its pooled factors
do not satisfy the positivity requirement, as detailed in
Appendix~\ref{app:real_ftir_details}.  CRPS and \(g\)-and-\(k\) do not use
the quarter-margin condition.  The nonlinear audit estimates can have appreciable
Monte Carlo bias, as the CRPS integration check below shows.

Table~\ref{tab:screen_audit_comparison} compares the first
Algorithm-2-feasible configuration with the configuration used for production.

\begin{table}[htbp]
    \centering
    \scriptsize
    \setlength{\tabcolsep}{3pt}
        \caption{Effect of the reference audit on block selection.  The
    first configuration is the earliest Algorithm-2-feasible pair in
    the prespecified candidate order.  Audit values in parentheses are
    nominal Monte Carlo standard errors.  Tuples are \((K,B,M)\).
    The FTIR production configuration does not satisfy the audit positivity check.}
    \label{tab:screen_audit_comparison}
    \begin{tabular}{@{}p{0.20\textwidth}p{0.11\textwidth}
                    p{0.18\textwidth}p{0.27\textwidth}
                    p{0.17\textwidth}@{}}
        \toprule
        Example & Blocks & Pilot choice &
        Audit at pilot choice & Production choice \\
        \midrule
        CRPS, \(n=800\) &
        RQMC &
        \((9,16,144)\) &
        TV \(0.0255\;(0.0016)\) &
        \((5,32,160)\) \\

        Random phase \phantom{lll}
        ($\kappa=60$/stress) &
        RQMC &
        \((5,16,80)\) &
        TV \(0.01340\;(0.00334)\) &
        \((9,16,144)\) \\

        \(g\)-and-\(k\) & RQMC & \((5,128,640)\) & TV \(0.0121\;(0.0010)\) & \((7,128,896)\) \\ \(g\)-and-\(k\) & MC & \((8,2048,16384)\) & TV \(0.00212\;(0.00123)\) & unchanged \\ Real FTIR &
        RQMC &
        \((10,8,80)\) &
        normalization not $>0$ &
        \((8,32,256)\), not validated \\

        \bottomrule
    \end{tabular}
\end{table}

The reference audit 
increases the CRPS budget only modestly. In other experiments which we omit for brevity, there was no change. The
random-phase audit increases the budget from \(80\) to \(144\) to obtain the
first candidate satisfying both the TV tolerance and the prespecified
uncertainty criteria.  At the pilot choice, the TV estimate \(0.0134\) was
itself only slightly above the threshold \(0.01\). In the FTIR example,
passing the pilot screen does not resolve the positivity and precision
requirements of the target audit. %

The audit remains reference-relative.  If
\(\ell_{\rm ref}(\theta)=\ell(\theta)+\Delta(\theta)\), variation of
\(\Delta\) over the audit states changes the normalized factor, whereas
a common additive component cancels.  We use
\(\epsilon_w=0.01\) only when the reference calculation is
accurate on the \(\beta_n\)-scale.  This tolerance is deliberately stringent and may be smaller than the Monte Carlo error of routine posterior summaries.  Otherwise we omit the audit and compare the selected configuration with one or more larger budgets, as described in Section~\ref{sec:block_configuration_screen}.
The variance rate is used only to set the candidate budgets.
Closeness to the intended posterior is assessed separately by the reference audit.  %

\subsection{Common implementation conventions}

RQMC blocks use independent randomizations, and the sample variance is computed
across independent blocks at one parameter value.  Fresh randomizations are
used at each proposal; no common random numbers are carried from one MCMC
transition to the next.  This transition-level convention is distinct from the
offline random-phase pilot and audit, which use common randomized point sets
across fixed states within a replicate.  When a proposal is rejected, the current parameter and
the complete auxiliary state are retained.  Signed posterior summaries use
Eq.~\eqref{eq:signed_ratio_estimator}.  Effective sample size is computed
after burn-in using Geyer's initial-monotone-sequence truncation of the sample
autocorrelation.  When total ESS over independent chains is reported,
coordinate ESS values are summed across chains before taking the minimum.  For
the random-phase mean-and-standard-deviation summaries, the coordinate minimum
is instead taken within each chain and then summarized across the five chains.
Wall-clock results are implementation-specific and are not used as complexity
claims.

For the random-phase comparison, the Zig--Zag process uses the
auxiliary-independent affine bound in
Eq.~\eqref{eq:random_phase_affine_bound}.  A fresh phase is drawn at every
candidate event.  RQ-SCBD and Zig--Zag start from the same posterior-scale
state, and both use the common stochastic-evaluation budget described in the
main text.  The label ``exact-bound'' refers to exact no-clipping thinning with
a rigorous bound; it does not assert that the affine bound is tight.

For the \(g\)-and-\(k\) benchmark, RQ-SCBD and MC-SCBD use a
Gaussian random-walk proposal with covariance
\[
    0.70\widehat\Sigma_{\rm ref}
    +\operatorname{diag}(10^{-5},10^{-6}),
\]
whereas the MMD-ABC runs use
\[
    \widehat\Sigma_{\rm ref}
    +\operatorname{diag}(10^{-5},10^{-6}).
\]
Here \(\widehat\Sigma_{\rm ref}\) is the covariance of the frozen pilot reference.
In both cases, the scalar proposal scale is adapted during burn-in from
the frozen pilot-grid mode and fixed thereafter.

The posterior summaries in
Figure~\ref{fig:gandk_posterior_comparison} are sign-corrected for
RQ-SCBD, MC-SCBD, and Russian roulette.  Panel~(a) uses the resulting
signed joint density estimates to draw the \(80\%\) HPD
boundaries.  Panels~(b) and~(c) use a common Gaussian kernel and
bandwidth for the marginal density estimates.  The MMD-ABC tolerance
sweep is reported separately in
Figure~\ref{fig:gandk_abc_tradeoff}.

For the real FTIR comparison in Appendix~\ref{sec:real_ftir}, a proposal covariance and scalar scale are
calibrated in a separate run on the exhaustive reference target and then frozen.
The same proposal is used for RQ-SCBD, positive RQMC pseudo-marginal MCMC, and
iid-MC pseudo-marginal MCMC.  Each method uses five independent chains and five
million pair-score lookups per chain.

\subsection{Experiment-specific settings}
\label{app:experiment_specific_settings}

The optional variance cap is disabled in the CRPS, random-phase, and FTIR
experiments.  The \(g\)-and-\(k\) screen uses \(C_\tau=3\), which is
nonbinding at the selected configurations.

\paragraph*{CRPS calibration.}
The master seed is \(20260622\).  The sample-size grid is
\(n=100,200,\ldots,6400\), and the threshold figure uses \(n=800\).  The
variance fit uses \(1200\) independent randomizations at each block size and is
reported over \(B=16,\ldots,1024\).  The sign pilot uses \(1800\) repetitions
at the seven exact-posterior quantiles
\(0.08,0.22,\ldots,0.92\).  Algorithm~\ref{alg:finite_block_audit} uses
\(181\) equally weighted posterior-quantile states and \(2200\) repetitions;
the independent holdout uses \(181\) states and \(1500\) repetitions.  The
candidate grid is \(B=2,4,\ldots,2048\) and \(K=5,\ldots,10\).

\paragraph*{Random phase.}
The output dimension is \(10\), \(\beta_n=20000\),
\(\epsilon=0.04\), \(\eta=0.8\), and \(\kappa=12,20,60\).  The matrices
\(\mathsf A\), \(\mathsf C\), and \(W\), together with the fixed vector \(y\),
are stored in the reproduction output.  The candidate grid is
\(B=4,8,\ldots,256\), \(K=5,\ldots,10\), ordered first by \(B\) and then by
\(K\).  At \(\kappa=60\), the production pilot uses six fixed states and 120
sign repetitions.  The target audit uses 128 equally weighted scrambled-Sobol
draws from the exact Gaussian target and 600 weight repetitions in six
independent batches of 100.  Within a candidate and repetition, the same
independently scrambled block point sets are evaluated at every fixed state;
repetitions and audit batches are independent.  The selected production
configuration is \((K,B,M)=(9,16,144)\).

The separate learning-scale calculation uses 200 sign repetitions, 32
beta-specific equally weighted target states, and 220 audit repetitions in four
independent batches of 55.  Its \(\beta_n=20000\) choice
\((K,B,M)=(5,32,160)\) is therefore not a repetition of the production audit.

Each method receives \(400000\) stochastic evaluations at each frequency and
uses five independent chain seeds.  Both start at
\((3\sigma_1,-3\sigma_2)\).  RQ-SCBD makes 2777 proposals (399888 loss
evaluations), burns the first 694, and retains 2083 proposal states; its
initialization weight is outside the reported budget.  Zig--Zag has no velocity
refreshments, processes 400000 candidates, burns the first 100000, and records
20000 equally spaced midpoint samples in continuous time over the retained
trajectory.  A fresh phase is drawn at every candidate event and no thinning
probability is clipped.

\paragraph*{\(g\)-and-\(k\) MMD.}
The run uses seed \(42\), \(\n=500\), generating parameter
\((0.10,0.05)\), uniform prior \([0,5]\times[0,2]\), bandwidth
\(h=3.9494\), and \(\beta_n=2500\).  The pilot reference uses a
\(54\times46\) grid and \(4096\) pair evaluations per grid point; posterior
comparisons use the refined full-prior reference described below.  The
SCBD candidate grid is \(B=4,8,\ldots,16384\) and
\(K=5,\ldots,10\).  The variance calculation uses \(64\)
randomizations at each displayed block size.  Each SCBD result
pools five chains of \(7000\) iterations with \(1000\) burn-in
iterations; each ABC result uses five chains of \(1400\) iterations with
\(350\) burn-in iterations.  Candidate selection uses nine posterior-weighted medoids, with \(300\) sign-screen repetitions per state.  The target audit uses eight independent scrambled references of \(2^{16}\) pairs per state and \(256\) independent batches of \(50\) repetitions.  In the MC production code, the fixed one-dimensional function \(m_y\) is evaluated by a deterministic cubic spline on a grid of spacing \(0.01\); a run-time quarter-grid check gives maximum absolute error \(2.74\times10^{-13}\), or \(1.37\times10^{-9}\) after both terms are scaled by \(\beta_n\).  The Russian-roulette result also pools five chains of
\(1400\) iterations with \(350\) burn-in iterations; its pilot and
expected-cost calculation are given in
Appendix~\ref{app:gandk_rr}.

\paragraph*{Real FTIR spectra.}
The processed data contain \(60\) specimens in three groups and
\(448\) spectral channels.  The fixed model settings are
\(\rho=22\), \(\eta=0.018\), and \(\gamma=1.5\), with
\(\beta_n=60\).  Algorithm~\ref{alg:adaptive_tuning} uses the
nine-state posterior-medoid pilot in
Appendix~\ref{app:real_ftir_details}.  The target audit
uses all \(1681\) states of the uniform \(41\times41\) grid on
\([-2.4,2.4]^2\), with the weights in
Eq.~\eqref{eq:ftir_audit_weights}.  The candidate grid is
\(B=8,16,\ldots,512\), \(K=5,\ldots,10\).  The audit is computed in
independent batches of \(1000\) repetitions per state, with additional
batches only when the decision is unresolved; the production
configuration used \(3000\) repetitions per state.  A separate fixed-size
check below retains this grid and does not change the production chains.

\subsection{Finite-block checks for examples in Sections~\ref{sec:crps_calibration}, \ref{sec:random_phase_simulator} and \ref{sec:g_and_k_benchmark}}

For the CRPS calibration, the selected configurations have no negative weights
on the selection states over the displayed \(n\)-grid.  The largest holdout
negative frequency is \(6.7\times10^{-4}\), and the posterior-weighted negative
probability is \(3.7\times10^{-6}\).  Closed-form score evaluations agree with
numerical quadrature to machine precision, so the finite-block checks do not
inherit a reference-loss error $\Delta(\theta)$ above.

For the random-phase example, the selected stress-regime configuration has
\(\widehat{\rm TV}_{\rm disc}=0.00449\) with Monte Carlo standard error
\(0.00085\).  The production RQ-SCBD chains contain no negative or zero retained
weights, so their pooled sign denominators equal one.  This does not make the
finite-block target exact; its agreement with the intended target is supplied
by the separate Algorithm~\ref{alg:finite_block_audit} calculation.

For the \(g\)-and-\(k\) example, the reference is numerical and the check is
reference-relative.  Table~\ref{tab:gandk_audit} gives additional diagnostics for the selected
configurations.

\begin{table}[htbp]
    \centering
    \scriptsize
    \setlength{\tabcolsep}{3pt}
    \caption{Reference-relative finite-block check for the
    \(g\)-and-\(k\) example.  The nine audit states are posterior-weighted
    medoids, and the TV and CV values come from line 7 of Algorithm~\ref{alg:finite_block_audit} and use independent high-precision    reference-loss evaluations.}
    \label{tab:gandk_audit}
    \begin{tabular}{@{}lrrr@{}}
        \toprule
        Method & audit TV & CV & posterior-mean displacement \((\Delta g,\Delta k)\) \\
        \midrule
        MC-SCBD & 0.00212 & 0.00521 & \((0.00025,0.00019)\) \\
        RQ-SCBD & 0.00765 & 0.01758 & \((0.00060,0.00076)\) \\
        \bottomrule
    \end{tabular}
\end{table}

The batch-jackknife 95\% intervals for audit TV are \([0.00624,0.00905]\) for RQ-SCBD and \([0,0.00454]\) for MC-SCBD. Russian roulette is not included in this table because it does not use
the finite-\((K,B)\) SCBD construction or its target factor \(w_M\).
Its sign and cost diagnostics are reported in
Appendix~\ref{app:gandk_rr}.

\subsection{Real FTIR Replicate Spectra}
\label{sec:real_ftir}

Randomized
Sobol points give such strong variance reduction that the effect of debiasing, the effect of moving from P-MCMC with RQ-block integration converging as in \eqref{eq:RQ-u-TV-a-app} to RQ-SCBD converging as in \eqref{eq:RQ-w-TV-a-app}, may be hard to see. The final example uses observed duplicate spectra from the MIRFreshMeats
archive \citep{aljowder1997midinfrared,quadramFTIRmeat} to illustrate a setting where the
debiasing correction is important.  The loss is smooth in the model
parameters, but one stochastic evaluation is obtained by looking up a
specimen--channel-pair contribution from a large finite population.  The
resulting integrand is discontinuous in the randomized indices, so randomized
Sobol points do not by themselves remove the bias created by exponentiation.  This low-dimensional example uses the exhaustive grid only for offline pilot construction and validation; the production transitions use stochastic lookup estimates.

\subsubsection{Data, working model, and stochastic loss}
\label{sec:real_ftir_model}

The data contain duplicate spectra from \(60\) independent specimens in three
groups, with \(J=448\) channels per spectrum.  Let \(Y_{irj}\) denote channel
\(j\) of replicate \(r\in\{1,2\}\) for specimen \(i\).  We form the processed
replicate difference
\begin{equation}
    D_{ij}
    =
    \frac{(Y_{i1j}-Y_{i2j})/\sqrt2-\bar D_i}{s_2},
    \qquad
    \bar D_i=\frac1J\sum_{j=1}^J
    \frac{Y_{i1j}-Y_{i2j}}{\sqrt2},
    \label{eq:ftir_difference}
\end{equation}
where \(s_2\) is the replicate-difference standard deviation in the middle
group.  Differencing removes the specimen-specific mean spectrum; centering
and scaling make the analysis concern relative repeatability rather than the
spectral level.

Denote by $\Sigma_c$ the $J\times J$ cross-channel covariance matrix for specimen in group \(c\in\{1,2,3\}\). The working model for $D_i=(D_{i,1},\dots,D_{i,J})$ is
\[
    D_i\mid c(i)=c\sim N_J(0,\Sigma_c),\qquad
    (\Sigma_c)_{jk}
    =
    \sigma_c^2\left[
      (1-\eta)\exp\{-|j-k|/\rho\}
      +\eta\,\mathbf 1\{j=k\}
    \right].
\]
The empirical correlation fit gives \(\rho=22\) and \(\eta=0.018\).  The scale
of the middle group is fixed at \(\sigma_2=1\) by the preprocessing and inference concerns the remaining two scales,
\[
    \theta=(\log\sigma_1,\log\sigma_3),
    \qquad
    \theta_j\stackrel{\rm ind}{\sim}N(0,0.8^2).
\]
These correspond to the chicken and turkey groups in
the archive.

The loss for GBI will be based on the density-power score of \citet{basu1998densitypower} with power \(\gamma=1.5\),
\begin{equation}
    s_\gamma(f,z)
    =
    \frac{1}{1+\gamma}\int f(x)^{1+\gamma}\,dx
    -\frac1\gamma f(z)^\gamma.
    \label{eq:ftir_dpd_score}
\end{equation}
This score would usually take the likelihood as its $f$-argument. We will modify the likelihood to replace the joint $J$-dimensional $D$-dependence with marginal pairwise dependence. Let
\(f_{\theta,c,h}\) be the bivariate normal density of
\((D_{ij},D_{i,j+h})\) under the working model.  The contribution from specimen
\(i\), lag \(h\), and starting channel \(j\) is
\[
    L_{i,j,h}(\theta)
    =
    s_\gamma\!\left(
      f_{\theta,c(i),h},
      (D_{ij},D_{i,j+h})
    \right),
\]
where the integral in \eqref{eq:ftir_dpd_score} is analytic for a bivariate normal density.  The exact loss for GBI is defined to be the average loss over all unordered channel pairs,
\begin{equation}
    \ell_{\rm FTIR}(\theta)
    =
    \frac{1}{N_{\rm pair}}
    \sum_{i=1}^{60}\sum_{h=1}^{J-1}\sum_{j=1}^{J-h}
    L_{i,j,h}(\theta),
    \qquad
    N_{\rm pair}=60\binom{J}{2}=6{,}007{,}680.
    \label{eq:ftir_full_loss}
\end{equation}
Distinct pairs within a spectrum often reuse one channel measurement, and all
pair contributions are computed from the same specimen-level spectrum.  The
learning scale is therefore
the number of independent specimens,
\(\beta_n=60\), rather than \(N_{\rm pair}\).

A stochastic evaluation uses three uniform random variables \(U\sim{\rm Unif}([0,1]^3)\) to select a specimen uniformly, a lag with probability
proportional to the number \(J-h\) of pairs at that lag, and then a starting
channel uniformly among \(1,\ldots,J-h\). This gives
\begin{equation}
    L_{\rm FTIR}(\theta,U)
    =
    L_{I(U),J_0(U),H(U)}(\theta),
    \qquad
    \E_U\{L_{\rm FTIR}(\theta,U)\}
    =
    \ell_{\rm FTIR}(\theta).
    \label{eq:ftir_lookup}
\end{equation}
For fixed indices the score is smooth in \(\theta\). However, RQMC works in the space of the $U$-variables, and as a function of \(U\),
however, Eq.~\eqref{eq:ftir_lookup} is piecewise constant with jumps at index
boundaries.  This is the computational feature studied below.

\subsubsection{Block selection and audit checks}
\label{sec:real_ftir_selection}

The exhaustive loss in Eq.~\eqref{eq:ftir_full_loss} is evaluated once
on a uniform \(41\times41\) grid over \([-2.4,2.4]^2\).  The reference
grid enters the two selection stages in different ways.

Algorithm~\ref{alg:adaptive_tuning} uses nine
posterior-representative pilot states: the grid mode and eight
deterministic posterior-weighted medoids.  The medoids are constructed
within the smallest grid set containing \(95\%\) of the reference
posterior mass, using weighted marginal-rank coordinates, and are frozen
before any candidate weights are generated.

Algorithm~\ref{alg:finite_block_audit}, by contrast, uses all
\(1681\) grid states.  Since the grid is uniform, their posterior-mass
weights are
\begin{equation}
    \omega_j
    =
    \frac{
      \pi_0(\theta_j)
      \exp\{-\beta_n\ell_{\rm FTIR}(\theta_j)\}
    }{
      \displaystyle
      \sum_{m=1}^{1681}
      \pi_0(\theta_m)
      \exp\{-\beta_n\ell_{\rm FTIR}(\theta_m)\}
    }.
    \label{eq:ftir_audit_weights}
\end{equation}
Thus the nine medoids define the pilot screen, whereas the full grid
defines the target audit.  Appendix~\ref{app:real_ftir_details} gives
the pilot coordinates and the remaining implementation details.

We retain the original configuration
\[
    (K,B,M)=(8,32,256),
\]
as a fixed budget for comparison.  Its original audit gave an empirical
discrepancy of \(0.00316\), but included
\(275\) non-positive pooled factor estimates among the \(1681\) states,
so it does not satisfy the positivity requirement of Algorithm~\ref{alg:finite_block_audit}.  Factor estimates at states with little reference posterior mass can be unstable.  The pointwise positivity check may therefore fail even when finite-chain summaries are close to the reference; this agreement does not establish global target accuracy.

\subsubsection{Equal-budget comparison and the Bessel correction}
\label{sec:ftir_bessel_comparison}

We compare eight independently scrambled blocks with the Bessel correction
against one independently scrambled net without correction, at
\(M\in\{128,256,512,1024,2048\}\).  Each method and budget uses \(8192\)
independent weights in eight batches.  Write \(\widehat D_M\) for the
half-\(L^1\) difference, computed as in Eq.~\eqref{eq:discrete_tv_audit},
between normalized estimates of the expected weights and reference masses on an
\(81\times81\) grid over \([-2.4,2.4]^2\). As some corrected estimates are
negative, \(\widehat D_M\) is a signed discrepancy rather than posterior TV.

Figure~\ref{fig:ftir_posterior_comparison}(c) shows smaller discrepancies for
RQ-SCBD at every budget.  At \(M=256\), \(\widehat D_M\) is \(0.00933\),
compared with \(0.08263\) for one-net RQMC; at \(M=512\), the values are
\(0.00418\) and \(0.03165\).  Independent repetitions give similar results
(Table~\ref{tab:ftir_budget_diagnostics}), as do the original
\(41\times41\) grid and an \(81\times81\) grid over \([-4,4]^2\).

To isolate the Bessel factor, we also evaluate corrected and uncorrected
weights on the same eight blocks.  Applying \(F_{\beta_n}\) reduces
\(\widehat D_M\) from \(0.1533\) to \(0.01118\) at \(M=256\), and
from \(0.04865\) to \(0.00222\) at \(M=512\), with the same pattern in
independent repetitions.  Thus the correction reduces the observed
distortion without changing the underlying loss estimates.  The single-net
comparison separately allows the uncorrected method to use all \(M\) points
in one net.

Rare combinations of low block losses still make some edge-state estimates
unstable, as a separate calculation over block combinations confirms.
We therefore use these estimates as finite-grid diagnostics.  The larger
budgets are sensitivity checks, not new selections under
Algorithm~\ref{alg:finite_block_audit}.

\subsubsection{Same-budget pseudo-marginal comparison}
\label{sec:real_ftir_comparison}

We compare three retained-auxiliary pseudo-marginal chains at the common
per-proposal budget \(M=256\).

\begin{enumerate}
    \item RQ-SCBD uses \(K=8\) independently scrambled blocks of size \(32\)
    and the Bessel correction.
    \item Positive RQMC pseudo-marginal MCMC uses one independently scrambled
    \(256\)-point net and no Bessel correction.%
    \item Naive iid-MC pseudo-marginal MCMC replaces the scrambled net by
    \(256\) independent lookup draws and directly exponentiates their mean.
\end{enumerate}

This design compares the complete constructions at a common total point budget.
Because items 1 and 2 use different block randomizations, their comparison does
not isolate the Bessel factor alone.  Items 2 and 3 compare scrambled and iid
inputs without correction; MC-SCBD is examined in Section~\ref{sec:g_and_k_benchmark}.  The two positive methods are exact for their own implicit finite-\(M\) targets,
not for the intended posterior.  On rejection they retain both the parameter
and the current auxiliary randomization.  All three methods use the same
Gaussian random-walk proposal, calibrated on the reference target before the
production runs.  The two RQ methods use the same five initial states.  Each method uses five chains, each assigned five million
proposal pair-score lookups.

Table~\ref{tab:ftir_chain_comparison} reports the pooled means and Monte Carlo
standard errors for the sign-reweighted ratios.  The reference moments are
computed over the full prior support, as described in
Appendix~\ref{app:real_ftir_details}.  At \(M=256\), RQ-SCBD is closer to
the reference, especially in the first coordinate.  There is little sign
cancellation in the recorded chains.  A fixed \(M=512\) sensitivity check is
reported in Appendix~\ref{app:real_ftir_details}.

\begin{table}[htbp]
    \centering
    \small
    \caption{Five-chain comparison at \(M=256\), with Monte Carlo standard errors in parentheses.  Each chain has the same proposal lookup budget; \(\widehat R_M\) is the retained mean sign.        }
    \label{tab:ftir_chain_comparison}
    \begin{tabular}{@{}lrrrr@{}}
        \toprule
        Method & \(M\) & \(\log\sigma_1\) & \(\log\sigma_3\) & \(\widehat R_M\) \\
        \midrule
        Reference & --- & \(-0.4504\) & \(0.0055\) & --- \\
        RQ-SCBD & 256 & \(-0.4262\;(0.0096)\) & \(0.0047\;(0.0084)\) & 0.9963 \\
        One-net RQMC & 256 & \(-0.5468\;(0.0109)\) & \(-0.0156\;(0.0090)\) & 1 \\
         
        \bottomrule
    \end{tabular}
\end{table}

\begin{figure}[htbp]
    \centering
    \includegraphics[width=\textwidth]{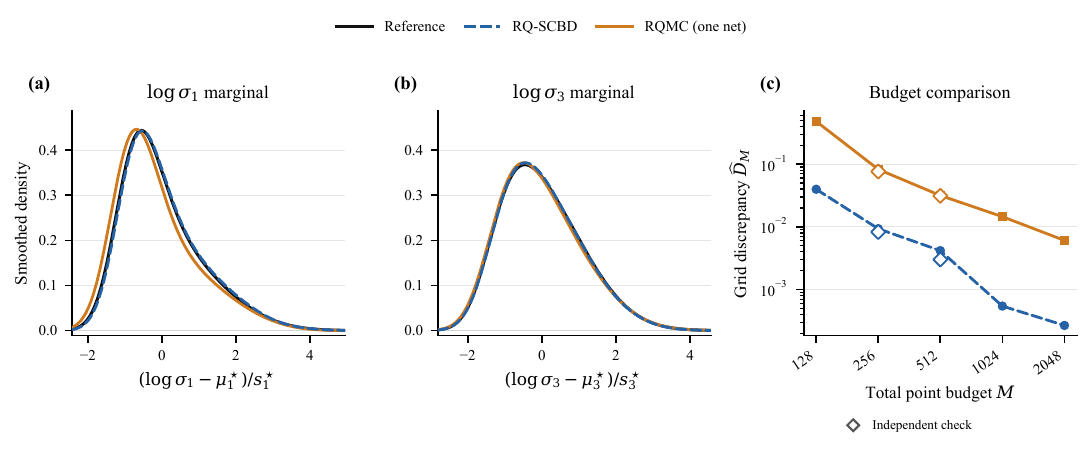}
    \caption{Real FTIR comparison.  Panels (a) and (b) show smoothed marginal
    densities at \(M=256\), using reference means \(\mu_j^\star\) and
    standard deviations \(s_j^\star\) to center and scale each coordinate;
    RQ-SCBD uses sign-reweighted samples.  Panel (c) shows the estimated signed discrepancy \(\widehat D_M\)
    against the point budget.  Open diamonds show the independent
    repetitions at \(M=256,512\); Table~\ref{tab:ftir_budget_diagnostics}
    gives the Monte Carlo standard errors.  }
    \label{fig:ftir_posterior_comparison}
\end{figure}

For comparison, naive iid-MC at \(M=256\) has a full-space TV error greater
than \(0.999999\); Appendix~\ref{app:real_ftir_details} gives the calculation.
Its chains did not reach this rare-event-dominated target and are omitted
from Figure~\ref{fig:ftir_posterior_comparison}.  The RQ comparisons instead
show that Bessel correction can improve on variance reduction alone, within
the limits of the numerical checks above.

\subsection{Real FTIR computational details}
\label{app:real_ftir_details}

\subsubsection{Posterior-representative pilot}

The pilot reference is computed by exhaustive evaluation of
Eq.~\eqref{eq:ftir_full_loss} on the \(41\times41\) compatibility grid.
Algorithm~\ref{alg:adaptive_tuning} uses a deterministic nine-state pilot, $|\mathcal O_{\rm pilot}|=9$. Those points are defined as follows. First, the 0.95 posterior-mass region is defined as the smallest collection of grid states with cumulative mass at least 0.95. This region contains
\(460\) grids and has cumulative mass \(0.95008\). Each coordinate is then transformed by its weighted marginal mid-CDF. The grid mode is fixed as one medoid, and eight further medoids are defined by minimising the posterior-weighted within-cluster squared distance in the transformed coordinates. The construction is completed and hashed before any candidate weights are generated.

Candidates are ordered by increasing \(B\), then by
\(K\in\{5,\ldots,10\}\), giving \(35\) pilot-feasible configurations.
The first, \((K,B,M)=(10,8,80)\), has a non-positive
posterior-weighted signed normalization despite passing the pilot screen.

\begin{table}[htbp]
    \centering
    \small
    \caption{Frozen Algorithm~\ref{alg:adaptive_tuning} pilot for the real
    FTIR example.  Coordinates are \((\log\sigma_1,\log\sigma_3)\).}
    \label{tab:ftir_pilot_states}
    \begin{tabular}{crr}
        \toprule
        Pilot state & \(\log\sigma_1\) & \(\log\sigma_3\) \\
        \midrule
        mode & -0.96 & -0.48 \\
        2 & -1.08 & -0.72 \\
        3 & -1.08 & 0.48 \\
        4 & -0.84 & 0.00 \\
        5 & -0.60 & -0.60 \\
        6 & -0.60 & 0.72 \\
        7 & -0.12 & -0.12 \\
        8 & 0.12 & 0.48 \\
        9 & 0.24 & -0.60 \\
        \bottomrule
    \end{tabular}
\end{table}

This pilot represents posterior mass rather than a covariance ellipse or a
credible-set boundary.  The latter constructions are inappropriate for the
strongly skewed first marginal: an ellipse based on the global covariance can
place a nominal posterior-scale point at essentially zero posterior mass.
For \((K,B)=(8,32)\), deterministic seven- and thirteen-medoid variants give
maximum negative frequencies \(0\) and \(0.002\), respectively, and both retain
the configuration as sign feasible.

\subsubsection{Reference and weight calculations}

The intended reference separates by group.  For the means and standard
deviations, we integrate each marginal under the original Gaussian prior
using composite Gauss--Legendre quadrature over \([-12,8]\), with \(16\)
and \(32\) nodes per interval.  The two orders agree to the precision
reported here; bounds on the omitted reference mass are below
\(2\times10^{-23}\).  The reference mean is \((-0.450410,0.005494)\), with
standard deviations \((0.664587,0.665381)\).  This calculation does not
change the frozen grid pilot.
The marginal curves use the original Gaussian bandwidths, \(0.4567\) and
\(0.4616\) in reference-standard-deviation units, for both methods and the
reference.

The original Algorithm~\ref{alg:finite_block_audit} calculation uses all
\(1681\) compatibility-grid states and their normalized reference posterior
weights.  Independent random streams are used for the sign screen,
each candidate audit, and the holdout.  The original \(K=8,B=32\) calculation
uses \(3000\) repetitions per state.

An earlier independent holdout gave a signed discrepancy of \(0.01555\)
(reported MCSE \(0.00452\)) and did not resolve the \(0.01\) threshold.

The budget comparison uses independent streams for each method, budget and
batch, with \(1024\) repetitions per batch.  The separate repetitions at
\(M=256,512\) use four new batches.  All grid states are retained, and joint
weights are accumulated in signed-log form before normalization.  The
uncertainty summaries use the eight independent batches.  The one-net estimate
at \(M=128\) is sensitive to a few large weights, while the small
RQ-SCBD discrepancies at \(M=1024,2048\) are close to the numerical
uncertainty, so no rate is fitted.  Table~\ref{tab:ftir_budget_diagnostics}
reports the estimates and their batch-jackknife standard errors.

\begin{table}[htbp]
    \centering
    \small
    \caption{FTIR grid discrepancies \(\widehat D_M\) for the complete
    constructions, with nominal Monte Carlo standard errors in parentheses.}
    \label{tab:ftir_budget_diagnostics}
    \begin{tabular}{lrrr}
        \toprule
        Calculation & \(M\) & RQ-SCBD & One-net RQMC \\
        \midrule
        Primary & 128  & \(0.03978\;(0.00481)\) & \(0.4809\;(0.0649)\) \\
        Primary & 256  & \(0.00933\;(0.00127)\) & \(0.08263\;(0.00319)\) \\
        Primary & 512  & \(0.00418\;(0.00127)\) & \(0.03165\;(0.00115)\) \\
        Primary & 1024 & \(0.000545\;(0.000286)\) & \(0.01454\;(0.00063)\) \\
        Primary & 2048 & \(0.000268\;(0.000358)\) & \(0.00608\;(0.00031)\) \\
        Independent & 256 & \(0.00830\;(0.00168)\) & \(0.07644\;(0.00227)\) \\
        Independent & 512 & \(0.00302\;(0.00158)\) & \(0.03124\;(0.00197)\) \\
        \bottomrule
    \end{tabular}
\end{table}

The same-block comparison in Section~\ref{sec:ftir_bessel_comparison} uses
\(16384\) repetitions with identical blocks for the corrected and uncorrected
weights.  A further \(8192\) independent repetitions give reductions from
\(0.1445\) to \(0.01050\) at \(M=256\) and from \(0.05023\) to
\(0.00170\) at \(M=512\).

The edge check uses independent pools of block losses at \((-1.68,-0.48)\),
the three states below, and \((-0.48,0)\).  Averaging over independent
combinations of these pools preserves the expected corrected weight and
makes rare combinations easier to examine.  Eight independent batches each
contain \(4096\) loss estimates per block, at \(B=32,64\).  At the three states
\((-1.56,-2.04)\), \((-1.44,-2.04)\), and \((-1.32,-2.04)\), the
calculation is dominated by very few low block losses.  At \(B=32\), the uncertainty
between pools is much larger than the error in averaging their combinations;
the ordinary batch intervals do not give a reliable sign determination there.

\subsubsection{Production chains and implicit-target checks}

Both RQ methods use five chains and five million proposal pair-score lookups per chain.
The proposal is calibrated on the reference posterior before the method-specific
runs and then fixed.  The RQ curves in
Figure~\ref{fig:ftir_posterior_comparison} are finite-run comparisons with
the reference.  At \(M=256\), each chain has \(19531\) proposals, discards
\(4882\), and retains \(14649\) states.  At \(M=512\), the corresponding
counts are \(9765\), \(2441\), and \(7324\).  Each initialization costs
one additional weight.  Monte Carlo standard errors are calculated from
the long-run variance of \(s_t(\theta_t-\widehat\mu)\), divided by the
squared mean sign, using the initial monotone sequence estimate and then
combining independent chains.

At \(M=512\), the pooled means of \(\log\sigma_1\) and \(\log\sigma_3\)
are \(-0.4730\;(0.0127)\) and
\(-0.0021\;(0.0113)\) for RQ-SCBD, and \(-0.4895\;(0.0120)\) and
\(0.0066\;(0.0118)\) for one-net RQMC; parentheses give Monte Carlo standard errors.
Both retained mean signs are one.  The difference between the methods is
small relative to the Monte Carlo uncertainty.

Equal point budgets do not imply equal running times.  At \(M=256\), the median time per
weight in five fixed rounds of \(64\) weights at each pilot state is \(1.87\) ms for RQ-SCBD and
\(0.26\) ms for one-net RQMC.  These pair-score lookups are cheap, whereas
RQ-SCBD constructs eight independent scrambles rather than one.  This
implementation therefore has no wall-clock advantage in this comparison.

For iid lookup draws,
\[
    \E\!\left[
      e^{-\beta_n\{\widehat\ell_M(\theta)-\ell(\theta)\}}
    \right]
    =
    \left[
      \E\!\left\{
        e^{-\beta_n\{L(\theta,U)-\ell(\theta)\}/M}
      \right\}
    \right]^M,
\]
so the corresponding factor is available without Monte Carlo error at each
state.  This does not make a central-grid normalization reliable: a rare
lookup can give a very large negative score at a remote parameter value.

For iid draws, write the unnormalized target as
\[
 q_M(\theta)=\pi_0(\theta)
 \left\{\frac1{N_{\rm pair}}\sum_{i,h,j}
 e^{-60L_{i,j,h}(\theta)/M}\right\}^{M},
 \qquad Z_M=\int q_M(\theta)\,d\theta.
\]
Every observed pair has a positive quadratic form, so its score has a
finite minimum and \(Z_M<\infty\).
One chicken pair has score below \(-1.35678\times10^8\) throughout
\(I=[-7.654887783-10^{-6},-7.654887783+10^{-6}]\) in \(\log\sigma_1\).
Keeping only this record in the sum gives
\[
 Z_{256}\ge
 \pi_0(I\times\mathbb R)N_{\rm pair}^{-256}
 e^{60(1.35678\times10^8)} > e^{8.14\times10^9}.
\]
On \(C=[-4,4]^2\), a uniform lower bound on the individual scores instead
gives \(\int_C q_{256}<e^{1.546\times10^6}\).
The intended target separates by group.  Its marginal normalizers are
bounded below by partitioning \([-2,2]\) into \(100\) intervals and bounding
each score above by its larger endpoint value; each score has a single
minimum.  In the negative tail below \(-4\), the group mean scores are
nonnegative.  In the positive tail above \(4\), their negative parts are
bounded by a constant times \(e^{-3t}\).  Combining these bounds with
Gaussian prior tail probabilities gives
\(\pi_n(C^c)<7.55\times10^{-7}\).
Numerical evaluation of these bounds therefore yields
\[
 \|q_{256}/Z_{256}-\pi_n\|_{\rm TV}
 \ge \pi_n(C)-Z_{256}^{-1}\int_Cq_{256}>0.999999.
\]
This is a bias of the implicit iid target, not a conclusion drawn from
chain mixing.  It supplies no corresponding accuracy bound for RQ-SCBD.

\subsection{Sensitivity test for the CRPS Calibration in Section~\ref{sec:crps_calibration}} \label{app:CRPS}

\begin{table}[htbp]
    \centering
    \small
    \caption{Threshold sensitivity for the CRPS calibration at \(n=800\).
    Entries are the original Monte Carlo audit estimates, with nominal standard
    errors in parentheses.  The holdout audit uses an
    independent posterior-quantile grid and independent randomizations; it is a
    non-selecting check and is not expected to be systematically larger or
    smaller than the selection audit.\\}
    \label{tab:crps_tolerance}
    \begin{tabular}{ccrrrrcc}
        \toprule
        \(q_0\) & \(\epsilon_w\) & \(K\) & \(B\) & \(M\) &
        selection TV & holdout TV \\
        \midrule
        0.01 & 0.01 & 5 & 32 & 160 & \(0.0067\;(0.0004)\) & \(0.0089\;(0.0005)\) \\
        0.05 & 0.01 & 5 & 32 & 160 & \(0.0067\;(0.0004)\) & \(0.0089\;(0.0005)\) \\
        0.05 & 0.03 & 9 & 16 & 144 & \(0.0255\;(0.0016)\) & \(0.0306\;(0.0016)\) \\
        0.20 & 0.03 & 9 & 16 & 144 & \(0.0255\;(0.0016)\) & \(0.0306\;(0.0016)\) \\
        \bottomrule
    \end{tabular}
\end{table}

Table~\ref{tab:crps_tolerance} shows that changing \(q_0\) leaves the choice
unchanged at each displayed \(\epsilon_w\).  Relaxing
\(\epsilon_w\) from \(0.01\) to \(0.03\) reduces the selected budget from
\(160\) to \(144\), but the independent TV estimate rises to approximately
\(0.031\).  Numerical integration instead gives target TV \(0.000386\)
for \((5,32)\) and \(0.00765\) for \((9,16)\).  Thus the audit noise
overstates the discrepancy in this example.  Figure~\ref{fig:crps_threshold_effect}
uses the integrated factors, not the noisy audit estimates.

The integration uses the block moment-generating function and the spherical
representation of the Bessel factor.  Increasing the loss-integration and
sphere resolutions and the posterior quadrature from \(65\) to \(129\) nodes
changes the two TV values by \(1.4\times10^{-7}\) and \(4.1\times10^{-6}\),
respectively.  %

\begin{figure}[htb]
    \centering
    \includegraphics[width=\textwidth]{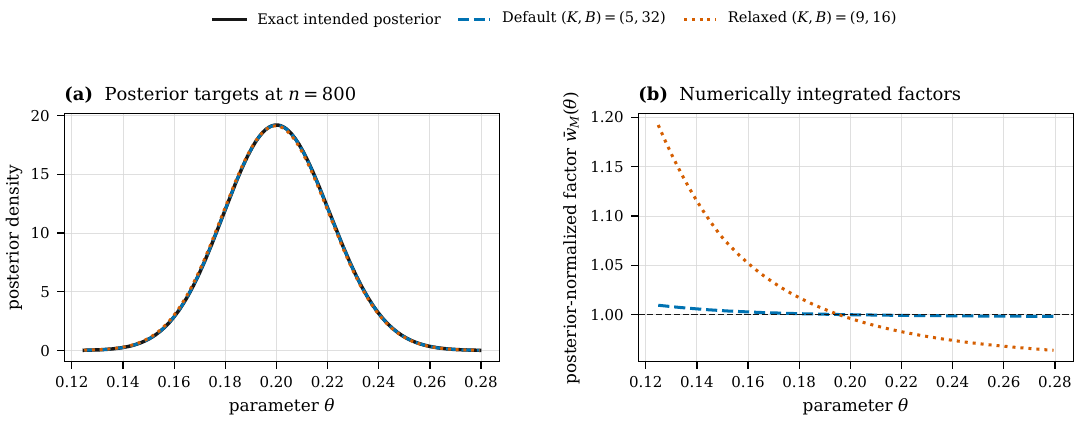}
    \caption{This figure relates to Appendix~\ref{app:CRPS}. Effect of the selection thresholds in the CRPS calibration at
    \(n=800\).  Panel (a) compares the intended posterior with the finite-block
    targets selected under the four pre-specified threshold pairs, computed by
    numerical integration.  Panel (b) plots the associated posterior-normalized
    factors.  Changing
    \(q_0\) leaves the choice unchanged at each displayed \(\epsilon_w\); the looser
    value \(\epsilon_w=0.03\) admits a visibly less flat factor.}
    \label{fig:crps_threshold_effect}
\end{figure}

\subsection{Additional details for Random-Phase example in Section \ref{sec:random_phase_model}}

\subsubsection{Model details}\label{random_phase2_model}

Repeating the model setup to be self-contained, let \(y_1,\ldots,y_n\in\mathbb R^{d_y}\) be observed vectors and write \(\overline y_n\) for their sample mean, with
\(d_y=10\).  For \(\theta\in\mathbb R^2\) and $i=1,\dots,d_y$, fixed vectors
\(a_i,c_i\in\mathbb R^2\), fixed  amplitude $\epsilon>0$ and frequency \(\kappa\) and independent phases
\(\phi_i\sim{\rm Unif}(0,1)\), define the simulator
output for one component of $Y=(Y_1,\dots,Y_{d_y})$ as
\begin{equation}
    Y_i(\theta,\phi_i)
    =
    a_i^\top\theta
    +
    \epsilon\,w_\eta(\kappa c_i^\top\theta+2\pi\phi_i),
    \qquad i=1,\ldots,d_y, 
\end{equation}
where
\begin{equation}
    w_\eta(t)
    =
    \frac{\exp\{\eta\cos t\}-I_0(\eta)}
    {\{I_0(2\eta)-I_0(\eta)^2\}^{1/2}},
    \label{eq:random_phase_waveform}
\end{equation}
$I_0$ is the modified Bessel function of the first kind and $\eta>0$ is a waveform parameter. Because $T=\kappa c_i^\top\theta+2\pi\phi_i$ is uniform in $\kappa c_i^\top\theta+[0,2\pi)$, Bessel identities give  $E\{w_\eta(T)\}=0$ and $\Var\{w_\eta(T)\}=1$ independent of $\eta, \kappa$ and $\theta$.
The waveform $w_\eta$ is a smooth function of $\phi$ which is \(2\pi\)-periodic, mean zero
and variance one under the chosen uniform phase.  Conditional on a fixed phase vector
\(\phi\), the map \(\theta\mapsto Y(\theta,\phi)\) is one simulated path.

For a fixed phase, the empirical quadratic loss \(n^{-1}\sum_{j=1}^n\|W\{Y(\theta,\phi)-y_j\}\|^2/(2d_y)\) differs from the loss below by \(C_n=\sum_{j=1}^n\|W(y_j-\overline y_n)\|^2/(2nd_y)\).  This term is independent of \(\theta\) and \(\phi\) and is therefore omitted.  Here $\phi$ plays the role of $\xi_{k,b}$ in \eqref{eq:setup_block_estimators}. Let
\begin{equation}
    L(\theta,\phi)
    =
    \frac{1}{2d_y}
    \left\|W\{Y(\theta,\phi)-\overline y_n\}\right\|^2,
    \qquad W\in\mathbb R^{d_y\times d_y},
\end{equation}
and $\ell(\theta)=\E\{L(\theta,\phi)\}$. With the \(d_y\times d_y\) cyclic shift matrix, \(P\) and \(\rho=0.003\), the weight matrix is 
\begin{equation} \label{eq:W}
W=\{I_{d_y}+\rho(P+P^\top)/2\}/\sqrt{1+\rho^2/2}     \,.
\end{equation}
This normalization gives each row of \(W\) unit Euclidean norm. Let $\mathsf A$ and $\mathsf C$ have rows $a_i^\top$ and $c_i^\top$,
respectively.   Writing
$
  r_i(\theta,\phi)
  =
  w_\eta\{\kappa c_i^\top\theta+2\pi\phi_i\},
$
we have $Y=\mathsf A\theta+\epsilon r$ with $\E_\phi r=0$ and $\E_\phi rr^\top=I_{d_y}$ so averaging over $\phi$ exactly gives
\begin{equation} 
  \ell(\theta)
  =
  \frac1{2d_y}\|W(\mathsf A\theta-\overline y_n)\|^2
  +
  \frac{\epsilon^2}{2d_y}{\rm tr}(W^\top W).
\end{equation}

\subsubsection{Random-Phase Zig--Zag Implementation and Validation}
\label{app:random_phase_validation}

The stochastic-gradient Zig--Zag
sampler follows the construction of \citet{frazier2025exact}.  The phase
distribution is independent of \(\theta\), and the waveform and its derivative
are bounded, so
\[
    \E_\phi\{\nabla_\theta L(\theta,\phi)\}
    =
    \nabla_\theta\ell(\theta).
\]
A fresh phase vector is generated at every candidate event.

To quantify the cost induced in the switching rate, let
\(v=(v_1,v_2)\in\{-1,1\}^2\) have independent Rademacher coordinates.
For an exact-gradient Zig--Zag process, the total switching
intensity at $(\theta,v)$ is
$
  \sum_{j=1}^2
  [\beta_n v_j\partial_j\ell(\theta)]_+.
$
With an unbiased stochastic gradient, the positive-part operation
is applied before averaging over the auxiliary variable.  Convexity
gives
\[
  \E_\phi
  [\beta_n v_j\partial_j L(\theta,\phi)]_+
  \geq
  [\beta_n v_j\partial_j\ell(\theta)]_+.
\]
We define $\mathcal R_{\rm ZZ}$ as the ratio of these total
switching intensities averaged over the stationary target and
velocities,
\begin{equation}
    \mathcal R_{\rm ZZ}(\kappa)
    =
    \frac{
    \E_{\theta,v,\phi}\sum_{j=1}^2
    [\beta_n v_j\,\partial_jL(\theta,\phi;\kappa)]_+
    }{
    \E_{\theta,v}\sum_{j=1}^2
    [\beta_n v_j\,\partial_j\ell(\theta)]_+
    },
    \qquad \theta\sim\pi_n.
    \label{eq:random_phase_event_rate_ratio}
\end{equation}
Thus $\mathcal R_{\rm ZZ}\geq1$ measures the
switching-rate inflation caused by stochastic-gradient variability.
A larger value means more candidate-gradient evaluations, per unit of continuous
trajectory.

Let
\[
    M_0=\sup_t|w_\eta(t)|,\qquad M_1=\sup_t|w_\eta'(t)|.
\]
For coordinate \(j\), set
\[
    D_j=\|\mathsf A_{\cdot j}\|+
    \epsilon\kappa M_1\|\mathsf C_{\cdot j}\|.
\]
Along a Zig--Zag segment \(\theta(t)=\theta+vt\),
\[
    \|Y(\theta(t),\phi)-y\|
    \le
    \|\mathsf A\|_2\|\theta\|
    +\|y\|
    +\epsilon M_0\sqrt{d_y}
    +t\|\mathsf A\|_2\|v\|.
\]
Consequently,
\begin{equation}
    \left|\partial_jL(\theta(t),\phi)\right|
    \le
    \frac{D_j\|W^\top W\|_2}{d_y}
    \left(
      \|\mathsf A\|_2\|\theta\|
      +\|y\|
      +\epsilon M_0\sqrt{d_y}
      +t\|\mathsf A\|_2\|v\|
    \right).
    \label{eq:random_phase_affine_bound}
\end{equation}
After multiplication by \(\beta_n\), the right-hand side is an affine
candidate-event bound.  It is independent of the fresh phase and valid for all
future times on the segment, so its integrated rate and inverse are available
analytically.  No empirical quantile, clipping, or post hoc safety factor is
used.

The RQ-SCBD configuration used at all three \(\kappa\)-values is selected at \(\kappa=60\).  The first
five target-audited candidates are
\[
\begin{array}{c|ccccc}
(K,B,M) &(5,16,80)&(6,16,96)&(7,16,112)&(8,16,128)&(9,16,144)\\
\hline
\widehat{\rm TV}_{\rm disc}
&0.01340&0.01225&0.00529&0.00247&0.00449\\
{\rm MCSE}
&0.00334&0.00157&0.00303&0.00216&0.00085
\end{array}
\]
The first four decisions remain too uncertain under the pre-specified
Algorithm~\ref{alg:finite_block_audit} rule; the fifth is selected.  The
learning-scale selections plotted in
Figure~\ref{fig:random_phase_sign_scaling} are
\[
\begin{array}{c|rrrrrrrr}
\beta_n &5000&7500&10000&15000&20000&30000&40000&50000\\
\hline
M&80&112&112&112&160&160&288&256.
\end{array}
\]
These learning-scale selections use the separate lower-resolution audit
described in Appendix~\ref{app:experiment_specific_settings}
under ``Random phase''.  In particular, the \(\beta_n=20000\) choice
\((K,B,M)=(5,32,160)\) is not a repetition of the production audit that
selected \((9,16,144)\).

The pathwise derivative was checked against finite differences at
posterior-scale states; the maximum absolute discrepancy was
\(7.4\times10^{-11}\).  The maximum standardized discrepancy in the
unbiased-gradient check was \(1.23\).  Across fifteen production Zig--Zag
chains and six million candidate events, there were no bound violations and no
clipped probabilities.  The maximum observed ratio of the stochastic rate to
the bound was \(0.418\).  An independent Gaussian-target validator gave a
maximum posterior-mean error of \(0.0078\) target standard deviations and a
relative covariance error of \(0.0176\).

Each production chain has a declared cap of \(400000\) stochastic evaluations
per regime; RQ-SCBD realizes \(399888\) evaluations (2777 whole proposals),
whereas Zig--Zag realizes \(400000\).
Five independent seeds are used for every method--frequency pair.  RQ-SCBD
records the retained sign and uses Eq.~\eqref{eq:signed_ratio_estimator};
the pooled sign denominator is one in all three regimes.  Wall-clock
measurements are implementation-specific and are included only as a secondary
diagnostic.

\subsubsection{Additional variance-scaling and block-selection details}\label{sec:random_phase2} 

We give details for the block configuration selected at \(\kappa=60\). The first
four Algorithm~\ref{alg:adaptive_tuning}-feasible candidates have total budgets
\(80,96,112,\) and \(128\), but none gives a sufficiently precise
Algorithm~\ref{alg:finite_block_audit} decision under the pre-specified
uncertainty rule.  The first accepted candidate is
\[
    K=9,\qquad B=16,\qquad M=144,
\]
with posterior-weighted TV estimate \(0.00449\), Monte Carlo standard error
\(0.00085\), and upper endpoint \(0.00668\) of the 95\% six-batch
Student-\(t\) interval.  The same split is used in
all three frequency regimes.

At \(\kappa=60\), the local finite-range variance exponents are
\[
    \widehat\alpha_{\rm MC}=0.98,
    \qquad
    \widehat\alpha_{\rm RQMC}=3.01.
\]
Figure~\ref{fig:random_phase_sign_scaling} reports the block-variance fit, the ordered
sign and target checks at \(\beta_n=20000\), and the target-audited RQ-SCBD
budgets over a grid of learning scales.  The MC curve in the final panel is a
sign-only planning calculation.
\begin{figure}[htbp]
    \centering
    \includegraphics[width=\textwidth]{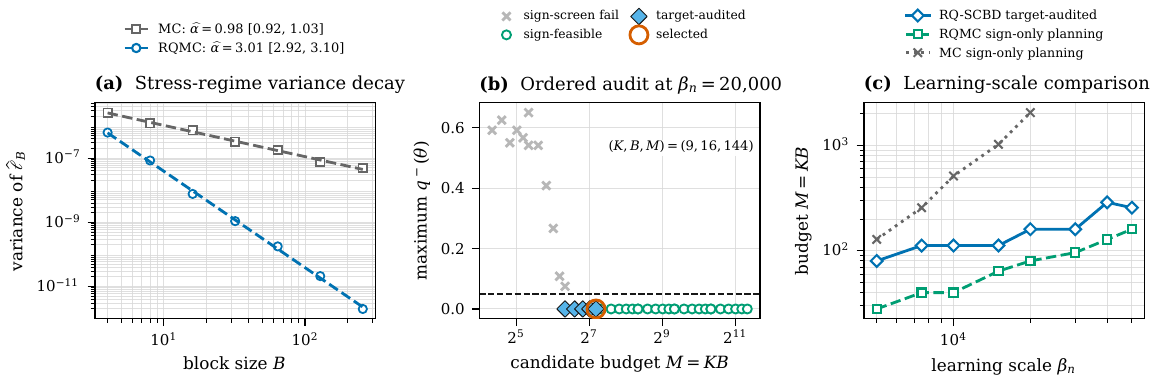}
    \caption{Variance and budget selection for the stress-regime at \(\kappa=60\).  Panel (a)
    gives the MC and scrambled-Sobol block-variance fits.  Panel (b) shows the
    fixed-state sign screen and the subsequent finite-block target audit; the
    selected split is \(K=9,B=16\).  Panel (c) distinguishes the target-audited
    RQ-SCBD budgets from sign-only planning curves as the learning scale
    increases.  Panel (c) uses a separate lower-resolution, beta-specific audit
    rather than repeating the production audit in panel (b); its
    \(\beta_n=20000\) budget therefore need not equal the production choice.}
    \label{fig:random_phase_sign_scaling}
\end{figure}

\begin{figure}[!ht] 
    \includegraphics[width=1\textwidth]{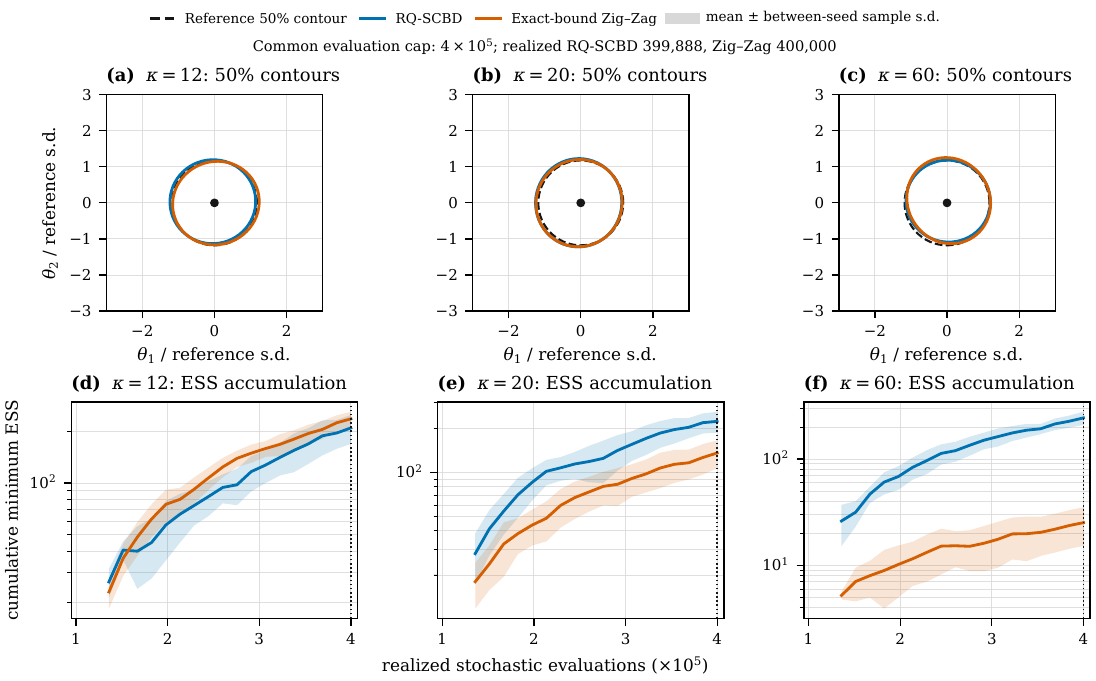}
    \caption{Posterior recovery in the random-phase comparison. Top: empirical Gaussian contours overlaid on the exact phase-averaged target. Bottom: across-chain mean cumulative minimum within-chain ordinary ESS, with sample-SD bands, versus stochastic evaluations, using proposal-indexed RQ-SCBD states and regular-time Zig--Zag samples. Accuracy remains stable across regimes, while ESS accumulation varies with pathwise-gradient frequency.}
    \label{fig:random_phase_posterior}
\end{figure}

\subsection{Additional checks and results for the \texorpdfstring{\(g\)-and-\(k\)}{g-and-k} example}

This subsection records the target specification and reference calculations for
the \(g\)-and-\(k\) example.  The empirical MMD target depends on the observed
data, the kernel bandwidth, and the learning-rate multiplier; the ABC
comparisons additionally require a tolerance pilot.  Table~\ref{tab:gandk_lambda_sensitivity}
shows how the finite-temperature MMD-GBI reference changes with the multiplier
\(\lambda\), while the remaining details describe the variance run, fixed
reference grid, and ABC tolerance construction used in
Section~\ref{sec:g_and_k_benchmark}.

The observed dataset is generated once, with seed \(42\), from the \(g\)-and-\(k\) model with \(A_{\rm GK}=3\), \(B_{\rm GK}=1\), \(c_{\rm GK}=0.8\), \(\theta_0=(g_0,k_0)=(0.10,0.05)\), and \(\n=500\).  The prior is uniform on \([0,5]\times[0,2]\).  The Gaussian kernel bandwidth is
\[
    h=3.949,
\]
equal to four times the median non-zero pairwise absolute distance among the first \(300\) observations.  This broad bandwidth is part of the empirical MMD-GBI target specification and is held fixed for the reference posterior, MC-SCBD, RQ-SCBD, and all ABC tolerance runs.

\begin{table}[htbp]
    \centering
    \caption{Target-scale check for the \(g\)-and-\(k\) empirical MMD-GBI posterior.  The observed dataset, bandwidth, quadrature rule, and prior are held fixed while only the learning-rate constant in \(\beta_n=\lambda \n\) changes.  The table illustrates how \(\lambda\) changes the finite-temperature target; it is not an algorithmic tuning comparison.}
    \label{tab:gandk_lambda_sensitivity}
    \begin{tabular}{ccccc}
        \toprule
        \(\lambda\) & posterior mean \(g\) & posterior mean \(k\) & posterior sd \(g\) & posterior sd \(k\) \\
        \midrule
        2  & 0.265 & 0.124 & 0.207 & 0.090 \\
        5  & 0.164 & 0.085 & 0.122 & 0.061 \\
        10 & 0.120 & 0.064 & 0.086 & 0.046 \\
        \bottomrule
    \end{tabular}
\end{table}

As expected, larger \(\lambda\) concentrates the Gibbs posterior closer to the
empirical MMD minimizer \(\widehat\theta_{\rm MMD}=(0.044,0.012)\).  The
main experiment uses \(\lambda=5\) as a moderate, pre-specified
finite-temperature target and then conditions all algorithmic comparisons on
that target.

The variance-decay calculation used in
Figure~\ref{fig:gandk_rqmc_diagnostics} is run at the frozen pilot-grid mode with
\(64\) independent randomizations at each displayed block size. The
finite-range $\hat\alpha$-fits are \(1.03\) for MC and \(2.26\) for RQMC.  The wider pilot
cloud used by Algorithm~\ref{alg:adaptive_tuning} is represented by nine
posterior-weighted medoids.  The same states, weighted by their posterior
cluster masses, are used for the reference-relative target audit.  The selected RQ-SCBD and MC-SCBD configurations have audit TV values \(0.00765\) and
\(0.00212\), respectively, below the default \(\epsilon_w=0.01\); the corresponding nominal 95\% interval upper endpoints are \(0.00905\) and \(0.00454\).
A separate RQ audit uses \(12800\) new repetitions at these same nine states
and Gauss--Hermite reference losses.  It gives \(0.00907\), with nominal
95\% interval \([0.00764,0.01050]\); the original independent holdout,
\(0.00911\) with upper endpoint \(0.01052\), is also inconclusive.
Neither check changes the historical selection or establishes a pass below
\(0.01\).  The MC holdout is \(0.00220\), with upper endpoint \(0.00436\).

ABC tolerances are computed before running the ABC chains from \(180\) pilot
synthetic datasets, each using \(S_{\rm ABC}=1792\) simulator draws.  This value
matches the simulator-call accounting \(S_{\rm ABC}=2K_{\rm RQ}B_{\rm RQ}\) for
the selected RQ-SCBD pair-integrand budget.  Half of the pilot parameters are
drawn from a broad rectangle, \(g\sim{\rm Unif}(0,0.90)\) and
\(k\sim{\rm Unif}(0,0.35)\); the other half are drawn from a Gaussian centered
at the frozen pilot-reference mean with covariance \(5\) times the
reference covariance, clipped to the prior support.  Table~\ref{tab:gandk_abc_tolerances}
records the resulting thresholds.  These quantiles define the ABC tolerance
posteriors used in Section~\ref{sec:g_and_k_benchmark}; they are not tuned to
match the empirical MMD-GBI posterior.

\begin{table}[htbp]
    \centering
    \caption{ABC tolerance thresholds used in the \(g\)-and-\(k\) benchmark.  Each threshold is a pilot quantile of the squared empirical MMD discrepancy between the observed and synthetic samples of size \(S_{\rm ABC}=1792\).}
    \label{tab:gandk_abc_tolerances}
    \begin{tabular}{cc}
        \toprule
        Pilot quantile & Tolerance \(\epsilon_q\) \\
        \midrule
        \(80\%\) & \(2.635\times 10^{-3}\) \\
        \(50\%\) & \(1.038\times 10^{-3}\) \\
        \(30\%\) & \(3.934\times 10^{-4}\) \\
        \(20\%\) & \(2.186\times 10^{-4}\) \\
        \(10\%\) & \(3.351\times 10^{-5}\) \\
        \(5\%\)  & \(1.356\times 10^{-5}\) \\
        \bottomrule
    \end{tabular}
\end{table}

The pilot reference was computed on a uniform \(54\times46\) grid,
\[
    g\in[0,0.66319],
    \qquad
    k\in[0,0.27056],
\]
using one scrambled Sobol evaluation with \(4096\) simulator pairs at each grid
point.  It remains fixed for all pilot states and ABC tolerances.
For reporting posterior moments and Figures~\ref{fig:gandk_target_geometry}
and~\ref{fig:gandk_posterior_comparison}, we use composite Gauss--Legendre
quadrature over the full prior and Gauss--Hermite quadrature for the loss.
The final orders are \(28\) per parameter interval and \(256\) for the latent
normal variable.  Refinement from \(20\) to \(28\) and from \(128\) to \(256\)
leaves the reported moments unchanged.  The original rectangle excludes
approximately \(0.72\%\) of the posterior mass.  This refinement changes
the reporting reference, not the selected budgets or production chains.
The empirical MMD minimizer \(\widehat\theta_{\rm MMD}\), the
finite-temperature posterior mean, and the ABC posterior means are therefore
different summaries of different targets.  The SCBD and Russian-roulette posterior means reported in
Table~\ref{tab:gandk_diagnostics} are sign-corrected.  The HPD boundaries and
marginal densities in Figure~\ref{fig:gandk_posterior_comparison} use
the corresponding signed density estimates.

\begin{figure}[ht!]
    \centering
    \includegraphics[width=\textwidth]
    {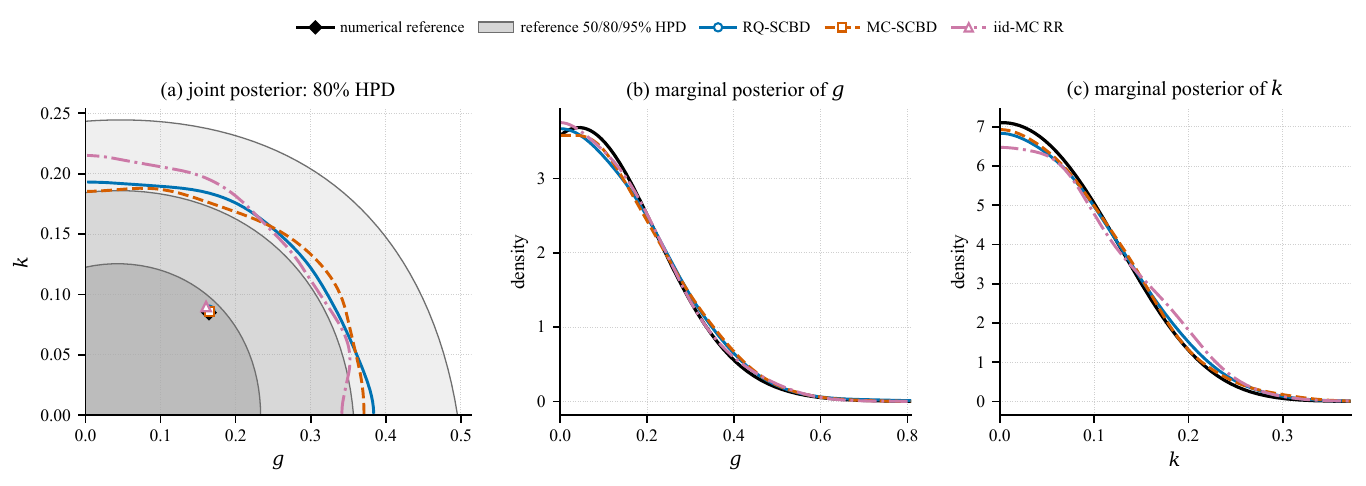}
    \caption{Posterior comparison for the \(g\)-and-\(k\) example.
    Panel (a) overlays sign-corrected 80\% highest posterior density (HPD) boundaries for
    RQ-SCBD, MC-SCBD, and iid-MC Russian roulette on the 50\%, 80\%,
    and 95\% HPD regions of the numerical reference.  Panels (b) and
    (c) show the corresponding marginal posterior densities of \(g\)
    and \(k\).  The black diamond marks the numerical-reference
    posterior mean.}
    \label{fig:gandk_posterior_comparison}
\end{figure}

\subsection{Russian-roulette comparator for the
\texorpdfstring{\(g\)-and-\(k\)}{g-and-k} example}
\label{app:gandk_rr}

The comparator uses the iid-MC Russian-roulette construction of
\citet{lyne2015russian}.  It does not use RQMC or an SCBD block split,
so its cost is reported as the expected number \(M'\) of MMD-pair
evaluations per weight evaluation within the prior support.

Let $L$ be a fixed prefix length used to define the roulette survival schedule: survival probabilities out to the $L$'th trial are
specified by the pilot construction, and beyond that
they decay geometrically with factor $1/2$. A separate preliminary calculation used two independent iid-MC pilot splits, each with \(32768\) MMD pairs per state, to compare \(L\in\{16,20,24,28,32,40\}\). It chose the smallest \(L\) for which the selected budgets were unchanged over all larger candidates and the expected cost and validation CV agreed with the \(L=40\) values to relative and absolute tolerances \(10^{-6}\), respectively. This gave \(L=24\). The operating point was then fixed before the production chains using
\(11\) states constructed from the frozen numerical-reference posterior geometry. Conditional on these state coordinates,
the selection used two further independent iid-MC pilot splits, each containing
\(32768\) MMD pairs per state; it did not use pointwise reference-loss values,
reference-posterior weights, or production output.

A weight evaluation uses \(20000\) iid MMD pairs to construct a local
centre and \(12500\) pairs for each Taylor factor.  The roulette
centre is
\[
    \widehat c(\theta)
    =
    \bar H_{20000}(\theta)
    +z_{0.999}\sqrt{\widehat V_{\rm upper}/20000},
    \qquad z_{0.999}=3.090232,
\]
where \(\widehat V_{\rm upper}\) is the prospective pilot upper bound on the
raw pair-integrand variance.  The roulette
survival distribution has a positive geometric tail after level \(L=24\),
with tail factor \(0.5\).  Its expected depth is
\[
    \E(N)=7.3120726439,
\]
and the expected cost is
\[
    M'
    =
    20000+12500\,\E(N)
    =
    111400.9080
\]
MMD pairs, or \(222801.8161\) simulator calls.  On the independent
validation split, the largest estimated coefficient of variation was
\(0.97276\), and the largest statewise one-sided Gaussian 95\% upper
bound on the negative-weight probability was \(0.02502\); these are
Gaussian-proxy tuning diagnostics.  A further check uses \(512\) actual
iid-MC roulette weights at each of the \(11\) states without retuning.
The largest one-sided binomial 95\% upper limit for the negative-weight
probability, with Bonferroni adjustment over the states, is \(0.0417\).
The sample coefficients of variation range from \(0.65\) to \(0.83\).

Five independent chains were run for \(1400\) iterations, with the
first \(350\) iterations discarded.  The \(5250\) retained draws have
acceptance \(0.3198\), negative-sign fraction \(0.00229\), and mean
sign \(0.99543\).  The sign-corrected posterior mean is
\[
    (0.160583,\ 0.090040),
\]
compared with numerical-reference mean
\[
    (0.164413,\ 0.085049).
\]
The standardized posterior-mean differences are \(-0.0314\) for \(g\)
and \(0.0817\) for \(k\), giving joint mean error \(0.0619\).  The
minimum ordinary coordinate ESS is \(0.0280\) per \(10^5\) nominal simulator
calls for retained iterations.

For the signed posterior means, an initial-monotone-sequence variance estimate
gives Monte Carlo standard errors \((0.00292,0.00146)\) for RQ-SCBD,
\((0.00638,0.00432)\) for MC-SCBD, and \((0.00685,0.00354)\) for RR.
Using the common reference variances to express precision as an equivalent
number of independent draws, and taking the minimum over coordinates, gives
\(3.233\), \(0.0203\), and \(0.0255\)
per \(10^5\) nominal retained calls.  Batch-means estimates with batch lengths
\(\lfloor\sqrt{T}\rfloor\) and \(2\lfloor\sqrt{T}\rfloor\), where \(T\)
is the retained length of each chain, instead give
\((3.804,0.0655,0.0348)\) and \((3.282,0.0439,0.0264)\).
Thus the MC-SCBD/RR precision ordering is sensitive to the variance estimate;
the RQ-SCBD advantage is consistent across these checks.

Across the \(4643\) weight evaluations within the prior support, the
realized mean cost was \(111288.0\) pairs and the realized mean
roulette depth was \(7.3030\).  These values are within \(0.13\%\) of
their prospective expectations.  Proposals outside the prior support
required no simulator calls.

\subsection{ABC tolerance sweep for the \texorpdfstring{\(g\)-and-\(k\)}{g-and-k} benchmark}
\label{app:gandk_abc_sweep}

The MMD-ABC comparison uses \(S_{\rm ABC}=1792\) simulator draws per
proposal, matching the simulator count of the selected RQ-SCBD
configuration.  For a tolerance \(\epsilon\), it targets
\[
    \pi_{\rm ABC,\epsilon,S}(\theta\mid y)
    \propto
    \pi_0(\theta)
    \mathbb P_\theta\!\left\{
      \widehat{\rm MMD}^{\,2}\le\epsilon
    \right\}.
\]
The tolerance therefore indexes a family of posteriors rather than an
approximation sequence for the fixed MMD-GBI posterior in
Eq.~\eqref{eq:gk_mmd_posterior}. Table~\ref{tab:gandk_abc_sweep} and Figure~\ref{fig:gandk_abc_tradeoff} report
the ABC tolerance sweep; see
Section~\ref{sec:g_and_k_benchmark}.  ABC targets tolerance posteriors rather
than the empirical MMD-GBI target, so the sweep is a likelihood-free benchmark
under matched simulator-call accounting; it should not be read as an estimator of the empirical MMD-GBI posterior. The smallest reference-scaled mean distance is $0.175$. At smaller tolerances the target changes and acceptance falls, so the increase in distance cannot be attributed to Monte Carlo error alone.

\begin{table}[htbp]
    \centering
    \scriptsize
    \setlength{\tabcolsep}{3pt}
    \begin{tabular}{@{}lrrr@{}}
        \toprule
        ABC tolerance & Acceptance & ESS/\(10^5\) sims & Ref.-scaled mean distance \\
        \midrule
        \(\epsilon_{80}\) & 0.268 & 5.37 & 1.308 \\
        \(\epsilon_{50}\) & 0.224 & 4.30 & 0.422 \\
        \(\epsilon_{30}\) & 0.221 & 4.93 & 0.175 \\
        \(\epsilon_{20}\) & 0.221 & 4.61 & 0.414 \\
        \(\epsilon_{10}\) & 0.142 & 1.60 & 0.801 \\
        \(\epsilon_{5}\) & 0.092 & 1.89 & 0.841 \\
        \bottomrule
    \end{tabular}
    \caption{ABC tolerance sweep for the \(g\)-and-\(k\) benchmark under the
same simulator-call budget used by RQ-SCBD\@.  The fourth column is the distance
from the numerical MMD-GBI posterior mean, scaled by its marginal posterior
standard deviations; it is not an estimation error for a common target.}
    \label{tab:gandk_abc_sweep}
\end{table}

\begin{figure}[htbp]
    \centering
    \includegraphics[width=\textwidth]{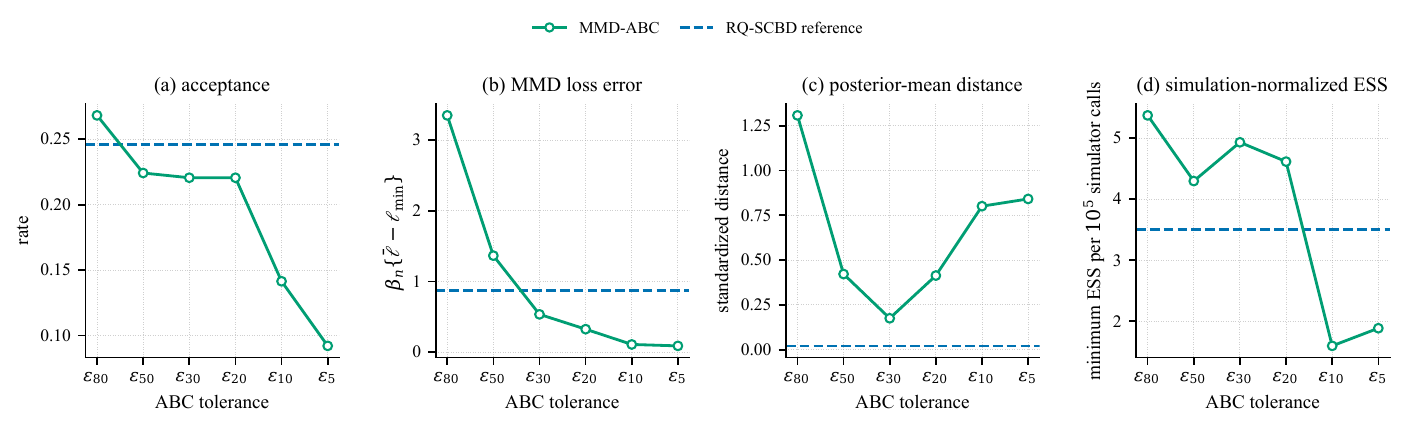}
    \caption{ABC tolerance and reference-distance trade-off under the same simulator and MMD discrepancy.  Tightening the tolerance changes the ABC posterior and does not target the finite-temperature empirical MMD-GBI posterior.  ESS uses ordinary coordinates and nominal simulator calls for retained iterations.}
    \label{fig:gandk_abc_tradeoff}
\end{figure}

\end{document}